\documentclass[11pt]{article} 
\pdfoutput=1
\usepackage[letterpaper, margin=1in]{geometry}
\usepackage{setspace}
\usepackage{graphicx} 
\usepackage{hyperref}
\usepackage{amsfonts}
\usepackage{amssymb}
\usepackage{amsmath}
\usepackage{authblk}

\usepackage[normalem]{ulem}

\date{}

\usepackage{float}
\usepackage{amsthm}
\usepackage{thmtools,thm-restate}
\usepackage{blkarray}
\usepackage{cleveref}
\usepackage{graphicx}
\crefformat{section}{\S#2#1#3} 
\crefformat{subsection}{\S#2#1#3}
\crefformat{subsubsection}{\S#2#1#3}

\usepackage{tikz}
\usepackage[symbol]{footmisc}
\usepackage{appendix}

\newcommand{\F}{\mathbb{F}}

\newcommand{\M}{\mathcal{M}}

\renewcommand{\angle}[1]{\mathopen{\langle} #1\mathclose{\rangle}}

 \newtheorem{claim}{Claim}[subsection]
\newtheorem{theorem}{Theorem}{}
\newtheorem{lemma}{Lemma}[subsection]
\newtheorem{remark}{Remark}[subsection]
\newtheorem{definition}{Definition}[subsection]
{}
\newtheorem{proposition}{Proposition}[subsection]
\newtheorem{observation}{Observation}[subsection]

\title{\bf Randomized Black-Box PIT for Small Depth +-Regular Non-commutative Circuits}

\begin{document}
\author{G V Sumukha Bharadwaj\footnote{\href{mailto:cs21d003@iittp.ac.in}{cs21d003@iittp.ac.in}} and S Raja\thanks{\href{mailto:raja@iittp.ac.in}{raja@iittp.ac.in}}}
\affil{IIT Tirupati, India}
\maketitle
\begin{abstract}
In this paper, we address the black-box polynomial identity testing (PIT) problem for non-commutative polynomials computed by $+$-regular circuits, a class of homogeneous circuits introduced by Arvind, Joglekar, Mukhopadhyay, and Raja (STOC 2017, Theory of Computing 2019). These circuits can compute polynomials with a number of monomials that are doubly exponential in the circuit size. They gave an efficient randomized PIT algorithm for +-regular circuits of depth 3 and posed the problem of developing an efficient black-box PIT for higher depths as an open problem.

We present a randomized black-box polynomial-time algorithm for +-regular circuits of any constant depth. Specifically, our algorithm runs in 
 $s^{O(d^2)}$ time, where $s$ and $d$ represent the size and the depth of the $+$-regular circuit, respectively. Our approach combines several key techniques in a novel way. We employ a nondeterministic substitution automaton that transforms the polynomial into a structured form and utilizes polynomial sparsification along with commutative transformations to maintain non-zeroness.  Additionally, we introduce matrix composition, coefficient modification via the automaton, and multi-entry outputs —methods that have not previously been applied in the context of black-box PIT. Together, these techniques enable us to effectively handle exponential degrees and doubly exponential sparsity in non-commutative settings, enabling polynomial identity testing for higher-depth circuits. Our work resolves an open problem from \cite{AJMR19}.
 
In particular, we show that if $f$ is a non-zero non-commutative polynomial in $n$ variables over the field $\F$, computed by a depth-$d$ $+$-regular  circuit of size $s$, then $f$ cannot be a polynomial identity for the matrix algebra $\mathbb{M}_{N}(\mathbb{F})$, where $N= s^{O(d^2)}$ and the size of the field $\F$ depending on the degree of  $f$.  Interestingly, the size of the matrices does not depend on the degree of $f$. Our result can be interpreted as an Amitsur-Levitzki-type result \cite{AL} for polynomials computed by small-depth $+$-regular circuits.

 \end{abstract}
 \newpage
 \tableofcontents
 \newpage
\section{Introduction} \label{Introduction}
The non-commutative polynomial ring, denoted by $\F\angle X$, over a field $\F$ in non-commuting variables
$X$, consists of non-commuting polynomials in $X$. These are just $\F$-linear combinations of words
(we call them monomials) over the alphabet $X= \{x_1,\ldots,x_n\}$.
Hyafil \cite{Hya77} and Nisan \cite{N91}, studied the complexity of
\emph{non-commutative} arithmetic computations, in particular the
complexity of computing the determinant polynomial with non-commutative
computations. 

Non-commutative arithmetic circuit families compute non-commutative polynomial families in a non-commutative polynomial ring $\F\angle X$, where multiplication is non-commutative (i.e., for distinct $x,y \in X$, $xy \neq yx$). We now recall the formal definition of non-commutative arithmetic circuits.

\begin{definition}[Non-commutative arithmetic circuits]
A \emph{non-commutative arithmetic circuit} $C$ over a field $\F$ is a
directed acyclic graph such that each in-degree $0$ node of the graph
is labelled with an element from $X\cup \F$, where
$X=\{x_1,x_2,\ldots,x_n\}$ is a set of noncommuting variables. Each
internal node has fan-in two and is labeled by either ($+$) or
($\times$) -- meaning a $+$ or $\times$ gate,
respectively. Furthermore, each $\times$ gate has a designated left
child and a designated right child. Each gate of the circuit
inductively computes a polynomial in $\F\angle{X}$: the polynomials
computed at the input nodes are the labels; the polynomial computed at
a $+$ gate (respectively $\times$ gate) is the sum (respectively product in
left-to-right order) of the polynomials computed at its children. The
circuit $C$ computes the polynomial at the designated output node.
An arithmetic circuit is a \emph{formula} if the fan-out of every gate is at most one.
\end{definition}

Designing an efficient deterministic algorithm for non-commutative polynomial identity testing is a major open problem. Let $f \in \F\angle X$ be a polynomial represented by a non-commutative arithmetic circuit $C$. In this work, we assume that the polynomial $f$ is given by a black-box access to $C$, meaning we can evaluate the polynomial $f$ on matrices with entries from $\F$ or an extension field. Note that the degree of an $n$-variate polynomial computed by the circuit $C$ of size $s$ can be as large as $2^s$  and the sparsity, i.e., the number of non-zero monomials, can be as large as $n^{2^s}$. For example, the non-commutative polynomial $(x+y)^{2^s}$ has degree $2^s$, doubly exponential sparsity $2^{2^s}$, and has a circuit of size $O(s)$.

The classical result of Amitsur-Levitzki \cite{AL} shows that a non-zero non-commutative polynomial $f$ of degree $2d-1$ does not vanish on the matrix algebra $\M_d(\F)$. Bogdanov and  Wee \cite{BW05} have given an efficient randomized PIT algorithm for non-commutative circuits computing polynomials of degree $d=poly(s,n)$. Their algorithm is based on the result of Amitsur-Levitzki \cite{AL}, which states the existence of matrix substitutions $M=(M_1, M_2,\ldots, M_n)$ such that the matrix $f(M_1, M_2,\ldots, M_n)$ is not the zero matrix, where the dimension of the matrices in $M$ depends linearly on the degree $d$ of the polynomial $f$. 

Since the degree of the polynomial computed by circuit C can be exponentially large in the size of the circuit, their approach will not work directly. Finding an efficient randomized PIT algorithm for general non-commutative circuits is a well-known open problem. It was highlighted at the workshop on algebraic complexity theory (WACT 2016) as one of the key problems to work on.

Recently, \cite{cAJMR17,AJMR19} gave an efficient randomized algorithm for the PIT problem
when the circuits are allowed to compute polynomials of exponential degree, but the sparsity could be
exponential in the size of the circuit. To handle doubly-exponential sparsity, they studied a class of homogeneous non-commutative circuits, that they call +-regular circuits, and gave an efficient deterministic white-box PIT algorithm. These circuits can compute non-commutative
polynomials with the number of monomials doubly exponential in the circuit size. For the black-box setting, they obtain an efficient randomized PIT algorithm only for depth-3 +-regular circuits. In particular, they show that if a non-zero non-commutative polynomial $f \in \F\angle{X}$ is computed by a depth-3 +-regular circuit of size $s$, then $f$ cannot be a polynomial identity for the matrix algebra $\mathbb{M}_{s}(\F)$ for a sufficiently large field $\F$.
Finding an efficient randomized PIT algorithm for higher depth +-regular circuits is listed as an interesting open problem. We resolve this problem for constant depth +-regular circuits. 
In particular, we show that if $f \in \F\angle{X}$ is a non-zero non-commutative polynomial computed by a depth-$d$ +-regular circuit of size $s$, then $f$ cannot be a polynomial identity for the matrix algebra $\mathbb{M}_{N}(\mathbb{F})$, with $N=s^{O(d^2)}$ and the size of the field $\F$ depends on the degree of polynomial $f$. This resolves an open problem given in \cite{AJMR19}. We note that we get a \emph{black-box} randomized polynomial time black-box PIT algorithm for constant depth $+$-regular circuits. 

\subsection*{Our results}

We consider the black-box PIT problem for \emph{+-regular circuits}. These are a natural subclass of homogeneous non-commutative circuits and these circuits can compute polynomials of exponential degree and a double-exponential number of monomials. Recall that a polynomial $f$ is homogeneous if all of its monomials have the same degree. The syntactic degree is inductively defined as follows: For a leaf node labeled by a variable, the syntactic degree is 1, and 0 if it is labeled by a constant. For a $+$ gate, its syntactic degree is the maximum of the syntactic degree of its children. For a $\times$ gate, its syntactic degree is the sum of the syntactic degree of its children. A circuit is called homogeneous if all gates in the circuit compute homogeneous polynomials. Now we recall the definition and some observations of $+$-regular circuits from \cite{AJMR19}.

\begin{definition}[$+$-regular circuits \cite{AJMR19}] \label{reg-def}
  A non-commutative circuit $C$, computing a polynomial in $\F\angle X$, where $X=\{x_1,x_2,\ldots,x_n\}$, is \emph{$+$-regular} if it
  satisfies the following properties:
\begin{enumerate}
\item The circuit is homogeneous. The $+$ gates are of unbounded fanin and $\times$ gates are of fanin 2.
\item The $+$ gates in the circuit are partitioned into layers (termed
  $+$-layers) such that if $g_1$ and $g_2$ are $+$ gates in the same $+$-layer
  then there is no directed path in the circuit between $g_1$ and
  $g_2$.
\item  All gates in a $+$-layer have the same syntactic degree. 
\item  The output gate is a $+$ gate.
\item Every input-to-output path in the circuit goes through a gate in each $+$-layer. 
\item Additionally, we allow scalar edge labels in the circuit. For
  example, suppose $g$ is a $+$ gate in $C$ whose inputs are gates
  $g_1,g_2,\ldots,g_t$ such that $\beta_i\in\F$ labels edge $(g_i,g),
  i\in[t]$. If polynomial $P_i$ is computed at gate $g_i, i\in
  [t]$, then $g$ computes the polynomial $\sum_{i=1}^t\beta_iP_i$.
\end{enumerate}
\end{definition}

The $+$-depth, denoted by $d^+$, refers to the number of $+$ layers in $C$. The $+$ layers in circuit C are numbered from the bottom upwards.
For $i\in[d]$, let $\mathcal{L}_i^+$ represent the $i$-th layer of addition ($+$) gates, and let $\mathcal{L}_i^\times$ represent the $i$-th layer of multiplication ($\times$) gates that are inputs to the addition gates in $\mathcal{L}_i^+$. It's important to note that all gates in $\mathcal{L}_i^\times$ and $\mathcal{L}_i^+$ have the same syntactic degree. 

The sub-circuit in $C$ between any two consecutive addition layers $\mathcal{L}_i^+$ and $\mathcal{L}_{i+1}^+$ consists of multiplication gates and is denoted by $\Pi^*$. The inputs of this sub-circuit come from layer $\mathcal{L}_i^+$. Let $\mathcal{L}_{i+1}^\times$ consist of all output gates of this sub-circuit, where $1\leq i\leq d-1$. Note that all the gates of $\mathcal{L}_{i+1}^\times$ are product gates. It is important to note that this sub-circuit depth, which is the number of gates in any input-to-output gate path in the sub-circuit, can be arbitrary and is only bounded by the size of the circuit $C$. 
The bottom-most $\times$-layer $\mathcal{L}_1^\times$ can be assumed without loss generality to be the input variables and gates in $\mathcal{L}_1^+$ compute homogeneous linear forms.

\begin{remark}
\label{top-bottom-plus}
   If the top layer is $\Pi^*$ (i.e., the output gate is a $\times$ gate), we can add an extra + gate at the top with the + gate having a single input (i.e., fan-in 1). This ensures that the top layer is a $\Sigma$ layer. If the bottom layer is $\Pi^*$, then for each input variable, we can add a sum gate having a single input. This will increase the circuit size by at most $n+1$, where $n$ is the number of input variables. This allows us to assume that in $+$-regular circuits, both the top and bottom layers are $\Sigma$ layers. This is done to simplify the analysis.
\end{remark}

The size of the $+$-regular circuit is the number of gates in the circuit.   
As noted earlier, the non-commutative polynomial $(x+y)^{2^s}$ can be computed by a depth-3 $+$-regular circuit, denoted by $\Sigma\Pi^*\Sigma$,  of size $O(s)$ using repeated squaring.  This circuit consists of two addition layers, namely $\mathcal{L}_1^+,\mathcal{L}_2^+$ and two multiplication layers, namely $\mathcal{L}_1^\times,\mathcal{L}_2^\times$.  The multiplication layer $\mathcal{L}_1^\times$ consists of only the two input gates labeled by $x$ and $y$ respectively. The addition layer $\mathcal{L}_1^+$ consists of only one addition gate computing the homogeneous linear form $(x+y)$. The multiplication layer $\mathcal{L}_2^\times$ consists of only one gate computing the polynomial $(x+y)^{2^s}$. The addition layer $\mathcal{L}_2^+$ consists of only the ouptut gate computing the polynomial $(x+y)^{2^s}$. In this example, the top-most addition gate (i.e., the output gate) essentially has one input.

As noted in \cite{AJMR19}, the computational power of +-regular circuits is limited. Nisan's rank-based argument \cite{nisan91} can be adapted to show that the non-commutative permanent cannot be computed by polynomial-size +-regular circuits. However, despite their structural constraints, polynomial-size +-regular circuits with +-depth 2 can compute polynomials with exponential degree and a doubly exponential number of monomials. 

The main result of the paper is the following theorem.
\begin{theorem}
 Let $f$ be a non-commutative polynomial of degree $D$ over $X=\{x_1,\ldots,x_n\}$ computed by a +-regular circuit of depth $d$ and size $s$. Then $f\not\equiv 0$  if and only if $f$ is not identically zero on  the matrix algebra $\mathbb{M}_{N}(\mathbb{F})$, with $N=s^{O(d^2)}$ and $\F$ is sufficiently large.
\end{theorem}
 
For degree $D$ non-zero non-commutative polynomial $f$, the classical Amitsur-Levitzki \cite{AL} theorem guarantees  that  $f$  does not vanish on the matrix algebra $\M_{\frac{D}{2}+1}(\F)$. If $D=2^{\Omega(s)}$, this gives us an exponential time randomized PIT algorithm \cite{BW05}, where $s$ is the size of the circuit computing $f$. In contrast, in our result, the dimension of the matrices is independent of the degree of the polynomial.
If the sparsity of the polynomial, i.e., the number of non-zero monomials, is doubly exponential, then the main result of \cite{AJMR19} gives only an exponential time randomized PIT algorithm as their matrix dimension depends on the logarithm of the sparsity.  

This above theorem demonstrates that if the polynomial $f$ is computed by a $+$-regular circuit of size $s$ and depth $o(\sqrt{s}/\log s)$, we can determine if $f$ is identically zero or not using a $2^{o(s)}$ time randomized PIT algorithm, which is exponentially faster than the existing methods. In particular, if depth is $O(1)$ then 
our algorithm runs in polynomial time. It is important to note that the number of product gates (within each $\Pi^*$ layers) in any input-to-output path can be arbitrary and is only bounded by the circuit size $s$.

We note that \cite{AJMR19} presented a white-box deterministic polynomial-time PIT for arbitrary depth $+$-regular circuits. For the small-degree case, \cite{raz05PIT} provided a white-box deterministic polynomial-time PIT for non-commutative ABPs, while \cite{FS13, AGKS15} have shown a deterministic quasi-polynomial-time black-box PIT algorithm for non-commutative ABPs. In the commutative setting, a randomized polynomial time PIT was given by \cite{IM83}.

\subsection{Outline of the Proofs: High-level Idea}

In the rest of the paper, we will use "n.c." as an abbreviation for "non-commutative." We begin by explaining the main ideas behind the randomized PIT algorithm for depth-5 \( + \)-regular circuits, as this represents the primary bottleneck. The existing techniques are not directly applicable to this case, and solving the depth-5 case requires several new ideas, which we outline below. It turns out that this is the key bottleneck, as we can later adapt the results from the depth-5 \( + \)-regular circuits to higher depths with additional techniques.

Consider an n.c. polynomial $f$ over $X=\{x_1,\ldots,x_n\}$ computed by a $\Sigma\Pi^*\Sigma\Pi^*\Sigma$ circuit of size $s$. 
The polynomial $f$, computed by a $\Sigma\Pi^*\Sigma\Pi^*\Sigma$ circuit of size $s$, can be written as follows:
\begin{center}
 \begin{equation}
 \label{eq1}
 f=\sum_{i \in [s]}\prod_{j \in [D_2]} Q_{ij}.
\end{equation}
\end{center}
Here, the degree of each $Q_{ij}$, $i\in[s],j\in[D_2]$, is denoted by $D_1$ and can be computed by a $\Sigma\Pi^*\Sigma$ circuit of size at most $s$. 
As the size of the circuit is $s$, the output gate's fan-in is bounded by $s$.
The syntactic degree $D$ of the circuit can be expressed as $D=D_1\times D_2$. In general, both $D_1$ and $D_2$ can be exponential in $s$.  Note that each $Q_{ij}$ is a polynomial computed at layer $\mathcal{L}_2^+$. 
One natural idea is that since each $Q_{ij}$ can be computed by a $\Sigma\Pi^*\Sigma$ circuit, we can try to use the known result of depth-3 $+$-regular circuits \cite{AJMR19} and convert the given polynomial $f$  into a commutative polynomial, and then perform the randomized PIT  using the Polynomial Identity Lemma for commutative polynomials (also known as the DeMillo-Lipton-Schwartz-Zippel Lemma \cite{DL78,Zippel79,Sch80}).

Recall that in \cite{AJMR19}, the given polynomial $f$ is computed by a depth-3 +-regular circuit of size $s$. That is, $f$ is a sum of products of homogeneous linear forms. Formally, $f=\sum\limits_{i=1}^s P_i$, where for all $i \in [s], P_i=L_{i,1} \cdots L_{i,D}$ and $D$ could be exponential in $s$. They show that there exists an index set $I \subseteq[D]$ of size at most $s-1$ such that by considering only those linear forms positions indexed by $I$ as n.c.  and the remaining as commutative, the non-zeroness of $f$ is preserved. This fact is crucially used in their black-box identity-testing algorithm for depth-3 +-regular circuits.
In our depth-5 setting, that is, when $f=\sum\limits_{i \in [s]}\prod\limits_{j \in [D_2]} Q_{ij}$, it is only natural to wonder if there exists such a small index.
Note that the number of $Q_{ij}$ polynomials in each product, denoted by $D_2$, can be exponential, in general.
It is generally impossible to have a small index set of polynomial size. 
This is because if the index set is only polynomial in size, then as a result, no variables in some of the $Q_{ij}$ are considered non-commutative. Thus, these $Q_{ij}$ are considered commutative, possibly resulting in the commutative polynomial becoming 0. We cannot consider a n.c. polynomial as commutative and still maintain non-zeroness, in general.

To resolve this problem, we convert the n.c. polynomial into a commutative polynomial in several steps, utilizing the fact that a $+$-regular circuit computes it.

\subsubsection{Step 1: Transforming the Polynomial for Improved Structure}
In this step, we show that the polynomial can be converted into a more structured polynomial at the cost of introducing some spurious monomials.
We show that in each $Q_{ij}$ there is a small index set 
such that by considering only those homogeneous linear forms appearing in that $Q_{ij}$ as n.c.  and the remaining as commutative, the non-zeroness of $f$ is preserved.  The index sets are encoded using n.c. variables. This results in an exponential-sized index set. This is because the number of terms in the product, denoted by $D_2$, can be exponential in general and each term in the product has a small index set.

 This fact is formalized in Lemma \ref{gen-proj}. However, this fact alone will not be sufficient to design an identity-testing algorithm for depth-5 +-regular circuits.  This is because, to choose a small index set from each $Q_{ij}$, one would have to know the \emph{boundary} between $Q_{ij}$ and $Q_{i,j+1}$ (i.e., the position at which $Q_{ij}$ ends and $Q_{i,j+1}$ begins), for any $i  \in [s], j \in [D_2-1]$. 

 \begin{definition}[Boundary of a Polynomial]
 \label{boundary}
 Let $P=\prod_{i \in [k]}P_i$ be a non-commutative polynomial over $X$. Each non-zero monomial $m$ of $P$ can be expressed as: $m=m_1m_2\cdots m_k$. We say the boundary of $P_i$, $i \in [k]$ as the position at which $P_{i}$ ends. Similarly, the boundary of $m_i$, $i \in [k]$ in the monomial $m$ is defined as the position at which $m_{i}$ ends.
 \end{definition}

Since non-commutative polynomials are $\F$-linear combinations of words/strings (called monomials) over $X$,  for a n.c. monomial $m$, we can identify the variable at position $e$ in $m$, where $1 \leq e\leq |m|$.
 
 To know where $Q_{ij}$ ends and $Q_{i,j+1}$ begins exactly, the algorithm will have to keep track of the degree (length) of $Q_{ij}$. This is not feasible as the degree of $Q_{ij}$ could be exponential in $s$. This is one of the main challenges in designing a black-box identity-testing algorithm for depth-5 +-regular circuits and is one of the main differences between depth-3 and depth-5 +-regular circuits (small depth $+$-regular circuits, in general).

To address this issue, we first convert each $Q_{ij}$ into a more \emph{structured n.c.  polynomial}, which we refer to as a \emph{$k$-ordered power-sum polynomial} (see Definition \ref{powersum-def}), where $k$ is the number of variables. This new n.c.  polynomial has the property that we can consider it as a commutative polynomial preserving non-zeroness (Claim \ref{ord-pow-sum-c-nc}). 
Note that when the variables in any n.c. linear form is considered commutative, then non-zeroness is preserved. However, this may not be true for n.c. polynomials, in general. For example, the polynomial $(xy-yx)$ is 0 if $x$ and $y$ are commuting variables. 
This property (Claim \ref{ord-pow-sum-c-nc}) alone will not let us consider the whole polynomial as commutative. This is because
we are dealing with the sum of the product of $Q_{ij}$ polynomials and the number of terms in the product in general could be exponential. 
Thus, we can not use a fresh set of variables for each \emph{$k$-ordered power-sum polynomial} as this will have exponentially many \emph{new} variables, in general. Because of this, we use the \emph{same $k$ variables} to convert each $Q_{ij}$ in the product to a \emph{$k$-ordered power-sum polynomial}, instead of using a fresh set of $k$ variables for each $Q_{ij}$. 
If we simply consider the resulting product as commutative, then variables that belong to $k$-ordered power-sum polynomials of different $Q_{i,j_1}$ and $Q_{i,j_2}$ in the product, could mix (i.e., exponents of the same variable appearing in different $k$-ordered power-sum polynomials gets added). As a result, we cannot guarantee that the resulting commutative polynomial will preserve non-zeroness.

To address this issue, we will not be considering the modified polynomial as commutative at this stage. Instead, we will transform each $Q_{ij}$ polynomial into a more structured n.c. polynomial, which we denote by $\hat{Q}_{ij}$, by introducing a \emph{new} set of n.c.  variables. We show that this transformation preserves non-zeroness. 

The substitution automaton we use for the first step generates some spurious monomials along with a structured polynomial (Claim \ref{f-hat-spurious}).
The spurious monomials are produced because we cannot definitively identify the boundaries of each $Q_{ij}$ polynomial. We can consider spurious monomials as noise generated by our approach.

Let $F_1$ be the sum of all spurious monomials produced and $\hat{f}_1$ be the structured polynomial resulting from this transformation.
We will prove that $\hat{f}_1+F_1\not\equiv0$ (see Lemma \ref{gen-proj} and Claim \ref{fhat-spurious-non-zero}). The key to this proof lies in demonstrating that the spurious monomials possess a distinctive property, allowing us to differentiate them from the structured part.

 After completing Step 1, we can transform the polynomial computed by the depth-5 circuit into a combination of a structured part and a spurious part. One of the outcomes of this first step is that we can efficiently identify the boundaries of the $\hat{Q}_{ij}$ polynomials in the structured part, particularly using a small automaton, even though this process introduces some spurious monomials. A similar concept of boundary also applies to the monomials in the spurious part, and we show that these boundaries can also be identified using a small automaton.
 Another outcome is that each $\hat{Q}_{ij}$ polynomial can be treated as commutative without resulting in a zero polynomial. These two outcomes are crucial for the remaining steps of our method. 

\subsubsection{Step 2: Product Sparsification}
In the second step, we demonstrate that within the polynomial $\hat{f}_1$ if we treat a small number of the  $\hat{Q}_{ij}$ polynomials as n.c. while considering the rest as commutative, non-zeroness is preserved. A key aspect of this step is the ability to treat each $\hat{Q}_{ij}$ polynomial as commutative.
Although treating n.c. polynomials as commutative while preserving non-zeroness is generally not feasible, in Step 1, we transformed $f$ into a sum of a structured part $\hat{f}_1$ and a spurious part $F_1$. The n.c. polynomials $\hat{Q}_{ij}$ within $\hat{f}_1$ can be treated as n.c.  without resulting in a zero polynomial (as established in Claim \ref{ord-pow-sum-c-nc}). We then show that the n.c. polynomial obtained after this transformation remains non-zero (see Lemma \ref{poly-proj}), a result we refer to as \emph{product sparsification}. 
We call this sparsification because we reduce the number of n.c. $\hat{Q}_{ij}$ polynomials in each product, which can be exponential, to a small number of them while preserving non-zeroness. It is important to note, however, that the sum of the exponents of the n.c. variables (a.k.a \emph{n.c. degree}) in this new polynomial can still be generally exponential. 

It's important to highlight that this product sparsification step affects both the structured part $\hat{f}_1$ and the spurious part $F_1$ of the polynomial obtained after Step 1.  
 Let $\hat{f}_2$ and $F_2$ represent the polynomials derived from $\hat{f}_1$ and $F_1$ respectively as a result of this product sparsification step.
We first consider the product sparsification of the n.c. polynomial $\hat{f}_1$. 
We will return to  $\hat{f}_2+F_2$ later.

\subsubsection{Step 3: Commutative Transformation}
In the third and final step, we show that the sum of products of a small number of structured $\hat{Q}_{ij}$ n.c. polynomials can be transformed into a commutative polynomial while preserving non-zeroness (see Lemma \ref{com-conversion}).
As mentioned earlier, the n.c. degree of the polynomial obtained after Step 2 can be exponential in general. However, we show that this exponential degree n.c. polynomial can be converted into a commutative polynomial using only a small number of \emph{new} commutative variables.  
It is noteworthy that there is no known method to convert a general exponential degree n.c. polynomial into a commutative polynomial with just a small number of commutative variables while preserving non-zeroness. 

Such transformations are known only when the number of non-zero monomials is bounded single-exponentially (a.k.a sparsity of the polynomial)  or when the polynomial is computed by a $\Sigma\Pi^*\Sigma$ circuit \cite{AJMR19}. The key to our transformation lies in the structured nature of the $\hat{Q}_{ij}$ polynomials, each of which requires only a small number of \emph{new commutative variables}. Since the number of n.c $\hat{Q}_{ij}$ polynomials is small,  we conclude that we have to introduce only a small number of \emph{new commutative variables} for this transformation while ensuring non-zeroness is maintained.

It's important to highlight that this commutative transformation step also affects both the structured part $\hat{f}_2$ and the spurious part $F_2$ of the polynomial derived after Step 2.  Let $\hat{f}^{(c)}_3$ and $F^{(c)}_3$ represent the polynomials derived from $\hat{f}_2$ and $F_2$ respectively as a result of this commutative transformation step.
We are ready to consider the sum of the structured part and spurious part now.

\subsubsection{Efficient Coefficient Modifications}
Both Steps (2) and (3) will be applied to the structured polynomial $\hat{f}_1$ and the spurious polynomial $F_1$ obtained after Step 1. Let $\hat{f}^{(c)}_3$ and $F^{(c)}_3$ denote the respective commutative polynomials obtained after applying these two steps to 
$\hat{f}_1$ and $F_1$. Since the same commutative variables are used, it is possible for monomials in $\hat{f}^{(c)}$ and $F^{(c)}$ to cancel each other.

If $\hat{f}^{(c)}_3+F^{(c)}_3\neq 0$,  we have successfully transformed the exponential degree n.c. polynomial into a commutative polynomial using only a small number of commutative variables while preserving non-zeroness. We can then evaluate the resulting commutative polynomial $\hat{f}^{(c)}_3+F^{(c)}_3$ by randomly chosen matrices for non-zeroness. 

On the other hand, if $\hat{f}^{(c)}_3+F^{(c)}_3=0$, we know that the transformations in Steps (2) and (3) carried out only on the structured part $\hat{f}_1$ obtained from Step (1)
ensure that $\hat{f}_3^{(c)}\neq 0$, which implies that $F_3^{(c)}\neq 0$. 
If a group of n.c. monomials in $\hat{f}_1$ transform into a commutative monomial 
$m$ with coefficient $\alpha_m$ in $\hat{f}_3^{(c)}$ after Steps (2) and (3) then 
another group of n.c. monomials in $F$ transformed into the \emph{same} commutative monomial $m$ with the coefficient $-\alpha_m$ in $F_3^{(c)}$.

To address this cancellation issue, we show that we can carefully modify the coefficients of at least one such group of n.c. monomials in the polynomial $\hat{f}_1+F_1$ obtained after Step 1  before executing Steps (2) and (3). The key to this coefficient modification lies in the fact that the spurious monomials introduced in Step (1) have a distinctive property, enabling us to differentiate them from the structured part using a small automaton.
We show that such a coefficient modification preserves non-zeroness (see Lemma \ref{lem-modp-counting}).
 We then establish that if we apply Steps (2) and (3) on this newly modified polynomial, we get a commutative polynomial while preserving non-zeroness.
 
Though our transformation of a non-commutative polynomial into a commutative one involves modifying the monomial coefficients using substitution matrices obtained from an automaton, it does not turn a zero polynomial into a non-zero one. Since we apply the same matrix substitution for each occurrence of a given variable, monomials that cancel before the transformation will continue to cancel afterward, as they are affected in the same way. This property is crucially used in applying our method inductively for higher depths.

\subsubsection{Matrix Compositions}

Each of these steps (Steps (1)-(3) and coefficient modification step) is performed using a small substitution automaton, which defines a substitution matrix for each n.c. variable. Throughout the process, we obtain four different sets of matrices.
It’s important to note that the matrices used in each step evaluate a n.c. polynomial obtained from the previous step.

As our model is black-box, it is not possible to evaluate the polynomial in this way. We need a single matrix substitution for each n.c. variable.
We show that the substitution matrices used across these four steps can be combined into a single matrix for each n.c. variable (see Lemma \ref{composing-substitution-mat}).

This lemma on matrix composition enables us to establish the existence of small substitution matrices for testing non-zeroness in a sequence of steps, which is a key novelty and contribution of this work.
This enables us to develop an efficient randomized polynomial identity testing (PIT) algorithm for depth-5 $+$-regular circuits (see  Theorem \ref{thm-depth-5}).
We further extend this idea to higher depths to develop an efficient randomized PIT algorithm (see Theorem \ref{main-thm}).
\\
Although \cite{AJMR19} provides a randomized polynomial-time black-box PIT algorithm for \( + \)-regular circuits of depth 3, existing techniques are not directly applicable to this case. In particular, depth-5 \( + \)-regular circuits require new ideas, which we introduce in our algorithm. While our algorithm operates in a black-box setting, it can be described as a sequence of transformations that convert the non-commutative polynomial into a commutative one, enabling identity testing. In Lemma \ref{composing-substitution-mat}, we show that the matrix substitutions can be composed into a single matrix for each input variable, capturing the entire transformation process. Our coefficient modification approach introduces a novel technique in polynomial identity testing, which may be useful for developing PIT algorithms for other circuit models.

\section{Preliminaries}
\subsection{Substitution Automaton}
The paper uses the standard definition of substitution automaton from automata theory using the terminology of \cite{AJR16,AJR18,AJMR19}.

\begin{definition}[Substitution NFA]
 A finite \emph{nondeterministic substitution automaton}
is a finite nondeterministic automaton $\mathcal{A}$ along with a substitution map $$\psi:Q\times
X\to \mathcal{P}(Q\times Y\cup \F),$$ where $Q$ is the set of states of $\mathcal{A}$, and $Y$
is a set of variables and $\mathcal{P(S)}$ is the power set of $\mathcal{S}$. If $(j,u) \in \psi(i,x)$ it means
that when the automaton $\mathcal{A}$ in state $i$ reads variable $x$ it
can replace $x$ by $u\in Y\cup \F$ and can make a transition to state $j\in
Q$. In our construction, $Y$ consists of both commuting and noncommuting variables.

Now, for each $x\in X$ we can define the transition matrix $M'_x$ as follows:
\begin{equation}\label{trans-matrix}
M'_x(i,j)=  u, 1\le i,j\le |Q|, \textrm{ where }(j,u) \in \psi(i,x).
\end{equation}
\end{definition}

We observe that all the automatons we design are more restrictive in that, for any \( i \in Q \) and \( x \in X \), if both \( (j,u) \) and \( (j,v) \) belong to \( \psi(i,x) \), then \( u \) must equal \( v \). In other words, the substitution applied to variable \( x \) when transitioning from state \( i \) to state \( j \) is uniquely determined by the destination state \( j \).

The substitution map $\psi$ can be naturally extended to handle strings as follows: $\hat{\psi}:Q\times
X^*\to \mathcal{P}(Q\times (Y\cup \F)^*)$. For a state $j \in Q$ and string $m' \in (Y\cup \F)^*$,  if $(j,m') \in \hat{\psi}(i,x)$, then it means that the automaton starting at state $i$, on input string $m \in X^*$, can nondeterministically move to state $j$ by transforming the input string $m$ to $m'$ on some computation path.

Now, we explain how a substitution automaton $\mathcal{A}$ processes a given polynomial $f$. First, we will describe the output of a substitution automaton $\mathcal{A}$ on a monomial $w\in X^D$. Let $s$ and $t$ be the designated
initial and final states of $\mathcal{A}$, respectively. For each variable $x\in X$, we have a transition matrix $M'_x$ of size $|Q|\times |Q|$. If $w=x_{j_1}x_{j_2}\ldots x_{j_d}$, then the output of the substitution automaton $\mathcal{A}$ on the monomial $w$ is given by the $(s,t)$-th entry of the matrix $M'_w=M'_{x_{j_1}}M'_{x_{j_2}}\ldots M'_{x_{j_D}}$. Note that this entry can be a polynomial since there could be multiple paths from $s$ to $t$ labeled by the same monomial $w$. 
The output of the substitution automaton \( \mathcal{A} \) on the given polynomial \( f \) is provided by the \( (s,t) \)-th entry of the matrix \( f(M'_{x_1}, M'_{x_2}, \ldots, M'_{x_n}) \), which is obtained by substituting the matrix \( M'_{x_i} \) for each variable \( x_i \) where \( i \in [n] \) in the polynomial.

The substitution automaton processes each monomial of $f$ and sums over the resulting monomials.

\subsection{Kronecker Product and the Mixed Product Property}
\label{subsec:kronecker}
We will crucially use the Kronecker product of matrices and the mixed product property of the Kronecker product in the proof of Lemma \ref{composing-substitution-mat}. We define them below.
\begin{definition} (Kronecker Product of Matrices) \label{kronecker_product-def}
Let $A$ be a $m \times n$ matrix and let $B$ be a $p \times q$ matrix. Then their Kronecker product $A \otimes B$ is defined as follows.
\begin{displaymath}
\mathbf{A \otimes B} = \left( \begin{array}{cccc}
a_{11}B &  a_{12}B & \ldots & a_{1n}B \\
\vdots & \vdots & \ddots & \vdots\\
a_{m1}B &  a_{m2}B & \ldots & a_{mn}B
\end{array} \right) 
\end{displaymath}
where $a_{ij}$ are the entries of matrix $A$.
\end{definition} 

\begin{lemma}(Mixed Product Property) \label{mixed_product_property-lem}
Let $A,B,C,D$ be matrices of dimensions such that the matrix products $AB$ and $CD$ can be defined.  Then, the mixed product property states that:
$$(A \otimes B)\cdot(C \otimes D) = (A\cdot C) \otimes (B\cdot D).$$
\end{lemma}

\subsection{The Polynomial Identity Lemma}

We state the well-known DeMillo-Lipton-Schwartz-Zippel lemma (a.k.a The Polynomial Identity Lemma) for commutative polynomials.
\begin{lemma}[DeMillo-Lipton-Schwartz-Zippel Lemma \cite{DL78,Zippel79,Sch80}]  
\label{dlsz}
Let $f(x_1,x_2,\ldots,x_n)$ be a non-zero degree $D$ polynomial over a field $\mathbb{F}$, and let $S\subseteq \mathbb{F}$ be a finite subset. If we choose a uniformly at random point $a=(a_1,a_2,\ldots,a_n) \in S^n$ then 
$$Pr[f(a_1,a_2,\ldots,a_n)=0]\leq D/|S|.$$
\end{lemma}

\section{Composing Substitution Matrices}
In our proof outline, we discussed the process of transforming the polynomial computed by a depth-5 +-regular circuit into a commutative polynomial. Our method involves a series of transformations of the polynomial computed by a depth-5 +-regular circuit (or any constant depth +-regular circuit, in general). The final result is a commutative polynomial that maintains non-zeroness.  

Each transformation uses substitution matrices obtained from the corresponding substitution automaton. It is important to note that a sequence of matrix substitutions cannot be used in our black-box model, as it requires a single matrix substitution for each input variable. We show that this sequence of substitutions can be combined into the substitution of a single matrix. 

We formally present this in the following lemma, and the proof can be found in Appendix \ref{Appendix-B}. The proof relies on the Kronecker product, and in particular, on the mixed product property of the Kronecker product and matrix multiplication (see \cref{subsec:kronecker}).
\begin{lemma}[Composing $\mathrm{K}$ Substitution Matrices]
 \label{composing-substitution-mat}
 Let $f\in \F\angle X$ be a non-zero n.c. polynomial of degree $D$ over $X=\{x_1,\ldots,x_{n}\}$. Let $\mathrm{K}$ be a natural number. For each
$2\leq i\leq \mathrm{K}+1,n_i\in \mathbb{N}$, define sets of non-commutative variables $Z_i=\{z_{i1},\ldots,z_{in_i}\}$, where $n_i\in \mathbb{N}$.

Assume we have matrices $A_i=(A_{i1},A_{i2},\ldots,A_{in_i})$ of dimension $d_i$ for each $2 \leq i \leq \mathrm{K}+1$. For each $j \in [n_i]$, the matrix $A_{ij}$ belongs to $ \F^{d_i\times d_i}\angle{Z_{i+1}}$ and is expressed as: $$A_{ij}=\sum_{k=1}^{n_{i+1}} A^{(k)}_{ij} z_{i+1,k},$$ where $A^{(k)}_{ij}\in \F^{d_i\times d_i}$.  Set $n_1=n$ and define  $f_0=f$. For $i\geq 1$, let $f_i$ be the $\left[1,d_i\right]$-th entry of $f_{i-1}(A_{i1},A_{i2},\ldots,A_{in_{i}})$ for all $i\geq 1$.

There exists a matrix substitution $C=(C_{1},C_{2},\ldots,C_{n})$ where 
each $C_i$ is of dimension  $\left(\prod_{i\in[\mathrm{K}]}d_i\right)$ such that the polynomial $f_\mathrm{K}$ is equivalent to the $\left[1,\prod_{i\in[\mathrm{K}]}d_i\right]$-th entry of the matrix  $f(C_{1},C_{2},\ldots,C_{n})$.

\end{lemma}
In our application of the above lemma, the final substitution matrices $A_{\mathrm{K}+1}$ only contain commutative variables as their entries. Consequently, the resulting polynomial  $f_\mathrm{K}$ is a commutative polynomial in these variables.

\section{Black-Box PIT for $\Sigma\Pi^*\Sigma\Pi^*\Sigma$ Circuits}
In this section, we show that if $f\in \F\angle X$ is a n.c. polynomial of degree $D$ computed by a depth-5 $+$-regular circuit of size $s$, then there exists an efficient randomized algorithm for identity testing the polynomial $f$. The main result of this section is the following theorem.
\begin{theorem}
\label{thm-depth-5}
 Let $f$ be a n.c. polynomial of degree $D$ over $X=\{x_1,\ldots,x_n\}$ computed by a $\Sigma\Pi^*\Sigma\Pi^*\Sigma$ circuit of size $s$. Then $f\not\equiv 0$  if and only if $f$ is not identically zero on the matrix algebra $\mathbb{M}_{s^6}(\mathbb{F})$.
\end{theorem}
 The polynomial $f$ can be written as follows: $f=\sum\limits_{i \in [s]}\prod\limits_{j \in [D_2]} Q_{ij}$. Here, the degree of each $Q_{ij}$ is $D_1$, where $i\in[s],j\in[D_2]$, and it can be computed by a $\Sigma\Pi^*\Sigma$ circuit of size at most $s$. We establish the Theorem \ref{thm-depth-5} in three steps. First, we transform each polynomial $Q_{ij}$ into a more structured n.c. polynomial, as defined below (see Definition \ref{powersum-def}). This structured polynomial has the property that we can simply consider it as a commutative polynomial preserving non-zeroness (Claim \ref{ord-pow-sum-c-nc}). As noted above, this is not true for n.c. polynomials, in general.

\begin{definition}
\label{powersum-def}
 Let $s\in \mathbb{N}\cup \{0\}$. We call a n.c. polynomial $g$ over $\xi=\{\xi_1,\xi_2,\ldots,\xi_s\}$ as an $s$-ordered power-sum polynomial if it is of the form $$g=\sum_{i_1\geq 0,\ldots,i_s\geq 0} \alpha_{\overline{i}}.\xi_1^{i_1}\xi_2^{i_1}\ldots\xi_s^{i_s}, \text{ where $\alpha_{\overline{i}} \in \F$}.$$
\end{definition}
\begin{remark}
In the above definition, it is important to note that  $i_j \geq 0$ for each $j \in [s]$. However, in this paper, we will focus on a special case of $s$-ordered power-sum polynomials, where $i_j > 0$, for each $j \in [s-1]$ and $i_{s} \geq 0$. 
\end{remark}
 We have the following observation about the \emph{$s$-ordered power-sum polynomial} and a proof of this can be found in Appendix \ref{small_proofs}.
\begin{claim}
\label{ord-pow-sum-c-nc}
Suppose $g \in \F\angle \xi$ is an $s$-ordered power-sum polynomial. Let $g^{(c)}$ be the polynomial obtained by treating variables in $\xi$ as commutative. Then, $g\not \equiv 0$ if and only if $g^{(c)}\not \equiv 0$.
\end{claim}
\begin{remark}
\label{k-ord-coeff-comm-polys}
    We will later allow coefficients $\alpha_{\overline{i}}$ to be commutative polynomials. The proof of this generalized statement follows the same reasoning as the proof of Claim \ref{ord-pow-sum-c-nc}.
\end{remark}

Let us define the following concept, which will be relevant in subsequent sections:
\begin{definition}[$\xi$-Pattern/Ordered Power-Sum Monomial]
\label{e-pattern}
We refer to the string $\xi_1^{\ell_1}\xi_2^{\ell_2} \cdots \xi_k^{\ell_k}$ as a \emph{$\xi$-pattern}, regardless of the specific exponents of the $\xi$ variables. We also refer \emph{$\xi$-pattern} as an \emph{ordered power-sum monomial}.
\end{definition}

We now explain the three steps of our method to establish Theorem \ref{thm-depth-5}.
\subsection{Step 1: Transforming the Polynomial for Improved Structure}
\label{step1}

The initial step of our method involves transforming the polynomial to introduce more structure. During this process, we obtain a structured polynomial but also introduce some additional spurious monomials. This is one of the main differences between this work and \cite{AJMR19}. 
We show that these spurious monomials have a distinguishing property that can be used to differentiate them from the structured part.
We discuss the process of transforming each polynomial $Q_{ij}$ into an $s$-ordered power-sum polynomial. To do this, we introduce a \emph{new} set of commutative and n.c. variables as follows.\\
Let $Z=\{z_1,\ldots,z_n\}$ and let $Y=\{y_{ij} \mid i \in [n] \text{ and } j \in [s-1] \}$. The variables in $Y$ and $Z$ are commutative. Let $\xi=\{\xi_1,\xi_2,\cdots,\xi_s\}$ be the set of n.c. variables.

\begin{figure}
\begin{center}
\begin{tikzpicture}
\node(pseudo) at (-1,0){}; 
\node(0) at (0,0)[shape=circle,draw, minimum size=1cm]
     {$q_0$}; 
     \node(1) at (2,0)[shape=circle,draw, minimum size=1cm] {$q_1$};
     \node(2) at (4,0)[shape=circle,draw, minimum size=1cm] {$q_2$}; 
     \node(3) at (6,0)[shape=circle,draw, minimum size=1cm] {$q_{s-3}$};
     \node(4) at (9,0)[shape=circle,draw, minimum size=1cm, minimum size=1cm] {$q_{s-2}$};
     \node(5) at (12,0)[shape=circle,draw,double,fill=green!40, minimum size=1cm] {$q_{s-1}$}; 

     \path [->] (0) 
     edge node [above] {$y_{1i}\xi_1$} (1) (1) 
     edge node [above] {$y_{2i}\xi_2$} (2) (2)
     edge[dotted] node [above] {$\cdots$} (3) (3) 
     edge node [above] {$y_{s-2,i}\xi_{s-2}$} (4) (4)
     edge node [above] {$y_{s-1,i}\xi_{s-1}$} (5) (5)
     
  (4)      edge [bend left=30]  node [below]  {$y_{s-1,i}\xi_{s-1}$}     (0)
  (5)      edge [bend left=45] node [below]   {$z_{i}\xi_{s}$}  (0)
  (0)      edge [loop above]    node [above]  {$z_{i}\xi_1$}     ()
  (1)      edge [loop above]    node [above]  {$z_{i}\xi_2$}     ()
(2)      edge [loop above]    node [above]  {$z_{i}\xi_3$}     ()
  (3)      edge [loop above]    node [above]  {$z_{i}\xi_{s-2}$}   ()
  (4)      edge [loop above]    node [above]  {$z_{i}\xi_{s-1}$}   ()
  (5)       edge [loop above]    node [above]  {$z_{i}\xi_{s}$}   ()
  (pseudo) edge                                       (0);
\end{tikzpicture}
\caption{The transition diagram for the variable $x_i : 1\leq i\leq n$}\label{fig1}
\end{center}
\end{figure}
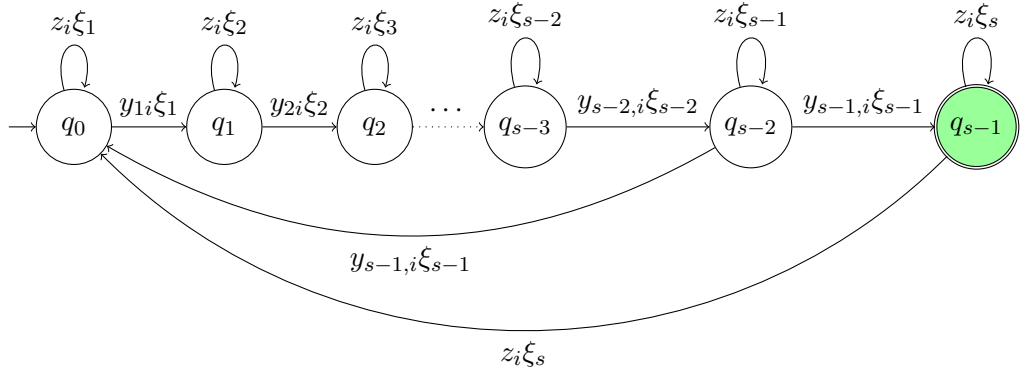
\subsubsection{Idea of the Substitution Automaton}
As each $Q_{ij}$ polynomial can be computed by a $\Sigma\Pi^*\Sigma$ circuit of size bounded by $s$, it is natural to try to use the PIT results that exist for $\Sigma\Pi^*\Sigma$ circuits. In \cite{AJMR19}, it was shown that there is a small substitution automaton that transforms the polynomial into a commutative polynomial while preserving non-zeroness. We can try to transform each $Q_{ij}$ polynomial into a commutative polynomial using the result of \cite{AJMR19}. This approach presents two issues.
\begin{enumerate}
    \item As explained in the proof outline, one would have to identify the \emph{boundary} between $Q_{ij}$ and $Q_{i,j+1}$ (i.e., the position at which $Q_{ij}$ ends and $Q_{i,j+1}$ begins), for any $i  \in [s], j \in [D_2-1]$. This task was straightforward in \cite{AJMR19} because each  $Q_{ij}$ is a linear form in their work.
    \item We need to ensure that the resulting commutative polynomial is non-zero because the same commutative variables appearing in different transformed $Q_{ij}$ polynomials get mixed (i.e., exponents of the same variable appearing at different transformed $Q_{ij}$ polynomials are added). 
    Additionally, after converting to the commutative polynomial, we must guarantee that the commutative counterpart of each $Q_{ij}$ polynomial is non-zero.  In the case of \cite{AJMR19}, since each $Q_{ij}$ is a linear form, we can treat them as commutative without the concern that $Q_{ij}$ may become a zero polynomial. However, as noted, if the degree of $Q_{ij}$ is greater than 1, we cannot guarantee non-zero values by merely considering them commutative. For example, the non-commutative polynomial  $xy-yx$ becomes zero if  $x$ and $y$ are allowed to commute.
    
\end{enumerate}
As the degree of $Q_{ij}$ could be exponential in $s$, it is not feasible to detect the \emph{boundary} by counting. So we \emph{guess} the boundary using a substitution automaton. The two states $q_{s-2}$ and $q_{s-1}$ of the substitution automaton given in Figure \ref{fig1} guess the boundary between $Q_{ij}$ and $Q_{i,j+1}$. We need two states for guessing the boundary. This is because 
in \cite{AJMR19}, it was proved that for polynomials computed by the 
$\Sigma\Pi^*\Sigma$ circuits, 
there is a small index set $I$ such that if we consider all linear forms appearing at positions indexed by $I$ as n.c. and rest as commutative then non-zeroness is preserved. Depending on whether the position of the last linear form is part of $I$ or not, the automaton is either in state $q_{s-2}$ or $q_{s-1}$, respectively. 

We consider the output of the substitution automaton given in Figure \ref{fig1} on the n.c. polynomial $f$.
 The $s \times s$ substitution matrix $\mathbf{M_{x_i}}$ for each variable $x_i$ is defined from Figure \ref{fig1}, as follows:


\begin{displaymath}
\mathbf{M_{x_i}} =
\left( \begin{array}{ccccccc}
z_{i}\xi_1 & y_{1i}\xi_1 & 0 & \ldots & 0 & 0 \\
0 & z_{i}\xi_2 & y_{2i}\xi_2 &  \ldots & 0 & 0 \\
\vdots & \vdots & \vdots & \ddots &\vdots &\vdots  \\
y_{s-1,i}\xi_{s-1} & 0 & 0 & \ldots & z_{i}\xi_{s-1} & y_{s-1,i}\xi_{s-1} \\
z_{i}\xi_{s} & 0 & 0 & \ldots & 0 & z_{i}\xi_{s} 
\end{array} \right) 
\end{displaymath}

It is helpful to consider $\mathbf{M_{x_i}}$ as the sum of matrices as follows. This view is useful when we compose matrices (see Lemma \ref{composing-substitution-mat}) from all three steps to obtain a single matrix later on.

\begin{displaymath}
\mathbf{M_{x_i}} =
\left( \begin{array}{ccccccc}
z_{i} & y_{1i} & 0 & \ldots & 0 & 0 \\
0 & 0 & 0 &  \ldots & 0 & 0 \\
\vdots & \vdots & \vdots & \ddots &\vdots &\vdots \\
0 & 0 & 0 &  \ldots & 0 & 0 \\
0 & 0 & 0 &  \ldots & 0 & 0
\end{array} \right) \cdot \xi_1 + 
 \ldots +
\left( \begin{array}{cccccc}
0 & 0 & 0 &  \ldots & 0 & 0\\
0 & 0 & 0 &  \ldots & 0 & 0 \\
\vdots & \vdots & \vdots & \ddots &\vdots &\vdots  \\
0 & 0 & 0 &  \ldots & 0 & 0 \\
z_{i} & 0 & 0 & \ldots & 0 & z_{i} 
\end{array} \right) \cdot \xi_s
\end{displaymath}
\subsubsection*{Output of the Automaton}
Let $\mathbf{M}=f(\mathbf{M_{x_1},M_{x_2},\ldots,M_{x_n}})$. Then we consider the \emph{output of the automaton} as:
\begin{equation}
\label{output-fig1}
    f'=\mathbf{M}[q_0,q_{s-1}]
\end{equation}
which is a polynomial in the variables $\xi\sqcup Y\sqcup Z$. 

Suppose a monomial $m$ is computed by a $+$-regular circuit $C$. The monomial $m$ has non-zero coefficient in $\prod_{j \in [D_2]} Q_{ij}$ for some $i \in [s]$. This monomial can be written as $m=m_1\cdot m_2\cdots m_{D_2}$, where each sub-monomial $m_j \in X^{D_1}$ has non-zero coefficient in $Q_{ij}$. 

Next, we consider the output of the substitution automaton $\mathcal{A}$ on $m$. 
The automaton knows how to replace/substitute any variable $x_j$ at any state $q$. For simplicity, Figure \ref{fig1} illustrates only the information for variable $x_i$.
Suppose  $m=x_{i_1}.x_{i_2}.x_{i_3}\ldots x_{i_D}$, the output of the
substitution automaton $\mathcal{A}$ on the monomial $m$ is given by $$\mathbf{M_{m}}[q_0,q_{s-1}]$$ where $\mathbf{M_{m}=M_{x_{i_1}}\cdot M_{x_{i_2}}\cdots M_{x_{i_D}}}$. 
 Each variable $x_i, i \in [n]$, is substituted by a degree two monomial over $\xi \sqcup Y \sqcup Z$ (one n.c. variable and one commutative variable). Consequently, the automaton transforms the monomial $m$ into a degree $2D$ polynomial over $\xi \sqcup Y \sqcup Z$. Importantly, the \emph{new} n.c. degree (i.e., sum of exponents of $\xi$ variables) equals to the original degree $D$.
 
\subsubsection*{Computation by the substitution automaton}
The automaton has exponentially many paths (with states allowed to repeat) from  $q_0$ to $q_{s-1}$, all labeled by the same monomial $m$. Each computation path transforms the monomial $m$, originally over the variables $X$, into a new monomial over $\xi \sqcup Y \sqcup Z$.

For any path $\rho$  from $q_0$ to $q_{s-1}$, we denote the transformed monomial as  $m_{\rho}$. 
The polynomial computed by $\mathbf{M_{m}}[q_0,q_{s-1}]$ is given by $$ \sum\limits_{\rho:q_0 \overset{m}\leadsto q_{s-1} 
} m_{\rho},$$ which is a polynomial in $\F[Y\sqcup Z]\angle \xi $.
Recall that n.c. polynomials are $\F$-linear combinations of words/strings (called monomials). For a n.c. monomial $m$, we can identify the variable at position $e$ in $m$, where $1 \leq e\leq |m|$.

Recall that $m$ can be written as $m=m_1\cdot m_2\cdots m_{D_2}$.
Each computation path $\rho$ substitutes each n.c. variable in $m$ according to the automaton's transition rules, resulting in a monomial $m_{\rho}$ over new variables  $\xi\sqcup Y\sqcup Z$. We group all commutative variables appearing in $m_{\rho}$ and denote it by $c_m$, which is a commutative monomial over $ Y\sqcup Z$. The resulting monomial $m_{\rho}$ has the following form and the proof can be found in Appendix \ref{step1-proofs}

\begin{proposition}
\label{m-rho-form}
Let $\rho$ be a path from $q_0$ to $q_{s-1}$  labeled by the monomial $m$.
The transformed monomial $m_{\rho}$ can be expressed in the form: $m_{\rho}=c_{m}\cdot m'_1\cdot m'_2\cdots m'_{N}$, where $N\geq 1$, $c_m$ is a monomial over $ Y\sqcup Z$, and each $m'_\ell$ (for $\ell \in [N]$) is given by  $m'_\ell=\xi^{\ell_1}_1.\xi^{\ell_2}_2\cdots\xi^{\ell_s}_s$, where $\ell_k>0$ for $k\in[s-1]$, and $\ell_s\geq 0$.
\end{proposition}

For $i\neq j$, the exponents of the $\xi$ variables in the sub-monomials $m'_i$ and $m'_j$ can vary. In particular, it is generally possible that
$(i_1,i_2,\ldots,i_s)\neq (j_1,j_2,\ldots,j_s)$.

\subsubsection*{Types of sub-monomials: Two cases}
 It is important to note that the number of new sub-monomials $m'_i$, denoted as $N$,  may not be equal to $D_2$. This is because $N$ depends on how many times the path $\rho$ returns to the initial state $q_0$ (see Figure \ref{fig1}). Also, the sum of exponents of $\xi$ variables in each sub-monomial $m'_i$ in $m_\rho$ can vary.
 This leads us to consider two possible cases for each computation path $\rho$ that starts at $q_0$ and ends at $q_{s-1}$:  (recall $m=m_1\cdot m_2\cdots m_{D_2}$).

\begin{itemize} \label{two_cases}
 \item \label{case1} {\bf Case 1:} For each $j<D_2$, the boundary between $m_j$ and $m_{j+1}$ in $m$ is respected by the path $\rho$.  
 In this  computation path $\rho$, the state of $\mathcal{A}$ is at $q_0$ precisely when it begins processing each sub-monomial $m_j \in X^{D_1}$ for $j \in [D_2]$. This means that when $\mathcal{A}$ reads the last variable of the sub-monomial $m_{j-1}$ (for $j>1$), it transitions back to state $q_0$. As a result,
 $\mathcal{A}$ is in state $q_0$ exactly when it reads the first variable  of the sub-monomial $m_j$. This holds true for all sub-monomials $m_j$ where $j\in[N]$. 
 By Proposition \ref{m-rho-form}, the transformed monomial can be expressed as: $m_{\rho}=c_{m}\cdot m'_1\cdot m'_2\cdots m'_{N}$ where $c_m$ is a monomial over $ Y\sqcup Z$ and each $m'_\ell$ is of form $m'_\ell=\xi^{\ell_1}_1\cdots\xi^{\ell_s}_s$.
 In this case, we observe that $N=D_2$ since there are exactly $D_2$ sub-monomials in $m$.  \underline{\emph{or}}

 \item {\bf Case 2:} For some $j<D_2$, the boundary between $m_j$ and $m_{j+1}$ in $m$ is not respected by the path $\rho$. In this case, there exists a sub-monomial $m_j$, where $j \in [D_2]$, such that either (1) the computation path $\rho$ visits the state $q_0$ while processing the variable located at position $c$, where $1 < c \leq D_1$. This means $\rho$ returns to $q_0$ in the middle of processing $m_j$.
 \underline{\emph or} (2) the path $\rho$ is in a state $q_j, j\neq 0$ (i.e., other than the initial state $q_0$) while processing the variable that appears  at the first position of the sub-monomial $m_j$. 
By Proposition \ref{m-rho-form}, the transformed monomial can be expressed as: $m_{\rho}=c_{m}\cdot m'_1\cdot m'_2\cdots m'_{N}$ where $c_m$ is a monomial over $ Y\sqcup Z$ and each $m'_\ell$ is of form $m'_\ell=\xi^{\ell_1}_1\cdots\xi^{\ell_s}_s$.
 In this case, we cannot definitively say whether $N$ is equal to $D_2$ or not. \\
\end{itemize}
\begin{remark}
    Any path $\rho$ from $q_0$ to $q_{s-1}$ labeled by a monomial $m \in X^{D}$ will satisfy either Case 1 or Case 2, but not both.\\
    \end{remark}
In Case 1, we can make the following important observation about the obtained monomial $m_{\rho}$ and the proof can be found in Appendix \ref{step1-proofs}. Recall that $D_1$ is the degree of $Q_{ij}$ polynomial for all $i \in [s]$ and $j \in [D_2]$.
\begin{claim}
\label{obs-good-mon}
Let $\rho$ be a path from $q_0$ to $q_{s-1}$  labeled by the monomial $m$ that satisfies Case 1.
 In this case, for each sub-monomial $m'_\ell$, where $\ell\in [D_2]$,  of the monomial $m_\rho$, the sum of the exponents of its n.c. variables  is $D_1$. That is, $\sum_{j \in [s]} \ell_j=D_1$.
\end{claim}

 For all paths $\rho$ that satisfy Case 2, this is not true. We note this down as the following claim and the proof can be found in Appendix \ref{step1-proofs}.
 \begin{claim}
 \label{obs-bad-mon} 
 Let $\rho$ be apath from  $q_0$ to $q_{s-1}$ labeled by the monomial $m$ that satisfies Case 2. In this case, 
there exists a sub-monomial $m'_\ell$, where $\ell\in [N]$, in the obtained monomial $m_\rho$ such that the sum of the exponents of its n.c. variables  is \emph{not} equal to $D_1$. That is, $\sum_{j \in [s]} \ell_j\neq D_1$.
 \end{claim}

We crucially utilize Claims \ref{obs-good-mon} and \ref{obs-bad-mon} later to ensure the non-zeroness of the transformed commutative polynomial.

\subsubsection*{The \emph{structured} part and the \emph{spurious} part}
For a monomial $m=m_1\cdot m_2\cdots m_{D_2}$, we define the polynomial $\hat{f}_m$ as the sum of all monomials that are obtained from computation paths $\rho$ labeled by $m$ from Case 1 above. Similarly, the polynomial $F_m$ is defined as the sum of all monomials obtained from computation paths $\rho$  labeled by $m$ from Case 2 above.  
We consider the output of the substitution automaton $\mathcal{A}$ on the given n.c. polynomial $f\in \F\angle X$. 

The output of the automaton is the sum of all monomials produced by computation paths $\rho$ starting from $q_0$ and leading to $q_{s-1}$, with these paths labeled by monomials generated by a depth-5 $+$-regular circuit.

Let $Mon(f)$ be the set of all monomials computed/generated by the given depth-5 $+$-regular circuit computing $f$. That is, suppose $m$ is computed by $\prod_{j\in[D_2]} Q_{i,j}$ for some $i\in[s]$, with coefficient $\alpha_{m,i}$ then $\alpha_{m,i}\cdot m\in Mon(f)$.
Let 
\begin{equation}
    \hat{f}=\sum\limits_{\alpha_{m,i}\cdot m \in Mon(f)}\hat{f}_{\alpha_{m,i}\cdot m}
\end{equation}

\begin{equation}
\label{F-poly-defn}
    F=\sum\limits_{\alpha_{m,i}\cdot m \in Mon(f)}F_{\alpha_{m,i}\cdot m}.
\end{equation}
 We refer to $F$ as the sum of spurious monomials obtained from the automaton, which can be viewed as noise resulting from our method.

We assume that the linear forms in the $Q_{i,j}$ polynomials are numbered from 1 to $D_1$.
For $I \subseteq [D_1]$ with size at most $s-1$, define $Q_{i,j,I}$ as the polynomial obtained from $Q_{i,j}$ by treating linear forms indexed by $I$ as non-commuting and the rest of the linear forms as commuting. We also substitute the n.c variables that appear in linear forms indexed by \( I \) with double-indexed commutative variables \( Y \), as shown in the substitution automaton (see Figure \ref{fig1}).

We have the following proposition regarding the polynomial $\hat{f}$. The proof can be found in Appendix \ref{step1-proofs}.
 \begin{claim}
 \label{f-hat-form}
The polynomial $\hat{f}$ can be expressed as 
 
 $$\hat{f} =\sum_{i\in [s]}\prod_{j\in[D_2]} \sum_{\substack{I \subseteq [D_1], |I|=s-1} }Q_{i,j,I}\times \xi_I$$
 where 
$\xi_I=\xi^{\ell_1}_1.\xi^{\ell_2-\ell_1}_2\cdots\xi^{D-\ell_{s-1}}_{s}$ for $I=\{\ell_1,\ell_2,\cdots,\ell_{s-1}\}$ such that $\ell_1<\ell_2<\cdots<\ell_{s-1}$.\\
 \end{claim}

Let 
\begin{equation}
\label{eq:q-ij-defn}
    \hat{Q}_{ij}=\sum\limits_{I \subseteq [D_1], |I|=s-1}Q_{i,j,I}\times \xi_I.
\end{equation}

Then we can express $\hat{f}$ as:  

\begin{equation}
\label{f-hat-defn}
   \hat{f}= \sum_{i\in [s]}\left(\prod_{j\in[D_2]}\hat{Q}_{ij}\right). 
\end{equation}

The output of the substitution automaton $\mathcal{A}$ on the polynomial $f$ 
is given by: $$f'= \hat{f} +F.$$ This is stated in the following claim. The proof can be found in Appendix \ref{step1-proofs}.
\begin{claim}
\label{f-hat-spurious}
 Let $f$ be a homogeneous n.c. polynomial computed by a $\Sigma\Pi^*\Sigma\Pi^*\Sigma$  circuit of size $s$. 
 Then, the output  $f'\in\F[Y \sqcup Z]\angle \xi$ of the substitution automaton $\mathcal{A}$ on the polynomial $f$ can expressed as $f'= \hat{f} +F,$ where $\hat{f}$
 be the structured and spurious part as defined in Equation \ref{f-hat-defn} and 
 $F$  be the spurious part as defined in Equation \ref{F-poly-defn}.

 \end{claim}

\subsubsection{ Non-zeroness of $\hat{f}$}
We establish that  $f'$ is non-zero by first proving that $\hat{f}$ is not zero. 
This is shown in Lemma \ref{gen-proj}, which builds on the result of PIT for $\Sigma\Pi^*\Sigma$ circuits (see Section 6.2 in \cite{AJMR19}). We briefly discuss this result.

 Let $Z=\{z_1,\ldots,z_n\}$ be the set of \emph{new} commuting variables.
Let $g\in \F\angle X$ be a polynomial of degree $D$ computed by a $\Sigma\Pi^*\Sigma$ circuit of size $s$.
Then $g$ can be expressed as $g=\sum_{i\in[s]}\prod_{j\in[D]}L_{ij}$, where $L_{ij}$ are homogeneous linear forms. 
Let $P_i=\prod_{j\in[D]}L_{ij}$, $i\in [s]$. We have $g=\sum_{i\in[s]}P_i$. For $I \subseteq [D]$ with size at most $s-1$, define $P_{i,I}$ as the polynomial obtained from $P_i$ by treating linear forms indexed by $I$ as noncommuting and the rest of the linear forms as commuting.
We replace each n.c.  variable $x_i$ appearing in $[D]\setminus I$ by a \emph{new} commuting variable $z_i$.

The number of n.c. linear forms appearing in $P_{i,I} \in \F[Z]\angle X$ is bounded by $|I|<s$. 
This is because the linear forms that appear with indices other than those in \( I \) are treated as commutative. Consequently, the number of non-commutative linear forms in \( P_{i,I} \in \mathbb{F}[Z] \langle X \rangle \) is bounded by \( |I| < s \).
We refer to this as the n.c. degree of the polynomial $P_{i,I}$.
Since this degree is small, $P_{i,I}$ can be converted into a commutative polynomial while preserving its non-zeroness.  Let $P^{(c)}_{i,I}$ denote the commutative polynomial obtained from $P_{i,I}$ and define $g_I=\sum_{i\in[s]}P^{(c)}_{i,I}$.
To keep all guesses of  the set $I$ distinct, additional commutative variables $\xi=\{\xi_1,\xi_2,\cdots,\xi_{k+1}\}$ are introduced in \cite{AJMR19}.
The transformed commutative polynomial obtained in \cite{AJMR19} is given by:  
\begin{equation}
\label{g-hat}
    g^{\star}=\sum\limits_{I \subseteq [D_1], |I|=k}g_I\times \xi'_I
\end{equation}

where 
$\xi'_I=\xi^{\ell_1-1}_1.\xi^{\ell_2-\ell_1-1}_2\cdots\xi^{D-\ell_k}_{k+1}$ with $I=\{\ell_1,\ell_2,\cdots,\ell_k\}$. The degree of the monomial $\xi'_I$ is $D-|I|$.

By Lemma 6.2 in \cite{AJMR19}, there exists a set of indices $I \subseteq [D]$, $|I| < s$, such that $g_I\neq 0$ implying $g^{\star}\neq 0$. Replacing $\xi'_I$ with $\xi_I=\xi^{\ell_1}_1.\xi^{\ell_2-\ell_1}_2\cdots\xi^{D_1-\ell_k}_{k+1}$
in $g^{\star}$ retains the non-zeroness of $g^{\star}$ while the degree of $\xi_I$ becomes $D$.

\begin{remark}
\begin{enumerate}
    \item     Without loss of generality, we assume that the automaton nondeterministically guesses exactly $(s-1)$ indices, i.e., $|I|=s-1$ and the rest as commutative. If $|I|<s-1$,  adding more indices still preserves non-zeroness.  
    \item If the degree of the polynomial $g$ is smaller than $(s-1)$, we will handle this small-degree case separately (See \cref{small-degree-case}). For now, we assume $D_1\geq s-1$ in Lemma \ref{gen-proj}.\\
\end{enumerate}

\end{remark}
We have the following lemma that shows $\hat{f}\neq 0$. The proof can be found in Appendix \ref{appendix:proof of Q_ij sparsification}.

\begin{lemma} \label{gen-proj}
Let $f=\sum_{i \in [s]}\prod_{j \in [D_2]} Q_{ij}$ be a n.c. polynomial over $X=\{x_1,\cdots,x_n\}$, computed by a  $\Sigma\Pi^*\Sigma\Pi^*\Sigma$ circuit of size $s$. 
Let $D_1\geq s-1$ denote the degree of the polynomial $Q_{ij}$, $i\in[s],j\in[D_2]$.
Let $\hat{f}\in\F[Y \sqcup Z]\angle \xi$ be defined as 
$\hat{f} =\sum_{i\in [s]} \prod_{j\in[D_2]}\hat{Q}_{ij},$ where 
$\hat{Q}_{ij}$ are as defined in Equation \ref{eq:q-ij-defn}. Then, if  $f\neq 0$ then $\hat{f}\neq 0$. 
\end{lemma}
We will need a generalization of this Lemma \ref{gen-proj} for polynomials computed by larger depth +-regular circuits. The generalization is stated in Lemma \ref{gen-proj-highdepth}.

The resulting n.c. polynomial $\hat{f}$ (from Lemma \ref{gen-proj})  still has an exponential degree in $\xi$ variables, but each $\hat{Q}_{ij}$ is structured as  $s$-ordered power-sum polynomials. Importantly, $\hat{f}$ does not contain any monomials from the spurious polynomial $F=\sum_{\substack{m \in Mon(f)}}F_m$.

\subsubsection{Small degree case}
\label{small-degree-case}

It is important to note that we require the degree $D_1$ of each polynomial $Q_{ij}$ to be at least $s-1$. If $D_1< s-1$, we can not use Lemma \ref{gen-proj}. 
The advantage of this lemma lies in the fact that in the polynomial $\hat{f}$, each polynomial $\hat{Q}_{ij}$ is an $s$-ordered power-sum polynomial. 
By Claim \ref{ord-pow-sum-c-nc}, each $\hat{Q}_{ij}$ can be treated as a commutative polynomial without leading to a zero polynomial.
This property is essential for the product sparsification lemma (see Lemma \ref{poly-proj}). \\

When $D_1 < s-1$, we can use a small substitution automaton of size $c$, bounded by $s$, to substitute fresh n.c. double-indexed variables at each position within each 
$Q_{ij}$. Let $Z=\{z_{\ell k}\mid \ell \in [c] \text{ and } k \in [n]\}$ be the \emph{new} set of n.c. variables. The automaton replaces the variable $x_k$ in the $\ell$-th position of a monomial of $Q_{ij}$ with $z_{\ell k}$ (this is a standard set-multilinear conversion of non-commutative polynomial $Q_{ij}$).  To avoid introducing new notation for these transformed polynomials, we refer to the resulting polynomial by $\hat{Q}_{ij}$, which remains n.c. over $Z$.
Unlike the high-degree case, in the small-degree case, the automaton can identify the \emph{boundary} between $Q_{ij}$ and $Q_{i,j+1}$, ensuring that \emph{no spurious monomials are produced}, that is, $F=0$.

The substitution automaton that accomplishes these substitutions and the corresponding substitution matrix for each variable $x_j$ can be found in Appendix \ref{automaton-small-degree}.
The following proposition is true because there is a bijection between monomials of $\hat{Q}_{ij}$ and $Q_{ij}$.
\begin{proposition}
    For $i\in[s],j\in[D_2]$, each $Q_{ij}$ in the polynomial $f\in\F\angle X$ is transformed into $\hat{Q}_{ij}$ such that $\hat{Q}_{ij}\equiv 0$ if and only if $Q_{ij}\equiv 0$.
\end{proposition}
It is easy to observe the following because the first index of each variable $z_{\ell k}$ indicates the position of the variable $x_k$ within each $Q_{ij}$.

\begin{observation}
    Suppose $\hat{Q}_{ij}\neq 0$. If we treat the variables $z_{\ell k}, \ell \in [c],k \in [n]$, appearing in $\hat{Q}_{ij}$ as commuting, the resulting commutative polynomial $\hat{Q}^{(c)}_{ij}$ remains non-zero.
\end{observation}
 
This guarantees that for the small degree case, we can transform the polynomial $f$ similarly to Lemma \ref{gen-proj}, ensuring that each $Q_{ij}$ is transformed into $\hat{Q}_{ij}$ which can be regarded as a commutative polynomial without making it zero.   While $\hat{Q}_{ij}$ remains a n.c. polynomial over $Z$, we acknowledge that our model is black-box and we do not know the value of $D_1$. However, for the purpose of analyzing the existence of matrices of small dimensions for identity testing, we can assume $D_1$ is known.

Thus, we can successfully transform the given polynomial in both scenarios -- whether
$D_1\geq s-1$ or $D_1<s-1$ -- ensuring that the resulting $\hat{Q}_{ij}$ can be considered as a commutative polynomial without making it a zero polynomial.

However, it is important to note that this transformation alone will not provide a black-box PIT, as we cannot guarantee the non-zeroness of the sum of products of these $\hat{Q}_{ij}$ polynomials. This is because if we simply treat all 
$\hat{Q}_{ij}$ as commutative,  the variables across different $\hat{Q}_{ij}$ polynomials could mix, which may lead to cancellations. At this stage, the variables in $Z$ are still considered n.c. in the polynomial $\hat{f}$.

\subsubsection{Non-zeroness of $f'$}
By Lemma \ref{gen-proj}, we established that $\hat{f} \not\equiv 0$. 
Next, we show that the polynomial $f'=\hat{f}+ F \not\equiv 0$.
In $\hat{f}$, for every monomial $m=m_1m_2\cdots m_{D_2}$, and for all $\ell \in [D_2]$ each $m_\ell$ takes the form $\xi^{\ell_1}_1.\xi^{\ell_2}_2\cdots\xi^{\ell_s}_s$  where  $\sum_{k\in[s]}\ell_k=D_1$ (see Claim \ref{obs-good-mon}). 
However, this property does not hold for monomials appearing in $F$ (see Claim \ref{obs-bad-mon}). Specifically, for any monomial $m'=m'_1m'_2\cdots m'_{N}$ in $F$,
there exists a sub-monomial  $m'_a=\xi^{a_1}_1.\xi^{a_2}_2\cdots\xi^{a_s}_s$  such that $\sum_{h\in[s]}a_h\not=D_1$.
This distinction ensures that the monomials of $\hat{f}$ do not cancel with those of $F$. Thus, we conclude that $f'=\hat{f}+F \not\equiv 0$. It's important to note that 
 if $f\equiv0$ then clearly $f'\equiv 0$ as well (converse statement). We note these observations in the following claim. 
 \begin{claim}
 \label{fhat-spurious-non-zero}
 Let $f$ be a homogeneous n.c. polynomial computed by a depth-5 +-regular circuit of size $s$. Then, $f\not\equiv 0$ if and only if $f'=\hat{f}+F\not\equiv 0$.
 \end{claim}

Next,  we can simplify the polynomial $f'$ using the Polynomial Identity Lemma for commutative polynomials. We replace the commuting variables $Y \sqcup Z$ with scalar substitutions from $\F$ or an extension field, yielding a non-zero polynomial.
Let us denote this resulting non-zero polynomial as $\tilde{f}$. After this substitution, the only remaining variables in $\tilde{f}$ will be n.c. variables $\xi$.

Let us denote new polynomials obtained after replacing the commuting variables by scalars in $\hat{f}$ and $F$ by $\hat{f}_1$ and $F_1$ respectively. That is, $\tilde{f}=\hat{f}_1+F_1$.

One of the goals of this transformation is to ensure that if $\xi_s$ is followed by $\xi_1$ in the transformed monomials (for all such occurrences of $\xi_s$ followed by $\xi_1$), there must be a transition from $Q_{i,j}$ to $Q_{i,j+1}$ for some $j\in[D_2]$. As noted, we cannot be sure of this. However, all those monomials where this transition occurs are captured in the structured part  $\hat{f}_1$. Since there is no such structure in $F_1$, we cannot conclude anything about the monomials appearing in the spurious part $F_1$.

\begin{remark}
\label{remark-product-of-e-pattern}
    It is important to note that each monomial of $\tilde{f}$ is a product of $\xi$-patterns (see Definition \ref{e-pattern}),  and the boundaries of each $\xi$-pattern can be easily identified by an automaton which is crucially used by the remaining steps.
\end{remark}

\subsection{Step 2: Product Sparsification} \label{step2-product-sparsification}

In the second step of our transformation, we prove a {\emph general lemma} that states if we have a sum of a small number of products of ordered power-sum polynomials, we can sparsify the product while preserving its non-zero property. Specifically, in each product term of the sum, we can treat only a small number of the ordered power-sum polynomials as non-commutative while treating the rest as commutative without affecting the non-zero nature of the polynomial. In particular, this step does not depend on the number of terms in each product. 

We focus on the sparsification of the n.c. polynomial $\tilde{f}\in \F \angle \xi$, which was the output of Step (1). This transformation affects both the good part $\hat{f}_1$ and the spurious part $F_1$ of the polynomial $\tilde{f}\in \F \angle \xi$.  We begin by analyzing the transformation of $\hat{f}_1$, which is defined as:
$$\hat{f}_1= \sum_{i\in [s]}\left(\prod_{j\in[D_2]}\hat{Q}_{ij}\right),$$

where each $\hat{Q}_{ij}$ is an $s$-ordered power-sum polynomial in the n.c. variables $\xi=\{\xi_1,\ldots,\xi_s\}$. Note that each $\hat{Q}_{ij}$ is a homogeneous and degree $D_1$ n.c. polynomial.

The key observation is that we can preserve the non-zeroness of $\hat{f}_1$ by retaining at most $s-1$ of the $s$-ordered power-sum polynomials $\hat{Q}_{ij}$ in each product $\prod_{j\in[D_2]}\hat{Q}_{ij}$ as non-commutative while treating the remaining ones as commutative.   This is stated in the following lemma, which we refer to as the \emph{product sparsification lemma}. This lemma generalizes Lemma 6.2 from \cite{AJMR19}. However, unlike in \cite{AJMR19}, we are working with the product of non-commutative polynomials where the degree of individual factors can be greater than 1. If we simply treat them as commutative, as in \cite{AJMR19}, they may become zero.   The proof of this lemma crucially relies on Claim \ref{ord-pow-sum-c-nc}, and the complete proof can be found in Appendix \ref{app-proof}. Unlike \cite{AJMR19}, one of the key distinctions in our setting is that the $\hat{Q}_{ij}$ polynomial can be non-homogeneous in general.

We will revisit the combination of $\hat{f}_1$ and $F_1$ in \cref{{step1.1}} to complete our analysis.

\begin{lemma}[Product Sparsification Lemma]
\label{poly-proj}
 Let $$\hat{f}_1=\sum_{i\in [s]} \prod_{j\in[D_2]}\hat{Q}_{ij},$$ where each $\hat{Q}_{ij}$ is an $s$-ordered power-sum polynomial of degree $D_1$ over $\xi=\{\xi_1,\xi_2,\ldots,\xi_s\}$. Then, there exists a subset $I \subseteq [D_2]$ with size at most $s-1$ such that if we treat the polynomials $\hat{Q}_{ij}$
 for $j \in I$, as non-commutative and the others ($j \not\in I$) as commutative, then the  polynomial $\hat{f}_1$ remains non-zero.  Furthermore, each $\hat{Q}_{ij}$ polynomial may be non-homogeneous in general. Moreover, there is a small substitution automaton of size $O(s)$ that performs this transformation. 
\end{lemma}
 In other words, for each n.c. variable $\xi_i$ (where $i \in [s]$) in  $\hat{f}_1$, there exists an  $O(s)$-dimensional matrix—acting as a transition matrix of a substitution automaton of size $O(s)$.  By evaluating $\hat{f}_1$ on these matrices, the polynomial is transformed into a product-sparsified polynomial while maintaining its non-zero property.

\begin{remark}
\label{prod-spars-gen}
  We remark that the proof of Lemma \ref{poly-proj} (see Appendix \ref{app-proof}) relies solely on the fact that each $\hat{Q}_{ij}$ polynomial in $\hat{f}_1$ is an ordered power-sum polynomial. In particular, the proof does not depend on the fact that each $\hat{Q}_{ij}$ is obtained from a $\Sigma\Pi^*\Sigma$  circuit. 
Instead, the proof relies on the following facts:

\begin{enumerate}
    \item the number of summands is small,
    \item the boundary of each ordered power-sum polynomial can be efficiently identified using a small automaton
\end{enumerate}

This makes it irrelevant where the $\hat{Q}_{ij}$ polynomials originate from.  As a result, we can apply this result whenever the given polynomial is represented as a sum of a small number of products of ordered power-sum polynomials (i.e., the number of summands is small). We will use this observation when working with higher-depth $+$-regular circuits. 
\end{remark}

Since we do not know the index set $I$, the substitution automaton guesses the index set $I$. 
The substitution automaton that accomplishes this product sparsification can be found in Appendix \ref{automaton-product-sparsification}. 
Since the index set  $I\subseteq [D_2]$ is unknown, the automaton non-deterministically selects which $\hat{Q}_{ij}$ polynomials will be treated as non-commutative. Given the structured nature of the polynomial $\hat{f}$, we can identify the \emph{boundary} of each $\hat{Q}_{ij}$, ensuring that no additional spurious monomials are generated. 

In the high-degree case ($D_1\geq s-1$),  either $\xi_s$ or $\xi_{s-1}$ followed by $\xi_1$  indicates the end of each $\hat{Q}_{ij}$, which can be easily recognized by the automaton. In the low-degree case ($D_1< s-1$), the smaller degree allows us to identify the ends of each $\hat{Q}_{ij}$ with a small automaton of size at most  $s-2$.

The substitution automaton selects at most $(s-1)$ of the $\hat{Q}_{ij}$ polynomials to be treated as n.c. while treating the remaining ones as commutative.  
The $\xi$ variables  in the chosen commutative polynomials $\hat{Q}_{ij}$ are substituted with fresh commutative variables $\zeta=\{\zeta_1,\ldots,\zeta_s\}$. 
In the selected commutative polynomials $\hat{Q}_{ij}$, 
each n.c. variable $\xi_k, k \in [s]$ is replaced by the corresponding commuting variable $\zeta_k$. Additionally, to distinguish between different guesses made by the substitution automaton, we use fresh commutative block variables 
$\chi=\{\chi_1,\ldots,\chi_s\}$.

Let $J=\{j_1,j_2,\ldots,j_{s-1}\}\subseteq [D_2]$ with $j_1<j_2<\ldots<j_{s-1}$. We define \mbox{$\chi_J=\chi_1^{j_1-1}.\chi_2^{j_2-j_1-1}\ldots \chi_s^{D_2-j_{s-1}}$}.
If the automaton guesses the $\hat{Q}_{ij}$ polynomials corresponding to the positions in the index set $J$ as n.c., the output $g_J$ of the substitution automaton for this specific guess  $J$ will be 
\begin{equation}
\label{gj-poly}
    g_J=\sum_{i\in [s]}\left(\prod_{j\in \overline{J}} \hat{Q}_{ij}\right) \left(\prod_{j\in J} \hat{Q}_{ij}\right) \times \chi_J.
\end{equation}

Note that $\left(\prod_{j\in \overline{J}} \hat{Q}_{ij}\right)$ is a commutative polynomial over $\zeta=\{\zeta_1,\ldots,\zeta_s\}$. 
We have the following lemma.

\begin{lemma} \label{chi signature lemma}
Let $\hat{f}_1 \in \F\langle\xi\rangle$ be the structured part of the polynomial obtained after Step 1. Let $\hat{f}^{'}_1$ be the output of the substitution automaton given in Figure \ref{automaton-product-sparsification} (can be found in Appendix \ref{app-proof}) on the structured polynomial $\hat{f}_1$ and it can be expressed as $$\hat{f}^{'}_1=\sum_{J \subseteq [D_2],|J|=s-1}g_J.$$
Moreover, $\hat{f}_1 \neq 0$ if and only if $\hat{f}^{'}_1\neq 0$.
\end{lemma}

It's evident that for distinct guesses $J$ and $J'$ where  $J\neq J'$,
the monomials of $g_J$ and $g_J'$ will not mix, since the sub-monomials $\chi_J$ and $\chi_{J'}$  are distinct (see Equation \ref{gj-poly}). By Lemma \ref{poly-proj}, there exists index set $J \subseteq [D_2]$ with size at most $s-1$, such that $g_J\neq 0$ implying $\hat{f}_1'\neq 0$.

Next, we can simplify  $\hat{f}^{'}_1$ by using the Polynomial Identity Lemma for commutative polynomials to eliminate the commuting variables $\zeta \cup \chi$ by substituting scalars. As a result, the remaining variables in the polynomial will be solely the n.c. variables $\xi$.

Let us denote the new polynomial obtained after replacing the commuting variables by scalars in $\hat{f_1}'$ by $\hat{f}_2$.

This product sparsification step affects both the good part $\hat{f}_1$ and the spurious part $F_1$ of the polynomial $\tilde{f}$ obtained after Step 1.  We will denote the new polynomial derived from the spurious part $F_1$ by $F_2$.  In Step 2,  we apply product sparsification to both $\hat{f}_1$ and $F_1$, which yields the n.c. polynomials $\hat{f}_2$ and $F_2$ respectively (with all commuting variables replaced by scalars).

\subsection{Step 3: Commutative Transformation of $\hat{f}_2$}

In this final step,  we prove a general commutative transformation lemma, which states that if we have a non-commutative polynomial represented as a sum of products of a small number of ordered power-sum polynomials (i.e., the number of terms in each product is small), we can convert it into a commutative polynomial while preserving its non-zeroness property. In particular, this step does not depend on the number of summands. The key idea is to introduce a small number of new commutative variables to perform this transformation.

We now describe how to transform $\hat{f}_2$ into a commutative polynomial while preserving its non-zeroness. Note that $\hat{f}_2$ is a polynomial over $\F\langle \xi\rangle$. If we treat $\hat{f}_2$ as commutative 
by considering the n.c. variables $\xi$ as commutative, the exponents of the variable $\xi_i$ (for $i \in [s]$) from different n.c. $\hat{Q}_{ij}$ polynomials
will be summed (or mixed). This mixing makes it impossible to guarantee that the resulting polynomial remains non-zero.

However, we can carefully convert $\hat{f}_2$ into a commutative polynomial while preserving its non-zeroness. This is stated in the following lemma and the proof can be found in Appendix \ref{proof:step3}. In particular, there is a substitution automaton of size $O(s^2)$ that carries out this commutative transformation and the substitution automaton is provided in Appendix \ref{automaton-step3}.

\begin{lemma}[ Commutative Transformation Lemma]
\label{com-conversion}
 Let $g=\sum_{i\in [s]}\beta_i\big(\prod_{j\in[s]} \hat{Q}_{ij})$, where $\beta_i\in \F$  and each $\hat{Q}_{ij}$ is an $s$-ordered power-sum polynomial over $\xi=\{\xi_1,\xi_2,\ldots,\xi_s\}$ of degree $D$. This can be expressed as: $g=\sum_m \alpha_m m$, where  $m=\big(\prod_{j \in [s]} \xi_1^{i_{j1}}\xi_2^{i_{j2}}\ldots\xi_s^{i_{js}})$.  
 The n.c. polynomial $g$ can be transformed into a commutative polynomial 
 $g^{(c)}$ while preserving its non-zeroness. In particular, there exists a small substitution automaton of size $O(s^2)$  that performs this commutative transformation. 
\end{lemma}

In other words, for each n.c. variable $\xi_i$ (where $i \in [s]$), there exists an  $O(s^2)$-dimensional matrix—acting as a transition matrix of a substitution automaton.  By evaluating $g$ on these matrices, the polynomial is transformed into a commutative polynomial $g^{(c)}$, while maintaining its non-zero property.

\begin{remark}
 \label{com-trans-gen}

 We remark that Lemma \ref{com-conversion} is more general. The proof (see Appendix \ref{proof:step3}) depends only on the fact that the given non-commutative polynomial can be represented as a sum of products of a small number of ordered power-sum polynomials (i.e., the number of terms in each product is small). Crucially, the proof of this commutative transformation does not depend on the fact that the polynomials are derived from $+$-regular circuits or whether they are homogeneous. Instead, the proof relies on the following two facts:
\begin{enumerate}
    \item The number of terms in each product is small, and
    \item Each term in the product is an ordered power-sum polynomial so that boundaries can be identified efficiently using an automaton. 
\end{enumerate}

This makes the result applicable whenever the given non-commutative polynomial is represented as a sum of products of a small number of ordered power-sum polynomials. We will use this observation when working with higher-depth $+$-regular circuits.

 \end{remark}

 By applying Lemma \ref{com-conversion}, we can transform the polynomial $\hat{f}_2$, the structured part obtained after Step 2, into a commutative polynomial while preserving its non-zeroness. Let $\hat{f}^{(c)}_3$ denote the resulting commutative polynomial derived from $\hat{f}_2$. Consequently, we establish that $\hat{f}^{(c)}_3\not\equiv 0$ 
as a result of this lemma.

Next, given $\tilde{f}=\hat{f}_1+F_1$, where $\tilde{f}$ was obtained after Step 1, we can likewise transform $\tilde{f}$ into a commutative polynomial.
 Let $F^{(c)}_3$ represent the commutative polynomial obtained from $F$ after applying steps (2) and (3). If $\hat{f}^{(c)}_3+F^{(c)}_3\not\equiv 0$, we have successfully converted a n.c. polynomial $f$, computed by a depth-5 $+$-regular circuit, into a commutative polynomial that preserves non-zeroness. 
We can now check the non-zeroness of this commutative polynomial using the Polynomial Identity Lemma for commutative polynomials. 

Assume $\hat{f}^{(c)}_3+F^{(c)}_3=0$.  We will now detail how to modify the coefficients of certain monomials in $\tilde{f}$, which was obtained in Step 1, before executing Steps (2) and (3). We establish that this coefficient modification maintains non-zeroness and remains non-zero even after the application of  Steps (2) and (3).

\subsection{Coefficient Modification by Modulo Counting Automaton} \label{step1.1}

Assuming $\hat{f}^{(c)}_3+F^{(c)}_3=0$, we know that $\hat{f}^{(c)}_3\not\equiv 0$ which implies that $F^{(c)}_3\not\equiv0$ and $\hat{f}^{(c)}_3=-F^{(c)}_3$. 
To resolve this, we will carefully modify some of the monomial coefficients in the n.c. polynomial $\tilde{f}=\hat{f}_1+F_1$ (see Lemma \ref{lem-modp-counting}) before proceeding with product sparsification (Lemma \ref{poly-proj}) and commutative transformation (Lemma \ref{com-conversion}). We will show that these modifications ensure the resulting polynomial remains non-zero after Steps (2) and (3).
Recall that $\tilde{f}=\hat{f}_1+F_1$ is derived after replacing the commutative variables $Y \sqcup Z$ in the polynomial $f'\in\F[Y \sqcup Z]\angle \xi$  with scalar substitutions (as part of Step 1).

Let $\tilde{m}$ be a monomial in the commutative polynomial $\hat{f}^{(c)}_3\in \mathbb{F}[W]$ with a non-zero coefficient $\alpha_{\tilde{m}}\in\F$. This monomial also appears in $F^{(c)}_3\in F[Z]$ with coefficient $-\alpha_{\tilde{m}}\in\F$.

Now, consider a non-zero n.c. monomial $m$ in $\hat{f}_1$. We use $m \in \hat{f}_1$ to denote this. The monomial $m$ can be expressed as $m=m_1.m_2\ldots m_{D_2}$ where $m_i=\xi_1^{i_1}.\xi_2^{i_2}\ldots \xi_s^{i_s}$ for $i\in[D_2]$.
By Claim \ref{obs-good-mon}, for all $i\in[D_2]$, we have $\sum_ji_j=D_1$.
 If we apply product sparsification and commutative transformation (Steps (2) and (3)) to this monomial $m$, the resulting polynomial $G_m$ over commutative variables $W$ may include the commutative monomial $\tilde{m}$ with a non-zero coefficient, denoted as $\exists m \leadsto \tilde{m}$.\\

 We define the set  
$A^{\tilde{m}}_{{\hat{f}_1}}=\{m \in \hat{f}_1 \mid \exists m \leadsto \tilde{m}\}$ and the set $B_{F_1}^{\tilde{m}}=\{m' \in F_1 \mid \exists m' \leadsto \tilde{m}\}.$ \\

It's important to note that $A^{\tilde{m}}_{{\hat{f}_1}} \cap B_{F_1}^{\tilde{m}}=\phi$.
Let  $m'\in B_{F_1}^{\tilde{m}}$,  expressed as $m'=m'_1.m'_2\ldots m'_{N}$ with each 
$m'_i=\xi_1^{i_1}.\xi_2^{i_2}\ldots \xi_s^{i_s}$, where $N$ may differ from $D_2$.
By Claim \ref{obs-bad-mon}, there exists an $i\in[N]$ such that the exponents of the sub-monomial $m'_i$ satisfies $\sum_ji_j\neq D_1$.

To ensure non-zeroness during the commutative transformation, we will modify the coefficients of the n.c. monomials in $A^{\tilde{m}}_{{\hat{f}_1}}$ and $B_{F_1}^{\tilde{m}}$ differently. Since only the monomials in $A^{\tilde{m}}_{{\hat{f}_1}}\sqcup B_{F_1}^{\tilde{m}}$
are transformed into the commutative monomial $\tilde{m}$, after this coefficient modification the monomial $\tilde{m}$ will have a non-zero coefficient after product sparsification and commutative transformation (Steps (2) and (3)).

 It is important to note that, during this process, each n.c. variable $\xi_i,i\in[s]$ is replaced by the same n.c. variable $\xi_i$ ensuring no additional monomials are created; the only change involves the adjustment of the monomial coefficients.

This is stated in the following lemma and the proof can be found in Appendix \ref{proof:modp-counting}.

\begin{lemma}[Coefficient Modification Lemma]
\label{lem-modp-counting}
    Let  $\tilde{f}=\hat{f}_1+F_1 \in \F\angle \xi$ be the non-zero non-commutative polynomial obtained in Step 1. Let $\hat{f}^{(c)}_3$ and $F^{(c)}_3$  be the commutative polynomials obtained 
 from $\hat{f}_1$ and $F_1$ after Steps (2) and (3), respectively. If $\hat{f}^{(c)}_3+F^{(c)}_3=0$, then the coefficients of $\tilde{f}$ can be modified before Steps (2) and (3)  in such a way that the resulting non-commutative polynomial $\tilde{f}^{'}\in \F\angle \xi$ (after coefficients modification step) can be transformed into a commutative polynomial through Steps (2) and (3) while ensuring that the resulting commutative polynomial retains its non-zeroness. \\
In particular, there exists a small substitution automaton of size $O(s^2)$  that performs this commutative transformation. 
\end{lemma}

In other words, if $\hat{f}^{(c)}_3+F^{(c)}_3=0$, then for each non-commutative variable $\xi_i \in [s]$, there exists an $O(s^2)$-dimensional matrix substitution—acting as a transition matrix of a substitution automaton— such that when the n.c. polynomial $\tilde{f}$ (obtained after Step 1) is evaluated on these matrices, the coefficients of
$\tilde{f}$ are modified, resulting in a non-commutative polynomial $\tilde{f}^{'}\in \F\angle \xi$ that maintains its non-zero property. 
 Furthermore, this newly obtained non-commutative polynomial $\tilde{f}^{'}$ can be transformed into a commutative polynomial through Steps (2) and (3), while ensuring that the resulting commutative polynomial retains its non-zero property.

\subsection{Black-box Randomized PIT for $\Sigma\Pi^*\Sigma\Pi^*\Sigma$ Circuits}

Each of these three steps, along with the coefficient modification step, results in its own set of matrices for evaluation. In particular, the matrices obtained in each step evaluate a n.c. polynomial derived from the previous step. 

Given that our model operates as a black box, we cannot evaluate the polynomial in this manner.  Instead, we require a single matrix substitution for each n.c. variable. To address this, we apply Lemma \ref{composing-substitution-mat} to combine the substitution matrices from all four steps into a single matrix for each n.c. variable.

This approach allows us to establish an efficient randomized polynomial identity testing (PIT) algorithm for depth-5 $+$-regular circuits, as demonstrated in the following theorem.

\begin{theorem}
\label{thm-depth-5-proof}
Let $f$ be a non-commutative polynomial of degree $D$ over $X=\{x_1,\ldots,x_n\}$,
computed by a $\Sigma\Pi^*\Sigma\Pi^*\Sigma$ circuit of size $s$. Then $f\not\equiv 0$ if and only if it does not evaluate to zero on the matrix algebra $\mathbb{M}_{s^6}(\mathbb{F})$.
\end{theorem}
\begin{proof}
There are two directions in the proof. The backward direction is straightforward, as evaluating a zero polynomial  $f$ will always result in a zero matrix. Now, we will proceed to demonstrate the forward direction. 

The dimensions of the substitution matrices obtained for Steps (1), (2), and (3) are $s$, $(4s-2)$, and $s^2$ respectively (see Figures \ref{fig1}, \ref{automaton-product-sparsification} and \ref{automaton-step3}). Let $\mathbf{A}$, $\mathbf{B}$ and $\mathbf{C}$  denote these substitution matrices. The output of Step 1 is the polynomial $\mathbf{A}[1,s]$. Evaluating this output on $\mathbf{B}$, we obtain the output of Step 2, which is given by the polynomial $\mathbf{B}[1,4s-3]+\mathbf{B}[1,4s-2]$. Finally, evaluating the output of Step 2 on $\mathbf{C}$, gives us the output of Step 3, $\mathbf{C}[1,s^2]$. \\

We define the dimensions as follows: let $d_1=s$, $d_2=4s-3$, $d^{'}_2=4s-2$, and $d_3=s^2$. \\

According to Lemma \ref{composing-substitution-mat}, the substitution matrices 
$\mathbf{A=(A_1,\ldots,A_n)}$, $\mathbf{B=(B_1,\ldots,B_s)}$ and $\mathbf{C=(C_1,\ldots,C_s)}$ can be combined into a single matrix substitution $\mathbf{M}=(\mathbf{M_{x_1},M_{x_2},\ldots,M_{x_n}})$ of dimension $O(s^4)$. This results in an overall matrix substitution $\mathbf{M}=(\mathbf{M_{x_1},M_{x_2},\ldots,M_{x_n}})$ of dimension $O(s^4)$.
 We evaluate the polynomial as $\mathbf{O}=f(\mathbf{M_{x_1},M_{x_2},\ldots,M_{x_n}})$.

The output of the substitution automaton is defined as the sum of two entries of the matrix $O$.
$$
   f^{(c)} = \mathbf{O}\left[\mathbf{1},\left(\mathbf{d_1\cdot d_2\cdot d_3}\right) \right]  + \mathbf{O}\left[\mathbf{1},\left(\mathbf{d_1\cdot d'_2\cdot d_3}\right)\right]
$$
This differs from previous works (\cite{AJMR19,BW05}) where only a single matrix entry was used as the output.
It is important to note that, by the matrix composition lemma (see Lemma \ref{composing-substitution-mat}), the polynomial $f^{(c)}$ is equal to the polynomial obtained as the output of Step 3, which is $\mathbf{C}[1,s^2]$. This summation arises because the automaton for Step (2) has two accepting states, leading to two entries in the final composed matrix.

If $f^{(c)}\in \F[W]$ is identically zero, we can apply the coefficient modification step (See \cref{step1.1}) to obtain a non-zero commutative polynomial. The maximum dimension of the substitution matrix in this coefficient modification step is $O(s^2)$ (see Figures \ref{RemainderNFA} and \ref{xi-pattern consumption}), which will result in a final matrix of dimension $O(s^6)$
due to the matrix composition lemma (Lemma \ref{composing-substitution-mat}).

Since one of the automatons used in the coefficient modifications step (see Figure \ref{RemainderNFA}) has two accepting states, the output will be the \emph{sum of four entries} of the resulting matrix. Consequently, we have shown that $f^{(c)}\in \F[W]$ is non-zero if and only if $f\neq 0$.  The non-zeroness of $f^{(c)}$ implies that at least one of the four entries of the matrix is non-zero. The summation of entries was defined for analysis purposes.

 Finally, the non-zeroness of the polynomial $f^{(c)}$ can be tested using the DeMillo-Lipton-Schwartz-Zippel lemma. Thus, we conclude that the polynomial $f$  is non-zero if and only if it is not identically zero on the matrix algebra $\mathbb{M}_{s^6}(\mathbb{F})$.

 This completes the proof of the theorem.
\end{proof}

\begin{remark}
We observe that the commutative polynomial $f^{(c)}\in \F[W]$ is an  $s^2$-ordered power-sum polynomial over $W$, in the sense that we can arrange the variables in each monomial of $f^{(c)}$ in increasing order according to the first index of the $W$ variables, allowing for some exponents to be zero as specified in Definition \ref{powersum-def}.
\end{remark}

This is summarized in the following theorem.

\begin{theorem}
\label{f-to-powersum-d5}
Let $f$ be a non-zero non-commutative polynomial of degree $D$ over $X=\{x_1,\ldots,x_n\}$ computed by a $\Sigma\Pi^*\Sigma\Pi^*\Sigma$ circuit of size $s$. 
Then, $f$ can be transformed into an $s^{2}$-ordered power-sum polynomial while preserving its non-zeroness. In particular, there exists a small substitution automaton of size $O(s^6)$  that performs this transformation. 
\end{theorem}
 In other words, for each n.c. input variable $x_i$ (where $i \in [n]$) in  $f$, there exists an  $O(s^6)$-dimensional matrix—acting as a transition matrix of a substitution automaton of size $O(s^6)$.  By evaluating $f$ on these matrices, the polynomial $f$ is transformed into an $s^{2}$-ordered power-sum polynomial while preserving its non-zeroness (in one particular entry of the resulting matrix).

\subsubsection{An Automaton for Theorem \ref{f-to-powersum-d5}}
\label{automaton-f-to-powersum-d5}
We can envision a substitution automaton $\mathcal{A}$ for Theorem \ref{f-to-powersum-d5} as follows.
By applying the matrix composition Lemma \ref{composing-substitution-mat}, we can combine the substitution matrices obtained from Steps (1) through (3), along with the modifications to coefficients, into a single substitution matrix  $\mathbf{M}=(\mathbf{M_{x_1},M_{x_2},\ldots,M_{x_n}})$ of dimension $O(s^6)$. We then evaluate the polynomial as $\mathbf{O}=f(\mathbf{M_{x_1},M_{x_2},\ldots,M_{x_n}})$.

The output of the substitution automaton is defined as the sum of several entries of the matrix $\mathbf{O}$ (refer to the polynomial $f^{(c)}$ defined in the proof of  Theorem \ref{thm-depth-5-proof}). 

It is crucial to note that each entry of the matrices in $\mathbf{M}=(\mathbf{M_{x_1},M_{x_2},\ldots,M_{x_n}})$ is a monomial over $Y\sqcup Z\sqcup \zeta \sqcup \chi\sqcup W$, where $Y\sqcup Z$ are commutative variables from Step (1),  and $\zeta\sqcup \chi$ are commutative variables from Step (2), while $W$ contains commutative variables from Step (3). 

As noted in Steps (1) and (2), we can replace the commutative variables in 
$Y\sqcup Z\sqcup \zeta \sqcup \chi$ with scalars without losing the non-zeroness of the output polynomial (by the Polynomial Identity Lemma). 

 After these replacements, each entry of the matrices in $\mathbf{M}=(\mathbf{M_{x_1},M_{x_2},\ldots,M_{x_n}})$ transforms into scalar multiples of variables over $W$. We denote the resulting matrices as $\mathbf{M'}=(\mathbf{M'_{x_1},M'_{x_2},\ldots,M'_{x_n}})$. 

We can construct a substitution automaton $\mathcal{A}$ such that the substitution matrix for the variable $x_i$ is given by the matrix $\mathbf{M'_{x_i}}$, where the entries are scalar multiples of variables in $W$.  These entries correspond to transitions that substitute a n.c. variable with a scalar multiple of a variable in $W$. This automaton $\mathcal{A}$ effectively transforms $f$ into an $s^{2}$-ordered power-sum polynomial $f^{(c)}$ while preserving its non-zeroness.

We can view the resulting $s^{2}$-ordered power-sum polynomial $f^{(c)}$ as a n.c. polynomial over $W$. This idea is crucial for developing black-box randomized polynomial identity testing (PIT) for circuits of larger depths using induction.

 It is important to note that the monomials of $f^{(c)}$ do not correspond to a single entry of the output matrix $\mathbf{O}=f(\mathbf{M'_{x_1},M'_{x_2},\ldots,M'_{x_n}})$. Instead, they represent the sum of several entries, as indicated in the polynomial $f^{(c)}$ defined in the proof of Theorem \ref{thm-depth-5-proof}. 
 Effectively, the column numbers of these entries form the set of accepting states for the new automaton $\mathcal{A}$, with row 1 serving as the starting state of this automaton.

\section{Black-Box Randomized PIT for Small Depth +-Regular Circuits}
In this section, we present an efficient randomized black-box polynomial identity testing (PIT) algorithm for polynomials computed by small-depth $+$-regular circuits. The main result of this section is the following theorem.
\begin{theorem}
\label{main-thm}
 Let $f$ be a non-commutative polynomial of degree $D$ over $X=\{x_1,\ldots,x_n\}$, computed by a $+$-regular circuit of size $s$ and depth $d$.  Then, $f\not\equiv 0$  if and only if $f$ is not identically zero on $\mathbb{M}_{N}(\mathbb{F})$, where $N=s^{O(d^2)}$ and $|\F|$ is sufficiently large.
\end{theorem}
\subsection{Transforming $f$ into an ordered power-sum polynomial}
To demonstrate Theorem \ref{main-thm}, we first establish that the polynomial $f$  can be transformed into a more structured n.c. polynomial that facilitates the testing for non-zeroness. Specifically, we show that $f$ can be converted into a $c$-ordered power-sum polynomial while preserving its non-zeroness, where  $c$  is a function of the depth $d$ and size $s$. To prove this conversion, we need a generalization of Lemma \ref{gen-proj};  For completeness, we state this generalized lemma (see Lemma \ref{gen-proj-highdepth}). Its proof follows a similar structure to that of Lemma \ref{gen-proj}. We note this as a remark here.

\begin{remark}
\label{step1-f-hat-non-zero-gen}
The proof of Lemma \ref{gen-proj} is based on the observation that if each polynomial $Q_{ij}$ (where $i \in [s]$ and $j \in [D_2]$, computable by a depth-3 +-regular circuit) can be converted into a n.c. polynomial $\hat{Q}_{ij}$ while preserving non-zeroness, then this lemma asserts that this transformation can be extended to the entire polynomial computed by the depth-5 circuit. Specifically,  each $Q_{ij}$ can be converted into the n.c. polynomial $\hat{Q}_{ij}$, leading to the creation of the n.c. polynomial $\hat{f}$, while maintaining the non-zeroness of the entire polynomial.
This can be extended to higher depths as follows: 

Suppose
\begin{enumerate}
    \item $f \in \F \langle x_1,\ldots,x_n\rangle$ is computed by the sum of the products of $Q_{ij}$ polynomials, that is $$f=\sum\limits_{i \in [s]}\prod\limits_{j \in [D_2]} Q_{ij}$$
    where the degrees of $Q_{ij}$ are the same for all $i \in[s]$ and $j\in[D_2]$ and are computed by depth ($d-2$) $+$-regular circuits of size at most $s$.
    \item Each $Q_{ij}$ individually can be converted into a n.c. polynomial $\hat{Q}_{ij}$ (which is an ordered power-sum polynomial) while preserving its non-zeroness. That is, $Q_{ij} =0 $ if and only if $\hat{Q}_{ij}=0$.
\end{enumerate}

 Then this transformation can be extended to the entire polynomial $f$ computed by the depth-$d$ circuit as follows: each $Q_{ij}$ in $f$ can be individually converted to $\hat{Q}_{ij}$ to create the n.c. polynomial $\hat{f}$ from $f$, maintaining the non-zeroness of the entire polynomial. 

\end{remark}

The proof of this Lemma \ref{gen-proj-highdepth} can be found in Appendix \ref{appendix:proof of Q_ij sparsification-general-depth}

\begin{lemma} \label{gen-proj-highdepth}
Let $$f=\sum_{i \in [s]}\prod_{j \in [D_2]} Q_{ij}$$ be a non-commutative polynomial over $X=\{x_1,\cdots,x_n\}$, computed by a  depth $d$ $+$-regular circuit of size $s$, where each $Q_{ij}$ is  computed by a depth $(d-2)$  $+$-regular circuit of size at most $s$.
Let $D_1\geq s-1$ denote the degree of the polynomial $Q_{ij}$, for  $i\in[s],j\in[D_2]$. \\
Suppose that any non-commutative polynomial $P$, computed by a depth $(d-2)$  $+$-regular circuit, can be converted into a non-commutative polynomial $\hat{P}$ (which is an ordered power-sum polynomial) while preserving non-zeroness. Then, each $Q_{ij}$ polynomial in $f$ can be converted to $\hat{Q}_{ij}$ to create the n.c. polynomial $\hat{f} =\sum_{i\in [s]} \prod_{j\in[D_2]}\hat{Q}_{ij}$. Then, if  $f\neq 0$ then $\hat{f}\neq 0$ maintaining the non-zeroness of the entire polynomial $\hat{f}$. That is, $f= 0$ if and only if $\hat{f}=0$.
\end{lemma}

\begin{remark}
\label{step1-non-homogeneous-remark}
   Unlike in Lemma \ref{gen-proj}, the \( \hat{Q}_{ij} \) polynomials in Lemma \ref{gen-proj-highdepth} are not necessarily homogeneous. One of the hypotheses of the lemma states that any polynomial \( P \), computed by a depth-\((d-2)\) \( + \)-regular circuit, can be transformed into an ordered power-sum polynomial \( \hat{P} \) while preserving non-zeroness. However, \( \hat{P} \) itself may be non-homogeneous. Importantly, the proof of Lemma \ref{gen-proj-highdepth} does not depend on whether \( \hat{P} \) is homogeneous.

\end{remark}

We will use this generalization to transform polynomials computed by larger depth $+$-regular circuits.  It is important to note that, like Lemma \ref{gen-proj}, Lemma \ref{gen-proj-highdepth} is primarily aimed at establishing the non-zeroness of $\hat{f}$. However, in our approach, we cannot use a small automaton to directly transform $f$ into $\hat{f}$ without introducing a spurious part, $F$, as we did in the depth-5 case. In other words, when we apply a substitution automaton to transform $f$ the result is $\hat{f}+F$, where $F$ represents the spurious part.

 We have the following theorem, which asserts that any non-zero polynomial \( f \), computed by a depth-\( d \) circuit, can be transformed into an ordered power-sum polynomial. The size of the automaton that performs this transformation is analyzed in \cref{automaton-size-power-sum-transform-depth-d}.

\begin{theorem}
\label{f-to-ps-gend}
Let $f$ be a non-zero non-commutative polynomial of degree $D$ over $X=\{x_1,\ldots,x_n\}$ computed by a $+$-regular circuit of size $s$ and depth $d$. Let $d^+$ be the number of addition (i.e., $\sum$) layers in the circuit. 

Then, $f$ can be transformed into a $s^{(d^+-1)}$-ordered power-sum polynomial while preserving its non-zeroness.
\end{theorem}

\begin{proof}
 The proof is by induction on the number of $\sum$ layers in the circuit. In the rest of the proof,  $+$-regular circuits are referred to as circuits. The number of $\sum$ layers in the circuit is denoted as 
 $d^+=(d+1)/2$, and we refer to $d^+$ as $+$-depth of the circuit.
 \begin{itemize}
     \item {\bf Base Cases: } The theorem holds for the base cases: $d^+=1$ (linear forms), and $d^+=2$ (as shown in \cite{AJMR19}) and $d^+=3$ as established in Theorem \ref{f-to-powersum-d5}. 
     \item {\bf Inductive Step: }
        \begin{itemize}
            \item {\bf Induction Hypothesis: } For the induction hypothesis, we assume that the theorem holds for any polynomial $g$ computed by a circuit of size $s$ and $+$-depth $(d^+-1)$. 
            Specifically, this means that $g$ can be converted to an \emph{$s^{(d^+-1)}$-ordered power-sum polynomial} while preserving non-zeroness.
        \end{itemize}

 Now, suppose $f$  is computed by a $+$-depth $d^+$ circuit of size $s$. 
 We can express this polynomial as:
$$f=\sum_{i \in [s]}\prod_{j \in [D_2]} Q_{ij}.$$

Here, each $Q_{ij}$ (for $i\in[s],j\in[D_2]$) can be computed by a $+$-depth $(d^+-1)$  circuit of size at most $s$. That is, each $Q_{ij}$  can be computed by a depth-($d-2$) +-regular circuit.

We denote the degree of $Q_{ij}$ polynomials as $D_1$. 
Therefore, $f$ is a sum of products of polynomials computed by $+$-depth $(d^+-1)$  circuits of size at most $s$. By induction hypothesis, 
for all $i \in [s], j \in [D_2]$, each $Q_{ij}$ can be converted into an 
\emph{$s^{(d^+-2)}$-ordered power-sum polynomial} $\hat{Q}_{ij}$ while preserving non-zeroness. As noted in Remark \ref{step1-non-homogeneous-remark}, $\hat{Q}_{ij}$ polynomials can be non-homogenous in general.

Next, we need to establish that these individual conversions preserve the non-zeroness of the transformed polynomial. By Lemma \ref{gen-proj-highdepth}, this can be established. Specifically, the n.c. polynomial 
 $\hat{f}_1$  defined as
$$\hat{f}_1=\sum_{i \in [s]}\prod_{j \in [D_2]} \hat{Q}_{ij},$$

where each $\hat{Q}_{ij}$ polynomial is a \emph{$s^{(d^+-2)}$-ordered power-sum polynomial}, maintains non-zeroness.


Following this, the polynomial $\hat{f}_1$ can be product sparsified in a manner analogous to the depth-5 case.  As $\hat{f}_1$ polynomial is represented as a sum of $s$ number of products of ordered power-sum polynomials, thus Lemma \ref{poly-proj} can be used to product sparsify the polynomial $\hat{f}_1$.
 Due to the structured nature of $\hat{Q}_{ij}$ polynomials (ordered power-sum polynomials), boundaries can be identified by a small automaton.  
Thus, we can construct an automaton $\mathcal{A}_2$ similar to the one in Figure \ref{automaton-product-sparsification}. We consider the output of $\mathcal{A}_2$ on $\hat{f}_1$ and we can replace the commutative variables from this resulting polynomial with scalars, yielding a new n.c. polynomial $\hat{f}_2$ while maintaining non-zeroness.

Finally, the polynomial $\hat{f}_2$ can be converted into a commutative polynomial while maintaining non-zeroness, similar to the depth-5 case.  As $\hat{f}_2$ is represented as a sum of products of a small number of ordered power-sum polynomials, thus Lemma \ref{com-conversion} can be used to convert $\hat{f}_2$ into a commutative polynomial $\hat{f}_3$  that preserves non-zeroness.

 Boundaries of $\hat{Q}_{ij}$ polynomials can be identified by a small automaton due to their structured nature. Thus, we can construct an automaton 
$\mathcal{A}_3$ similar to one in Figure \ref{automaton-step3}. 
We consider the output  of $\mathcal{A}_3$ on $\hat{f}_2$. The resulting commutative polynomial is denoted $\hat{f}_3$.
The commutative polynomial $\hat{f}_3$ is an \emph{$s^{(d^+-1)}$-ordered power-sum polynomial}. 

Thus, we have shown that a polynomial $f$ computed by a circuit of size $s$ and $+$-depth $d^+$ can be transformed into a $s^{(d^+-1)}$-ordered power-sum polynomial while preserving its non-zeroness.
 \end{itemize}
\end{proof}

\subsubsection{Analyzing the Automaton Size for Theorem \ref{f-to-ps-gend}}
\label{automaton-size-power-sum-transform-depth-d}
Now we analyze the size of the automaton that transforms the given polynomial $f$ into an ordered power-sum polynomial in Theorem \ref{f-to-ps-gend}.

Before we proceed further, we recall that the result of Step 1 transformation in depth-5 case is a homogeneous n.c. polynomial $\tilde{f}=\hat{f}_1+F_1$. As noted the main advantages of Step 1 are 
\begin{enumerate}
    \item $\hat{f}_1$ is an ordered power-sum polynomial
    \item boundaries of $\hat{Q}_{ij}$ polynomials (that is boundaries of monomials belonging to $\hat{Q}_{ij}$) can be  efficiently identified using an automaton
\end{enumerate}

However, after the product sparsification step (Step 2), the resulting polynomial \( \hat{f}_2 + F_2 \) could become non-homogeneous. This is because the spurious part \( F_1 \) from Step 1 is a sum of the product of \( \xi \)-patterns (see Definition \ref{e-pattern}), and we cannot guarantee anything about the degree of each \( \xi \)-pattern in the monomials of \( F_1 \). But in each monomial in \( \hat{f}_1 \), each \( \xi \)-pattern has the same degree. In each monomial of the polynomial \( \tilde{f} = \hat{f}_1 + F_1 \), the product sparsification step (Step 2) just selects some number of \( \xi \)-patterns as non-commutative, and the rest of the \( \xi \)-patterns as commutative. Therefore, the resulting polynomial \( \hat{f}_2 + F_2 \) can be non-homogeneous in general (after removing commutative variables by scalar substitutions).

In Step 3, we simply convert the polynomial \( \hat{f}_2 + F_2 \) into an ordered power-sum polynomial using fresh variables. These variables can be considered as non-commutative for subsequent applications using induction. Thus, the resulting ordered power-sum polynomial \( \hat{f}_3 + F_3 \) can be non-homogeneous in general. Moreover, the coefficient modification step only updates the coefficient of the monomial (without changing any variables in it), thus it will not affect the degree.

The n.c. polynomial \( \hat{f}_2 + F_2 \) obtained after the Step 2 transformation in the depth-5 case can be expressed as follows. Let \( K \) represent the degree of \( \hat{f}_2 \). After the product sparsification step, the n.c. degree becomes \( K = s \times D_1 \), where \( s \) is the number of terms in the sum and \( D_1 \) is the degree of the individual polynomials \( Q_{ij} \).

We can write \( \hat{f}_2 + F_2 \) as:

\begin{equation}
\label{redefine-good}
    \hat{f}_2 + F_2 = \left(\hat{f}_2 + F^{K}_2\right ) + F^{\neq K}_2
\end{equation}

where:

\begin{itemize}
    \item \( F^{K}_2 \) represents the sum of all monomials in \( F_2 \) that have degree exactly \( K \).
    \item \( F^{\neq K}_2\) includes all the monomials in \( F_2 \) that have a degree different from \( K \).
\end{itemize}

Now, observe that \( \hat{f}_2 \) and \( F^{\neq K}_2 \) are distinct, meaning their monomials do not interfere with each other's coefficients. This is because the degree of the monomial does not match. This distinction remains true even after Step 3, as the transformation from Step 2 is bijective.

On the other hand, the monomials in \(  F^{K}_2 \) could potentially affect the coefficients of monomials in \( \hat{f}_2 \). Specifically, if \( \hat{f}_2 +  F^{K}_2 = 0 \), this will hold true even after Step 3. To preserve non-zeroness, we can use the coefficient modification step (after Step 1), which ensures that \( \hat{f}_2 +  F^{K}_2\) maintains non-zeroness throughout the process.

For higher depths, we need to redefine the notion of "good" and "spurious" parts. This is because we use only the degree to distinguish between the good and spurious parts.
For example, consider the polynomial computed by a depth-7 circuit, which can be expressed as:

\[
g = \sum_{i \in [s]} \prod_{j \in [D_3]} P_{ij}
\]

where \( P_{ij} \) is a polynomial computed by a depth-5 circuit. Let \( f \) be \( P_{ij} \) for some \( i \in [s] \) and \( j \in [D_3] \).

To apply the results from the depth-5 case to depth-7 circuits, we define the "good" part of the polynomial \( P_{ij} \) as \( \hat{P}_{ij} = \hat{f}_2 + F^{K}_2 \). 
This is because we use only the degree to distinguish between the good and spurious parts, the monomials in \( F^{K}_2 \) are included in the good part for applying these results to higher depths.
Specifically, the "spurious" part of the Step 1 transformation of the polynomial computed by the depth-7 circuit includes all monomials whose degrees are not equal to \( K \). Similar to the depth-5 case, in both the good and spurious parts each monomial after Step 1 is a product of ordered power-sum monomials.

Let  $$f=\sum_{i \in [s]}\prod_{j \in [D_2]} Q_{ij}$$  be a polynomial computed by a depth-$d$ circuit of size $s$. We denote the number of addition (i.e., $\sum$) layers in the circuit by  $d^+$. In particular,  $Q_{ij}$ can be computed by a depth ($d-2$) circuit.

 Suppose we have substitution matrices  $\mathbf{M}=(\mathbf{M_{x_1},M_{x_2},\ldots,M_{x_n}})$ that transforms each $Q_{ij}$ polynomials into an $s^{(d^+-2)}$-ordered power-sum polynomial while preserving its non-zeroness. We define that this \( s^{(d^+-2)} \)-ordered power-sum polynomial includes only the monomials where the degrees match. Any remaining monomials, which were obtained using a depth \( (d-2) \) circuit, will be included in the spurious part $F_1$ of the depth-\( d \) polynomial $f$.

 Let $\mathcal{A}$ be the corresponding substitution automaton with these substitution matrices  $\mathbf{M}$.

We can then construct an automaton to perform Step (1) of our method, transforming $f$ into $$f_1=\hat{f}_1+F_1,$$ as follows. In \cref{automaton-f-to-powersum-d5}, we described how to construct a substitution automaton that converts a polynomial $f$ computed by a depth-5 circuit into an $s^{2}$-ordered power-sum polynomial $f^{(c)}$ while preserving its non-zeroness.  We can treat $f^{(c)}$ as a n.c. polynomial by treating variables introduced in Step 3 as non-commutative.

Given that $f$ is a sum of products of $Q_{ij}$ polynomials and that the boundaries of these $Q_{ij}$ polynomials are unknown, we can modify the substitution automaton $\mathcal{A}$ to guess the boundaries of each 
$Q_{ij}$. This involves adding transitions to the starting state of automaton $\mathcal{A}$ (similar to the automaton depicted in Figure \ref{fig1}). 
For example, we need to add transitions from each accepting state of the automaton $\mathcal{A}$  to its starting state. We denote the new substitution automaton by $\mathcal{A'}$

Due to the uncertainty in identifying the boundaries, the polynomial \( f \), computed by a depth-\( d \) circuit, transforms into the n.c. polynomial
$$ f_1 = \hat{f}_1 + F_1, $$
where \( F_1 \) includes both the spurious sub-monomials resulting from incorrect boundary guesses made by the automaton \( \mathcal{A}' \) for \( Q_{ij} \), as well as the ordered power-sum monomials (see Definition \ref{e-pattern}) produced by \( \mathcal{A}' \) where the degrees do not match (see, for example, Equation \ref{redefine-good}). Since the degrees of these sub-monomials do not match those of \( \hat{Q}_{ij} \), any monomials containing these sub-monomials will be part of \( F_1 \). Thus, we can use claims similar to Claims \ref{obs-bad-mon}, \ref{obs-good-mon} to distinguish between the good and spurious parts.

Here,

$$\hat{f}_1=\sum_{i \in [s]}\prod_{j \in [D_2]} \hat{Q}_{ij},$$ 

where each $\hat{Q}_{ij}$ being an $s^{(d^+-2)}$-ordered power-sum polynomial. The polynomial $f_1$ preserves non-zeroness.

The advantage of $\hat{f}_1$  is that the boundaries of $\hat{Q}_{ij}$ can be effectively identified by the automaton. Subsequently, similar to the depth-5 case, Steps (2), (3), and coefficient modifications can be applied to the polynomial $f_1$. Using Lemma \ref{composing-substitution-mat}, the resulting substitution matrices at each step can be combined to form a single matrix for each input variable of the polynomial $f$. 

When we evaluate the polynomial $f$ using these substitution matrices, the output is an $s^{(d^+-1)}$-ordered power-sum polynomial while preserving its non-zeroness. As in the depth-5 case, the output is generally spread across several entries of the output matrix due to the automata used having multiple accepting states.

\subsection*{Size of the substitution automaton $\mathcal{A}'$}
Let $\hat{f}_1$ be the structured part of the polynomial obtained after Step (1), applied to the polynomial $f$, which is computed by a
depth-$d$ circuit of size $s$.  Recall that $\hat{f}_1$ is a sum of products of $s^{(d^+-2)}$-ordered power-sum polynomials.

\begin{itemize}
    \item {\bf Automaton size for Step 2:} The size of the substitution automaton that product-sparsifies the polynomial $\hat{f}_1$ is $(4s-2)$. 
    It is very similar to the substitution automaton shown in Figure \ref{automaton-product-sparsification}.  We only need to add more transitions to this figure because $\hat{f}_1$  is a sum of products of 
    $s^{(d^+-2)}$-ordered power-sum polynomials, where $s^{(d^+-2)}$ could be greater than $s$. Thus, there are more variables than in the depth-5 case. Let the resulting polynomial be $\hat{f}_2$ with no commutative variables.

    \item {\bf Automaton size for Step 3:} The size of the substitution automaton that transforms $\hat{f}_2$ into a commutative polynomial, specifically into an $s^{(d^+-1)}$-ordered power-sum polynomial, is $s^{(d^+-1)}$.  This automaton is quite similar to that used in the depth-5 case (see Figure \ref{fig2}).
  
    \item {\bf Automaton size for coefficients modification step:} 
    The size of the substitution automaton that filters out one spurious monomial from the polynomial $F_1$ (obtained by Step 1) is at most $14s$.
    This is due to the prime number $p$ required for this step being bounded by $\leq 4.4 s$ (see Lemma \ref{prime-lemma}). Additionally,  we need 3 states to process the $\xi$-patterns (ordered power-sum monomial, see Definition \ref{e-pattern}).  This automaton construction is similar to the automaton illustrated in Figure \ref{xi-pattern consumption} in the depth-5 case. We refer to each ordered power-sum monomial as a \( \xi \)-pattern, where \( \xi \) denotes the set of variables used in the definition of an ordered power-sum polynomial (see Definition \ref{powersum-def}).

    The size of the substitution automaton that modifies the coefficients of the polynomial $\hat{f}_1 + F_1$ (also obtained by Step 1) is bounded by $4.4s^{(d^+-1)}+3$. This is analogous to Figure \ref{RemainderNFA} in the depth-5 case. Each $\xi$-pattern (see Definition \ref{e-pattern}) of $\hat{f}_1$ contains $s^{(d^+-2)}$ variables, and the prime number $p$ required for this step is bounded by $4.4\log D$, where $D$ is the degree of $f$ and bounded by $2^s$. Therefore, with one starting state and two additional accepting states, the automaton size becomes 
    $4.4s^{(d^+-1)}+3$. 
     
    Consequently, the maximum size of the automaton for this step is $4.4 s^{(d^+-1)}+3$.

    \item {\bf Automaton size for Step 1:}  Let $$M_{n}= \sum\limits_{\substack{k \text{ odd} \\ k \leq n}} k-3$$ for $n\geq 5$. Note that $M_{n}$ represents the sum of all odd numbers up to $n$, excluding 3. 
    Assuming $d\geq 5$, the automaton for this step has a size of $O\left(18^{(d^+-3)}\cdot s^{M_{d-2}}\right)$. This matches the automaton size of Step 1 in the depth-5 case (where $d^+=3$), which is $s$.
 
\end{itemize}

Therefore, the final automaton size is the product of the automaton sizes in each step. This results in the following expression for the final automaton size:
\begin{eqnarray*}
   \text{Size} &=&O\left(18^{(d^+-3)}\cdot s^{M_{d-2}}\right)\times \left(4.4s^{(d^+-1)}+3\right)\times (4s-2) \times \left(s^{(d^+-1)}\right)\\
    &=&O\left(18^{(d^+-3)}\cdot s^{M_{d-2}}\right) \times \left(17.6s^d-8.8s^{(2d^+-2)}+12s^{d^+}-6s^{(d^+-1)} \right)   \text{ as $d=(2d^+-1$)} \\
    &=&O\left(18^{(d^+-2)}\cdot s^{M_{d}}\right) \\
    &=&O\left(18^{(d^+-2)}\cdot s^{\left(\frac{d+1}{2} \right)^2-3}\right)\\
    &=& s^{O\left(d^2\right)} 
\end{eqnarray*}

Since we can always assume that the top and bottom layers of the given $+$-regular circuit are addition layers (as noted in Remark \ref{top-bottom-plus}), the depth $d$ is always odd.

\begin{remark}
 Similar to the depth-5 case, we can address both small and high degree scenarios:  $D_1 \geq s-1$ and $D_1 < s-1$ (see \cref{small-degree-case}).
\end{remark}

It is important to note that while the automaton for Theorem \ref{f-to-ps-gend} utilizes an automaton that changes the coefficients of monomials, it does not turn a zero non-commutative polynomial  $f$ into a non-zero one. 
Suppose $f$  is computed by a depth $d$ circuit of size $s$. 
 We can express this polynomial as:
$$f=\sum_{i \in [s]}\prod_{j \in [D_2]} Q_{ij}.$$

Now, consider the case where two monomials generated by $f$ cancel each other out. Specifically, suppose a monomial  $m$ is generated by two different products, $\prod_{j \in [D_2]} Q_{i_1j}$ and $\prod_{j \in [D_2]} Q_{i_2j}$, where $i_1\neq i_2$ and they cancel each other. Let $M$ be the obtained substitution matrices as above. If we evaluate these monomials on $M$, the two different ways of generating the monomial transform in the same manner. Therefore, if they cancel each other before the transformation, they will also cancel each other after the transformation.

\subsection{Randomized Identity Test for Small Depth +-Regular Circuits}
 We are now ready to state and prove the main theorem. 
 
\begin{theorem}
 Let $f$ be a non-commutative polynomial of degree $D$ over $X=\{x_1,\ldots,x_n\}$, computed by a $+$-regular circuit of size $s$ and depth $d$. We denote the number of addition (i.e., $\sum$) layers in the circuit by  $d^+$. Then, $f\not\equiv 0$  if and only if $f$ is not identically zero on $\mathbb{M}_{N}(\mathbb{F})$, where $N=s^{O(d^2)}$ and $|\F|$ is sufficiently large.
\end{theorem}
\begin{proof}
We use Theorem \ref{f-to-ps-gend} to convert the n.c. polynomial $f$ into an $s^{(d^+-1)}$-ordered power-sum polynomial $f_{ops}$, while preserving its non-zeroness. As discussed above, there is a substitution automaton of size bounded by $s^{O(d^2)}$, which results in substitution matrices of dimension $s^{O(d^2)}$. By Claim \ref{ord-pow-sum-c-nc}, $f_{ops}$ can be treated as a commutative polynomial while preserving its non-zeroness. Using the DeMillo-Lipton-Schwartz-Zippel lemma, we can have a randomized PIT for depth $d$ $+$-regular circuit of size $s$ using matrices of dimension at most $s^{O(d^2)}$. This completes the proof of the theorem.
\end{proof}

\section{Discussion}

In this work, we presented a randomized polynomial-time algorithm for black-box polynomial identity testing (PIT) for non-commutative polynomials computed by $+$-regular circuits. Our method efficiently handles circuits of any constant depth. 
While our algorithm resolves the PIT problem for $+$-regular circuits of constant depth, the randomized identity testing problem for general non-commutative circuits, where the degree and sparsity can be exponentially large, remains an open question. We hope that some of the ideas developed in this work will prove useful in addressing the more general case.

\section*{Acknowledgement}
We would like to extend our sincere gratitude to Prof. Arvind (IMSc and CMI) for his valuable discussions. SR also expresses his sincere gratitude to Prof. Arvind and Prof. Meena Mahajan (IMSc) for facilitating a visit to IMSc, where part of this research was conducted.
We also acknowledge the assistance of ChatGPT in rephrasing sections of this paper to improve clarity and articulation. However, we affirm that no technical ideas or proofs presented in this paper were generated by ChatGPT.

\appendix
\section{\\Proof of Lemma \ref{gen-proj} (part of Step 1)} \label{appendix:proof of Q_ij sparsification}
\begin{proof}
The proof is by induction on $D_2$. Assume $f$ is non-zero. 
\begin{itemize}
    \item {\bf Base Case: $D_2=1$. }  
    For $D_2=1$, each $Q_{i1}$ for $i\in[s]$ can be computed by a $\Sigma\Pi^*\Sigma$ circuit. Since the sum  of $\Sigma\Pi^*\Sigma$ circuits is also a $\Sigma\Pi^*\Sigma$ circuit, it follows that a $\Sigma\Pi^*\Sigma$ circuit can compute the polynomial $f$. Furthermore, using the results from \cite{AJMR19}, the polynomial $f$ 
    can be transformed into another polynomial  $\tilde{f}$ preserving non-zeroness (see Equation \ref{g-hat}). Treating the block variables used in \cite{AJMR19} as non-commuting maintains the non-zeroness of $\tilde{f}$. Thus, the resulting polynomial remains non-zero. This completes the proof of the base case.\\
 \item {\bf Inductive Step: }
 \begin{itemize}
     \item {\bf Inductive Hypothesis:} We assume that the lemma holds for polynomials computed by a sum of products of at most $(D_2-1)$ $Q_{ij}$ polynomials. We will show it also holds for $D_2$.\\
 \end{itemize}
Consider the polynomial $f$ expressed as: $$f=\sum_{i \in [s]}Q_{i1}P_i,$$
 where $P_i=\prod_{j=2}^{D_2} Q_{ij}$.
We can expand $P_i$ into a sum of monomials, denoting the coefficient of a monomial $m$ in   $P_i$ by $[m]P_i$.  Thus, we can rewrite $f$ as:
 $$f=\sum_{i \in [s]}Q_{i1}\left(\sum_{m \in X^{D-D_1}}\left([m]P_i\right)\times m\right),$$
 where $D=D_1\times D_2$ is the degree of the $f$ and $\left([m]P_i\right) \in \F$. Since $f\not \equiv 0$, there exists a monomial $m \in X^{D-D_1}$ such that the right derivative $f^{(m)}$ does not vanish. We define $f^{(m)}$ as:
$$f^{(m)}=\sum_{i \in [s]}Q_{i1}\times \left([m]P_i\right),$$ 
 where $\left([m]P_i\right) \in \F$. \newline

Since $f^{(m)}\not\equiv 0$, it can be computed by a $\Sigma\Pi^*\Sigma$ circuit of size at most $s$. This is because the sum of $\Sigma\Pi^*\Sigma$ circuits is also a $\Sigma\Pi^*\Sigma$ circuit, it follows that a $\Sigma\Pi^*\Sigma$ circuit can compute the polynomial $f$.
 This reduces to the base case. 
Therefore, $f^{(m)}$ can be transformed into $\hat{f}^{(m)}=\sum_{i \in [s]}\hat{Q}_{i1}\times \left ([m]P_i\right)$ while preserving non-zeroness.\\

{\bf Handling Variables in $Q_{ij}$: } \\
We assume the linear forms in the $Q_{i,j}$ polynomials are indexed from 1 to $D_1$. For a subset $I \subseteq [D_1]$ of size $s-1$, we define $Q_{i,j,I}$ as the polynomial obtained from $Q_{i,j}$ by treating linear forms indexed by $I$ as non-commuting and the rest of the linear forms as commuting. Let $I=\{i_1,i_2,\ldots,i_{s-1}\}$ with $i_1< i_2<\ldots<i_{s-1}$.\\

Next, we introduce sets of variables: 
 $Z=\{z_1,\ldots,z_n\}$ and let $Y=\{y_{ij} \mid i \in [n] \text{ and } j \in [s-1] \}$, where variables in $Y$ and $Z$ are commutative. The set  $\xi=\{\xi_1,\xi_2,\cdots,\xi_s\}$ consists of non-commuting variables. \\

Replace the variable $x_i$ that appears in $[D_1]\setminus I$ with a \emph{new} commuting variable $z_i$. \\

The number of non-commuting variables in $Q_{i,j,I} \in \F[Z]\angle X$ is bounded by $|I|<s$, which is referred to as the n.c. degree of $Q_{i,j,I}$. Since this degree is small, $Q_{i,j,I}$  can be converted into a commutative polynomial while preserving its non-zeroness by replacing $x_i$ variables of linear form that appears at position $i_k \in I$ by $y_{k,i}$. Let $Q_{i,j,I}^{(c)}$ denote the resulting commutative polynomial, which lies in  $\F[Y\sqcup Z]$. 
To ensure all guesses of the set $I$ are distinct, additional variables $\xi=\{\xi_1,\xi_2,\cdots,\xi_{s}\}$ are introduced in \cite{AJMR19}.  We keep $\xi$ variables as \emph{non-commutative}.\\

The transformed n.c. polynomial $\hat{Q}_{i1}\in \F[Y\sqcup Z]\angle \xi$ can be expressed as:  
\[
 \hat{Q}_{i1}=\sum\limits_{I \subseteq [D_1], |I|=s-1}Q_{i,1,I}\times \xi_I
\]
where 
$\xi_I=\xi^{\ell_1}_1.\xi^{\ell_2-\ell_1}_2\cdots\xi^{D_1-\ell_{s-1}}_{s}$ with $I=\{\ell_1,\ell_2,\cdots,\ell_{s-1}\}$. The degree of the monomial $\xi_I$ is $D_1$. By Lemma 6.2 in \cite{AJMR19}, there exists a set of indices $I \subseteq [D]$, $|I| < s$, such that $Q_{i,1,I}\neq 0$ implying $\hat{Q}_{i1}\neq 0$. \\ 

Since $\hat{f}^{(m)}\in\F[Y \sqcup Z]\angle \xi$ is non-zero, the polynomial
$$f^{\dag}=\sum_{i \in [s]}\hat{Q}_{i1}P_i,$$ is also non-zero. Note that the derivative is only for analysis; we will not compute the right derivative of  $f$. \\

Now, expand the polynomial $\hat{Q}_{i1}$ as a sum of monomials in  $f^{\dag}$:
 $$f^{\dag}=\sum_{i \in [s]}\sum_{m \in \xi^{D_1}}\left([m]\hat{Q}_{i1}\right )m\times P_i.$$

The variables $\xi$ are the only non-commuting variables in $\hat{Q}_{i1}$,
and the n.c. degree of $\hat{Q}_{i1}$ is exactly $D_1$.
Since $ f^{\dag}\not \equiv 0$, there exists a monomial $m \in \xi^{D_1}$ such that the left derivative of $f^{\dag}$ with respect to $m$ does not vanish. Let $f^{\dag}_{(m)}$ be this polynomial. Again, note that this derivative is for analysis only, and we will not compute it. In this case, the coefficient $\left ([m]\hat{Q}_{i1}\right )$ is a polynomial in $\F[Y \sqcup Z]$.

$$f^{\dag}_{(m)}=\sum_{i \in [s]}\left ([m]\hat{Q}_{i1}\right )\times P_i,$$ $\left ([m]\hat{Q}_{i1}\right ) \in \F[Y \sqcup Z]$.\\

Clearly, $f^{\dag}_{(m)}\not\equiv 0$. To analyze this, we can simplify the polynomial   $f^{\dag}_{(m)}$ using the DeMillo-Lipton-Schwartz-Zippel Lemma for commutative polynomials to eliminate the commuting variables in 
$Y \sqcup Z$. We can replace these commuting variables with scalar substitutions from $\F$ or an extension field, preserving non-zeroness. Let $f^{''}$ be the resulting polynomial, which will have only variables non-commuting variables left:

$$f^{''}=\sum_{i \in [s]}\beta_i\times P_i=\sum_{i \in [s]}\beta_i\times\prod_{j=2}^{D_2} Q_{ij},  \text{ where $\beta_i \in \F$. } $$

Observe that $f^{''}\not\equiv 0$ and the number of product terms in the polynomial $P_i$  is exactly $(D_2-1)$. \\

By induction hypothesis, each $Q_{ij}$ in $f^{''}$ can be transformed into  $\hat{Q}_{ij}$ polynomial such that the resulting polynomial is non-zero.  
 Since $f^{\dag}_{(m)}\not\equiv 0$, the following polynomial is also non-zero:
$$\hat{f}=\sum_{i \in [s]}\prod_{j \in [D_2]} \hat{Q}_{ij}.$$ 
This completes the proof of the inductive step and therefore the lemma.
\end{itemize}
\end{proof}

\subsection{\\Proof of Lemma \ref{gen-proj-highdepth}} \label{appendix:proof of Q_ij sparsification-general-depth}

\begin{lemma} 
Let $$f=\sum_{i \in [s]}\prod_{j \in [D_2]} Q_{ij}$$ be a non-commutative polynomial over $X=\{x_1,\cdots,x_n\}$, computed by a  depth $d$ $+$-regular circuit of size $s$, where each $Q_{ij}$ is  computed by a depth $(d-2)$  $+$-regular circuit of size at most $s$.
Let $D_1\geq s-1$ denote the degree of the polynomial $Q_{ij}$, for  $i\in[s],j\in[D_2]$. \\
Suppose that any non-commutative polynomial $P$, computed by a depth $(d-2)$  $+$-regular circuit, can be converted into a non-commutative polynomial $\hat{P}$ (which is an ordered power-
sum polynomial) while preserving non-zeroness. Then, each $Q_{ij}$ polynomial in $f$ can be converted to $\hat{Q}_{ij}$ to create the n.c. polynomial $\hat{f} =\sum_{i\in [s]} \prod_{j\in[D_2]}\hat{Q}_{ij}$. Then, if  $f\neq 0$ then $\hat{f}\neq 0$ maintaining the non-zeroness of the entire polynomial $\hat{f}$. That is, $f= 0$ if and only if $\hat{f}=0$.
\end{lemma}
\begin{proof}
The proof is by induction on $D_2$. Assume $f$ is non-zero. 
\begin{itemize}
    \item \textbf{Base Case: \( D_2 = 1 \).}  
That is, the number of terms in each product is 1, and we have 
\[
f = \sum_{i \in [s]} Q_{i1}.
\]

For \( D_2 = 1 \), each \( Q_{i1} \) for \( i \in [s] \) can be computed by a depth-\((d-2)\) \( + \)-regular circuit. Consequently, the polynomial \( f \) can be computed by a depth-\((d-2)\) \( + \)-regular circuit of size at most \( s \). 
This is possible because the top gate of each circuit computing \( Q_{i1} \) is a sum gate (denoted as a \(\Sigma\)-gate), and the sum gates across all circuits can be combined into a single \(\Sigma\)-gate without increasing the depth of the circuit. Therefore, the overall circuit computing \( \sum_{i \in [s]} Q_{i1} \) retains a depth of \( d-2 \).

Using one of the hypotheses of the lemma: if any polynomial $P$, computed by a depth $(d-2)$  $+$-regular circuit, can be converted into a $\hat{P}$ (which is an ordered power-sum polynomial) while preserving non-zeroness, we conclude that 
the polynomial $f$ 
    can be transformed into another polynomial  $\hat{f}$ preserving non-zeroness. This completes the proof of the base case.\\
 \item {\bf Inductive Step: }
 \begin{itemize}
     \item {\bf Inductive Hypothesis:} We assume that the lemma holds for polynomials computed by a sum of products of at most $(D_2-1)$ $Q_{ij}$ polynomials. We will show it also holds for $D_2$.\\
 \end{itemize}
Consider the polynomial $f$ expressed as: $$f=\sum_{i \in [s]}Q_{i1}P_i,$$
 where $P_i=\prod_{j=2}^{D_2} Q_{ij}$.
We can expand $P_i$ into a sum of monomials, denoting the coefficient of a monomial $m$ in   $P_i$ by $[m]P_i$.  Thus, we can rewrite $f$ as:
 $$f=\sum_{i \in [s]}Q_{i1}\left(\sum_{m \in X^{D-D_1}}\left([m]P_i\right)\times m\right),$$
 where $D=D_1\times D_2$ is the degree of $f$ and $\left([m]P_i\right) \in \F$. Since $f\not \equiv 0$, there exists a monomial $m \in X^{D-D_1}$ such that the right derivative $f^{(m)}$ does not vanish. We define $f^{(m)}$ as:
$$f^{(m)}=\sum_{i \in [s]}Q_{i1}\times \left([m]P_i\right),$$ 
 where $\left([m]P_i\right) \in \F$. \newline

The polynomial $f^{(m)}\not\equiv 0$ can be computed by a depth $(d-2)$ +-regular circuit of size at most $s$.  This is because as each $Q_{i1}$ can be computed by a  depth $(d-2)$ +-circuit, then the sum $\sum_{i \in [s]}Q_{i1}\times \left([m]P_i\right)$ can also be computed by a depth-($d-2$) \(+\)-regular circuit. This is because the top gate of each circuit computing \( Q_{i1} \) is a sum gate (denoted as a \(\Sigma\)-gate), and the sum gates across all circuits can be combined into a single \(\Sigma\)-gate without increasing the depth of the circuit. Therefore, the depth of the overall circuit computing \( \sum_{i \in [s]}Q_{i1}\times \left([m]P_i\right) \) remains \(d-2\).

 This reduces to the base case.
Therefore, $f^{(m)}$ can be transformed into $$\hat{f}^{(m)}=\sum_{i \in [s]}\hat{Q}_{i1}\times \left ([m]P_i\right)$$ while preserving non-zeroness.
Since $\hat{f}^{(m)}$ is non-zero, the polynomial
$$f^{\dag}=\sum_{i \in [s]}\hat{Q}_{i1}P_i,$$ is also non-zero. Note that the derivative is only for analysis; we will not compute the right derivative of  $f$. \\
Note that the polynomial $\hat{Q}_{i1}$ is an ordered power-sum non-commutative polynomial over some new set of non-commutative variables. 
Now, expand the polynomial $\hat{Q}_{i1}$ as a sum of monomials in  $f^{\dag}$:
 $$f^{\dag}=\sum_{i \in [s]}\sum_{m \in \xi^{D_1}}\left([m]\hat{Q}_{i1}\right )m\times P_i.$$

Since $ f^{\dag}\not \equiv 0$, there exists a monomial $m$ such that the left derivative of $f^{\dag}$ with respect to $m$ does not vanish. Let $f^{\dag}_{(m)}$ be this polynomial. Again, note that this derivative is for analysis only, and we will not compute it. Let $\beta_i=[m]\hat{Q}_{i1}$. Note that $\beta_i\in \F$.

$$f^{\dag}_{(m)}=\sum_{i \in [s]}\beta_i\times P_i=\sum_{i \in [s]}\beta_i\times\prod_{j=2}^{D_2} Q_{ij}.$$

Clearly, $f^{\dag}_{(m)}\not\equiv 0$  and the number of product terms in the polynomial $P_i$  is exactly $(D_2-1)$. \\

By induction hypothesis, each $Q_{ij}$ in $f^{\dag}_{(m)}$ can be transformed into  $\hat{Q}_{ij}$ polynomial such that the resulting polynomial is non-zero.  
 Since $f^{\dag}_{(m)}\not\equiv 0$, the following polynomial is also non-zero:
$$\hat{f}=\sum_{i \in [s]}\prod_{j \in [D_2]} \hat{Q}_{ij}.$$ 
This completes the proof of the inductive step and therefore the lemma.
\end{itemize}
\end{proof}

\section{\\Proof of Lemma \ref{poly-proj} (Step 2: Product Sparsification Lemma)}
\label{app-proof}

We recall the following proposition from \cite{AJMR19} (see Proposition 3.1).  
\begin{proposition}[\cite{AJMR19}]
    Let $A : \mathbb{F}^n \to \mathbb{F}^n$ be an invertible linear transformation, and let $f(x_1, x_2, \dots, x_n) \in \mathbb{F} \langle X \rangle$ be any homogeneous polynomial of degree $d$. Let $A_j(f)$ be the polynomial obtained by applying the transform $A$ to the $j$-th position of monomials in $f$ for $j \in [d]$. Then, $f \neq 0$ if and only if $A_j(f) \neq 0$.
\end{proposition}

Here, we will present the proof of the \emph{product sparsification lemma}. This proof follows the proof of Lemma 6.2 in \cite{AJMR19}, but crucially uses Claim \ref{ord-pow-sum-c-nc}. Unlike \cite{AJMR19}, the key distinction in our setting is that the $\hat{Q}_{ij}$ polynomial can be non-homogeneous.

\begin{proof}
We use induction on $s$.

\begin{itemize}
    \item {\bf Base Case:} For $s=1$, we have $\hat{f}=\prod_{j \in [D_2]}\hat{Q}_{1j}$. By applying Claim \ref{ord-pow-sum-c-nc}, we can treat each n.c. polynomial $\hat{Q}_{1j}$ as commutative. Since the product of non-zero commutative polynomials is non-zero, we conclude that $J=\emptyset$ suffices.
    \item {\bf Induction hypothesis:} Assume the lemma holds for sums of products of \emph{$s$-ordered power-sum polynomials}  with at most $s-1$ terms. Recall
    $$\hat{f}_1=\sum_{i\in [s]} \prod_{j\in[D_2]}\hat{Q}_{ij},$$ where each $\hat{Q}_{ij}$ is an $s$-ordered power-sum polynomial of degree $D_1$ over $\xi=\{\xi_1,\xi_2,\ldots,\xi_s\}$.\\

    There are two directions to the proof. The backward direction is straightforward: by contrapositive, if we have a zero n.c. polynomial, then treating some of its variables as commuting does not yield a non-zero polynomial.

    Now, we prove the forward direction. 
    
It is important to note that the \( \hat{Q}_{ij} \) polynomials can, in general, be non-homogeneous (but in depth-5 case, each \( \hat{Q}_{ij} \) is a homogeneous polynomial). However, each product within the summand involves the same number of \( \hat{Q}_{ij} \) polynomials. Specifically, each summand is a product of exactly \( D_2 \) \( \hat{Q}_{ij} \) polynomials. Therefore, we will treat the \( \hat{Q}_{ij} \) polynomials as a whole, rather than expanding them as a sum of monomials or multiplying them with other \( \hat{Q}_{ij} \) polynomials in the product. This approach is valid because the boundaries of the \( \hat{Q}_{ij} \) polynomials can be efficiently identified, allowing us to select the entire \( \hat{Q}_{ij} \) polynomial as either non-commutative or commutative.

    Let $P_i=\prod_{j \in [D_2]} \hat{Q}_{ij}$ for all $i \in [s]$. Suppose $\sum_{i \in [s]}\beta_iP_i \neq 0$. We want to show there exists a subset $J \subseteq [D_2]$ of size at most $s-1$ such that  $\sum_{i \in [s]}\beta_iP_{i,J}\neq 0$, where $P_{i,J}=\left(\prod_{j\in \overline{J}} \hat{Q}_{ij}\right) \left(\prod_{j\in J} \hat{Q}_{ij}\right)$.\\

    Let $j_0 \in [D_2]$ be the smallest index such that $dim\{\hat{Q}_{1,j_0},\hat{Q}_{2,j_0},\cdots,\hat{Q}_{s,j_0}\}>1$. 
    If no such index $j_0$ exists, then all $P_i$ are scalar multiples of each other, leading to $\sum_{i \in [s]}\beta_iP_i =\alpha P_1$ for some non-zero $\alpha$, which reduces to the base case.\\

    Assume $j_0$ exists, we can renumber the polynomials such that 
    $\{\hat{Q}_{1,j_0},\hat{Q}_{2,j_0},\cdots,\hat{Q}_{t,j_0}\}$  are the only linearly independent polynomials at $j_0$,  where  $1<t \leq s$. Each polynomial  $P_i$     can be expressed as: 
  
\[
P_i = c_i P \cdot \hat{Q}_{i,j_0} \hat{Q}_{i,j_0+1} \cdots \hat{Q}_{i,D_2}, \quad i \in [t] \text{ and } c_i \in \mathbb{F}
\]
 \[
 P_i = c_i P \cdot \left( \sum_{\ell \in [t]} \alpha^{(i)}_{\ell} \hat{Q}_{\ell,j_0} \right) \hat{Q}_{i,j_0+1} \cdots \hat{Q}_{i,D_2}, \quad i \in [t+1,s] \text{ and } \alpha^{(i)}_{\ell} \in \mathbb{F}
 \]   

where $P$ is a product of $\Sigma\Pi^*\Sigma$ circuits.  Let $P'_i=c_i\prod^{D_2}_{j=j_0+1}\hat{Q}_{ij}, \text{ for } i \in [s].$\\

The polynomial $\hat{f}$ can be expressed as:
\[
\sum^{s}_{i=1}\beta_i P_i=P\times \left (\sum^{t}_{i=1}\beta_i \hat{Q}_{i,j_0}P'_i \right) + P\times \left (\sum^{s}_{i=t+1}\beta_i \hat{Q}_{i,j_0}P'_i \right).
\]
Note that:
\[
P\times\left( \sum^{s}_{i=t+1} \beta_i\hat{Q}_{i,j_0}P'_i\right )=P\times\left(\sum^{s}_{i=t+1} \beta_i\left (\sum_{\ell \in [t]}\alpha^{(i)}_{\ell}\hat{Q}_{\ell,j_0}\right)P'_i \right ).
\]

\begin{eqnarray*}
    \sum^{s}_{i=1}\beta_i P_i&=&P\times \left (\sum^{t}_{i=1}\beta_i \hat{Q}_{i,j_0}P'_i \right) + P\times\left(\sum^{s}_{i=t+1} \beta_i\left (\sum_{\ell \in [t]}\alpha^{(i)}_{\ell}\hat{Q}_{\ell,j_0}\right)P'_i \right )\\
    &=&   P \times \sum^{t}_{k=1}\hat{Q}_{k,j_0}P^{''}_k
\end{eqnarray*}
where $P^{''}_k=\beta_kP'_k+\beta_{t+1}\alpha^{(t+1)}_{k}P'_{t+1}+\beta_{t+2}\alpha^{(t+2)}_{k}P'_{t+2}+\cdots+\beta_{s}\alpha^{(s)}_{k}P'_{s}$, where $k\in[t]$.

Note that since $t>1$, each $P^{''}_k$ is a sum of at most $s-1$ polynomials and each of these polynomials is a product of linear forms.

Since we treat all polynomials $\hat{Q}_{k,j_0}, k\in[t]$ as n.c with each $\hat{Q}_{k,j_0}$ being \emph{$s$-ordered power-sum polynomials}, we can view them as linear forms as follows.  This is only for analysis. Let $$Y=\{y_{\overline{\ell}} \mid \ell=(\ell_1,\ldots,\ell_s) \text { and } \sum_i \ell_i\leq D\}$$ be the set of \emph{new} non-commuting variables, where $D$ is the degree of the given polynomial $\hat{f}_1$.  Each monomial $m_a$ with coefficient $\beta_{m_a}$ in $\hat{Q}_{k,j_0}$ can be expressed as $m_a=\xi_1^{a_1}\cdot \xi_2^{a_2}\cdots \xi_s^{a_s}$ and can represented by $\left(\beta_{m_a}\cdot y_{\overline{a}}\right )$.
  Thus,  we can envision  $\hat{Q}_{k,j_0}$ polynomial as a linear form over $Y$ variables, denoted by  $L_{k,j_0}$. Note that there is a bijection between monomials of $\hat{Q}_{k,j_0}$ and $L_{k,j_0}$ (which consists of only variables from  $Y$).\\

Since all polynomials $\hat{Q}_{k,j_0}$ for $ k\in[t]$ are linearly independent, it follows that the linear forms $L_{k,j_0}, k\in[t]$ are also linearly independent. \\

Let $A$ be an invertible linear transform such that $A:L_{k,j_0}\mapsto x_k$ for $k\in[t]$. The dimension of $A$ corresponds to the cardinality of $Y$, i.e., $|Y|$. Applying the map $A$ to the $j_0$-th factor of polynomial  $ \left(P \times \sum^{t}_{k=1}L_{k,j_0}P^{''}_{k}\right)$, we obtain:
$$R_{j_0}=\left(P \times \sum^{t}_{k=1}x_kP^{''}_{k}\right)$$

Since $\sum^{s}_{i=1}\beta_i P_i=\left(P \times \sum^{t}_{k=1}Q_{k,j_0}P^{''}_{k}\right)$ is non-zero, 
by applying Proposition 3.1 in \cite{AJMR19}, we conclude that $R_{j_0}\neq 0$. Consequently, there exists $k\in[t]$ such that $P^{''}_{k}\neq 0$.\\

Given $t>1$, $P^{''}_k$ is a sum of at most $s-1$ polynomials. By the induction hypothesis, there exists a subset $J'\subseteq [j_0+1, D_2]$ with size at most $s-2$ such that the resulting polynomial remains non-zero: $$P^{''}_{k,J'}\not=0$$ where

$$P^{''}_{k,J'}=\beta_kP'_{k,J'}+\beta_{t+1}\alpha^{(t+1)}_{k}P'_{t+1,J'}+\beta_{t+2}\alpha^{(t+2)}_{k}P'_{t+2,J'}+\cdots+\beta_{s}\alpha^{(s)}_{k}P'_{s,J'}.$$
 
 Let $J=\{j_0\}\cup J'$. We now show that the polynomial
 \[
 \sum^{s}_{i=1}\beta_i P_{i,J}=P^{(c)} \times \sum^{t}_{k=1}Q_{k,j_0}P^{''}_{k,J'}
 \]
remains non-zero, where $P^{(c)}$ is the commutative polynomial obtained by replacing $x_i$ by commuting variable $z_i$ in $P$. Since $P$ is a product of \emph{$s$-ordered power-sum polynomials} and by Claim \ref{ord-pow-sum-c-nc}, we can treat each of these \emph{$s$-ordered power-sum polynomials} as commutative, while preserving the non-zeroness of  $P$. 

Thus, it suffices to demonstrate that $\sum^{t}_{k=1}Q_{k,j_0}P^{''}_{k,J'}$ is non-zero. By Proposition 3.1 in \cite{AJMR19}, applying the linear transform $A$ to the first position of the polynomial $\sum^{t}_{k=1}Q_{k,j_0}P^{''}_{k,J'}$ yields $\sum^{t}_{k=1}x_kP^{''}_{k,J'}$. This sum is zero if and only if each $P^{''}_{k,J'}$ is zero. However, we established that there exists $k\in[t]$ such that $P^{''}_{k,J'}\neq 0$.\\

 This concludes the inductive step and completes the proof.
\end{itemize}
\end{proof}
 \subsection{Substitution Automaton for Product Sparsification (see \cref{step2-product-sparsification})}
 \label{automaton-product-sparsification}
 \begin{figure}[H]
    \centering
    \includegraphics[width=\linewidth]{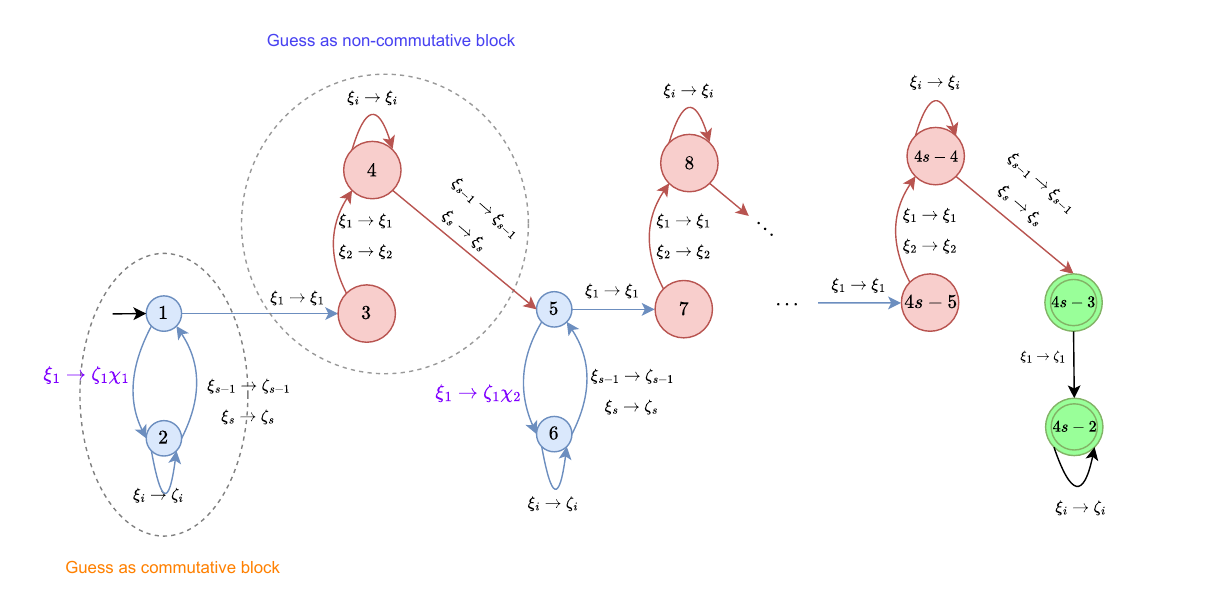}
    \caption{Substitution automaton for product sparsification}
    \label{aut-prod-spar}
\end{figure}

\subsection*{Description of the automaton in Figure \ref{aut-prod-spar}} 
Recall that the polynomial $\tilde{f}$ can be expressed as: $$\tilde{f}=\hat{f}_1+F_1.$$ where $$\hat{f}_1=\sum_{i\in [s]} \prod_{j\in[D_2]}\hat{Q}_{ij}.$$ 

According to product sparsification lemma \ref{poly-proj}, there exists an index set $I \subseteq [D_2]$ of size $s-1$ such that for each $i \in [s]$, considering only the $\hat{Q}_{ij}$ where $j \in [I]$, as non-commutative and the remaining $\hat{Q}_{ij}$  where $j \notin I$ as commutative preserves the non-zeroness of $\hat{f}$. Since $I$ is unknown, the automaton guesses the set $I$.

As previously mentioned (see Remark \ref{remark-product-of-e-pattern}), each monomial $m$ of the polynomial $\tilde{f}$ is a product of $\xi$-patterns (see Definition \ref{e-pattern}). This holds for both $\hat{f}_1$ and $F_1$. The automaton can identify the boundary of each 
$\xi$-pattern  because each $\xi$-pattern ends either with the variable
$\xi_{s-1}$ or $\xi_s$. 

For the structured polynomial $\hat{f}_1$, which is a sum of products of $s$-ordered power-sum polynomials, the boundary of each $\hat{Q}_{ij}$ can be identified. This allows us to explain the automaton's process of guessing the set $I$ in terms of the $\hat{Q}_{ij}$ polynomials.

The automaton applies its transformation to both   $\hat{f}_1$ and $F_1$.
For any given monomial  $m$ in $\tilde{f}$, the automaton classifies some of the sub-monomials ( i.e., $\xi$-patterns) as non-commutative and the rest as commutative. If the monomial
 $m$ belongs to  $\hat{f}_1$ it gives rise to set $\hat{Q}_{ij}$ polynomials that correspond to an index set $I$, that are considered non-commutative. 
 However, if  $m$ belongs to $F_1$, no such index set $I$ can be assigned, because the polynomial $F_1$ is unstructured. Nevertheless, for each monomial in the unstructured polynomial $F_1$, exactly $s-1$ $\xi$-patterns are considered non-commutative and the rest as commutative.
 We now explain the automaton's guessing process using the structured part $\hat{f}_1$.

In Figure \ref{automaton-product-sparsification}, the states labeled 
$4i+1$ (for $i\geq 0$) serve as guessing states. If a $\hat{Q}_{ij}$ polynomial is guessed as non-commutative by one of these guessing states,  it will be processed by one of the blocks marked as "Guess as non-commutative block". 
Similarly, if a $\hat{Q}_{ij}$ polynomial is guessed as commutative,  it will processed by one of the blocks marked as "Guess as commutative block".

It is crucial to note that if a $\hat{Q}_{ij}$ polynomial is guessed as non-commutative,  all its monomials are treated as non-commutative by replacing the non-commutative variable $\xi_i$  in this polynomial by the same non-commutative variable $\xi_i$. In particular, this $\hat{Q}_{ij}$ polynomial is fully treated as non-commutative before moving on to process the next factor,  $\hat{Q}_{i,j+1}$.

Similarly, if a $\hat{Q}_{ij}$ polynomial is guessed as commutative,  all its monomials are treated as commutative by replacing the non-commutative variable $\xi_i$ by the commutative variable $\zeta_i$. In particular, this $\hat{Q}_{ij}$ polynomial is fully treated as commutative before processing the next factor  $\hat{Q}_{i,j+1}$. 

When a  $\hat{Q}_{ij}$ polynomial is guessed as commutative,  the non-commutative variable $\xi_i$ appearing at the first position of each monomial in $\hat{Q}_{ij}$  is replaced by a commutative monomial $\zeta_i\chi_k$ where $k$  represents the count of previously guessed non-commutative $\hat{Q}_{i,k'}$ polynomials (for $k'<k$).The commutative $\chi_k$ variables are used to distinguish different guesses made by the automaton. 
If $\xi_i$ appears in positions other than the first, it will be replaced by the commutative variable $\zeta_i$. 

Given the structured nature of all $\hat{Q}_{ij}$ polynomials, the automaton can easily identify the boundaries of each $\hat{Q}_{ij}$, facilitating the guessing of the index set  $I$.

\section{Proof of Lemma \ref{com-conversion} (Step 3: Commutative Transformation)}
\label{proof:step3}
\begin{proof}
The degree of the n.c. polynomial $g$ is $s \times D$, where $D$  can be exponential in $s$. However, $g$ is in a more structured form as it is a sum of products of $s$ $s$-ordered power-sum polynomials.

To define the position of substrings within the monomial $m$, we consider $m$ as a string. Recall that $m$ can be thought of as a string since $m$ is n.c..

For the commutative transformation, we introduce $s^3$ fresh commutative variables  $W=\{w_{ij} \mid i \in[s^2] \text{ and } j \in[s]\}$. 
Each monomial $m=m_1.m_2\ldots m_s$ can be transformed as follows: for all $i \in [s]$, if  $m_i=\xi_1^{i_{1}}\xi_2^{i_{2}}\ldots\xi_s^{i_{s}}$, we convert it to  $m'_i=w_{(i-1)s+1,1}^{i_{1}}w_{(i-1)s+2,2}^{i_{2}}\ldots w_{(i-1)s+s,s}^{i_{s}}$.
Suppose the degree of $\xi_s$ in $m_i$ is 0 (i.e., $i_{s}=0$), then we 
skip the variable $w_{(i-1)s+s,s}$ and continue from $w_{is+1,1}$ for the next monomial $m_{i+1}$. This adjustment simplifies our automaton.

This conversion is implemented by the substitution automaton given in Figure \ref{fig2} and can be achieved using substitution matrices of dimension $s^2$.
Importantly, this process does not introduce new cancellations, as there is a bijection between the monomials in $g$ and $g^{(c)}$. 

Consider the mapping of monomials of $g$ into the monomials of $g^{(c)}$ by the substitution automaton in Figure \ref{fig2}. Note that every monomial of $g$  gets mapped into exactly one monomial in  $g^{(c)}$.  We  establish that this mapping is a bijection as follows:\\
\begin{itemize}
    \item {\bf One-to-One: }\\
Let $m$ and $m'$  be distinct monomials of $g$. Let  $m^{(c)}$ and $m'^{c}$ be the corresponding transformed commutative monomial (by the automaton). 
We can express these monomials as follows: $m=m_1.m_2\ldots m_s$ and $m'=m'_1.m'_2\ldots m'_s$.

We first convert the monomials $m$ and $m'$ as follows: for each $i\in[s]$, if $m_i$ does not include the variable $\xi_s$ (i.e., $i_s=0$), we append $\xi^0_s$ to $m_i$, 
 resulting in $m_i$ becoming $m_i\cdot\xi^0_s$. While $\xi^0_s=\epsilon$ (empty string), we retain it as a placeholder for clarity in our proof. We apply the same conversion to $m'$. We can refer to the modified monomials as $m$ and $m'$  again.

After this conversion, each monomial $m$ can be expressed as $m=m_1.m_2\ldots m_s$ 
where $m_i=\xi_1^{i_1}\xi_2^{i_2} \cdots \xi_s^{i_s}$ with $i_s\geq0$ and $i_j>0$ for $j\in[s-1]$. After these conversions, we still have $m\neq m'$.

Next, we consider any monomial $m$ of $g$ as a string over $\xi$ and break this string into segments by identifying maximal substrings containing the same variable
$\xi_i$ for $i\in[s]$. We encode the position of each segment within $m$ 
using the first index of the variables in $W$. It is important to note that in any monomial $m$ of $g$, there are at most $s^2$ such segments. In particular, following the conversion that adds $\xi^0_s$ when it is absent, there are exactly $s^2$  segments present.

 If $m\neq m'$, then there exists a segment $k \in[s^2]$ where the exponents of the variable in that segment differ between $m$ and $m'$. Assuming the variable in the $k$-th segment is $\xi_j$,  we conclude that the exponents of $w_{k,j}$ in $m^{(c)}$ and $m'^{c}$ are not the same. For any variable $w \in W$, we replace  $w^0$ by $1$ in the transformed commutative monomial.
 This shows that $m\neq m'$ implies $m^{(c)}\neq m'^{c}$.
\item {\bf Onto : }\\
Let $m^{(c)}$ be a monomial in $g^{(c)}\in \F[W]$. We first order the variables in $m^{(c)}$ according to the first index of the $W$ variables appearing in $m^{(c)}$.
Next, we break this into segments by identifying maximal substrings that contain  $W$ variables,  ensuring that the second subscript of these variables is non-decreasing, either from  $1$ to $s-1$ or from $1$ to $s$. Each segment can be represented in one of two forms:

(1) $w^{\ell_1}_{i,1}\cdot w^{\ell_2}_{i+1,2}\cdots w^{\ell_{s-1}}_{i+s-2,s-1}$ or \\
(2) $w^{\ell_1}_{i,1}\cdot w^{\ell_2}_{i+1,2}\cdots w^{\ell_s}_{i+s-1,s}$. 

We then convert each of these segments  into segments over the 
$\xi$ variables. For instance, if we take the segment $w^{\ell_1}_{i,1}\cdot w^{\ell_2}_{i+1,2}\cdots w^{\ell_{s-1}}_{i+s-2,s-1}$, we can convert it to $\xi^{\ell_1}_{1}\cdot \xi^{\ell_2}_{2}\cdots \xi^{\ell_{s-1}}_{s-1}$. 

This converts $m^{(c)}$ into a monomial $m$ over $\xi$ variables. It is clear that if we consider the output of the automaton in Figure \ref{fig2} on the monomial $m$, then it is the monomial  $m^{(c)}$. 
\end{itemize}

 Thus, the n.c. polynomial $g$ can be transformed into a commutative polynomial $g^{(c)}$ while preserving its non-zeroness. 
\end{proof}

\subsection{Substitution Automaton for Commutative Transformation}
\label{automaton-step3}
\begin{figure}[h]
\centering
\resizebox{\textwidth}{!}{ 
\begin{tikzpicture}
\node(pseudo) at (-1,0){}; 
\node(0) at (0,0)[shape=circle,draw, minimum size=1cm] {$q_0$}; 
\node(1) at (2.5,0)[shape=circle,draw, minimum size=1cm] {$q_1$}; 
\node(3) at (5,0)[shape=circle,draw, minimum size=1cm] {$q_{s-2}$}; 
\node(4) at (7.5,0)[shape=circle,draw, minimum size=1cm] {$q_{s-1}$};
\node(5) at (10.5,0)[shape=circle,draw, minimum size=1cm] {$q_s$};
\node(6) at (12.5,0)[shape=circle,draw, minimum size=1cm] {$q_{s^2-2}$};
\node(7) at (15.5,0)[shape=circle,draw,double,fill=green!40, minimum size=1cm] {$q_{s^2-1}$};

\path [->] 
    (0) edge node [above] {$\xi_2 \rightarrow w_{2,2}$} (1)
    (1) edge[dotted] node [above] {$\cdots$}  (3)
    (3) edge node [above] {$\xi_s \rightarrow w_{s,s}$} (4)
    (4) edge node [above] {$\xi_1 \rightarrow w_{s+1,1}$} (5)
    (5) edge[dotted] node [above] {$\cdots$} (6)
    (6) edge node [above] {$\xi_s \rightarrow w_{s^2,s}$} (7)

    (0) edge [loop above] node [above] {$\xi_1 \rightarrow w_{1,1}$} ()
    (1) edge [loop above] node [above] {$\xi_2 \rightarrow w_{2,2}$} ()
    (3) edge [loop above] node [above] {$\xi_{s-1} \rightarrow w_{s-1,s-1}$} ()
    (4) edge [loop above] node [above] {$\xi_s \rightarrow w_{s,s}$} ()
    (5) edge [loop above] node [above] {$\xi_1 \rightarrow w_{s+1,1}$} ()
    (6) edge [loop above] node [above] {$\xi_{s-1} \rightarrow w_{s^2-1,s-1}$} ()
    (7) edge [loop above] node [above] {$\xi_s \rightarrow w_{s^2,s}$} () 

    (pseudo) edge (0);

\draw (3) edge[->,bend right] node [below] {$\xi_1 \rightarrow w_{s+1,1}$} (5);
\draw (6) edge[->,bend right] node [below] {$\xi_{s-1} \rightarrow w_{s^2-1,s-1}$} (7);
\end{tikzpicture}
}

\caption{Automaton that converts n.c. variables to commutative variables}\label{fig2}
\end{figure}
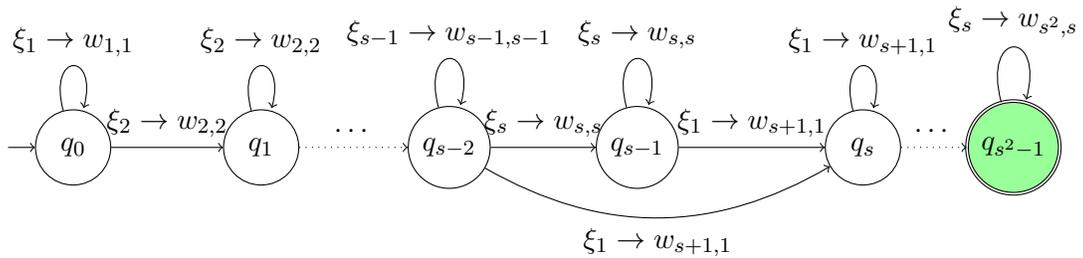

As an example, we provide the substitution matrix $M_{\xi_1}$ obtained from Figure \ref{automaton-step3} for the variable $\xi_1$.

\[\mathbf{M_{\xi_{1}}}= \begin{blockarray}{r*{8}{ >{}c}}
& q_0 & q_1 & \ldots & q_{s-2} & q_{s-1} & q_{s} & \ldots & q_{s^2-1}\\
\begin{block}{ >{\scriptstyle}r!{\,}(cccccccc)}
q_0 & w_{1,1} & 0 & \ldots & 0 & 0 & 0 & \ldots & 0 \\
q_1 & 0 & 0 & \ldots & 0 & 0 & 0 & \ldots & 0\\
\vdots & \vdots & \vdots & \ddots & \vdots & \vdots & \vdots &\ddots & \vdots \\
q_{s-2} & 0 & 0 & \ldots & 0 & 0 & w_{s+1,1} & \ldots & 0\\
q_{s-1} & 0 & 0 & \ldots & 0 & 0 & w_{s+1,1} & \ldots & 0\\
q_{s} & 0 & 0 & \ldots & 0 & 0 & w_{s+1,1} & \ldots & 0\\ 
\vdots & \vdots & \vdots & \ddots & \vdots & \vdots & \vdots & \ddots & \vdots\\
q_{s^2-1} & 0 & 0 & \ldots & 0 & 0 & 0 & \ldots & 0 \\
\end{block}
\end{blockarray}
\]

\section{ Proof of Lemma \ref{lem-modp-counting} (Coefficient Modification by Automaton)}
\label{proof:modp-counting}

 We recall the following proposition from  \cite{DESW11} (see Proposition 1).
\begin{proposition}
\label{prop-modp-automaton}
    Let $n \in \mathbb{N}$. Suppose we have two strings, $x$ and $y$ such that $|x|\neq|y|$ and $|x|,|y|\leq n$. There exists a DFA with a number of states bounded by $O(\log n)$ that can distinguish between the two strings. Specifically, when processing the strings $x$ and $y$  from the initial state, the DFA will reach different states for each string. This DFA essentially computes the string length modulo $p$.
\end{proposition}
This above proposition relies on a lemma from \cite{SB96} (also see Lemma 1 from  \cite{DESW11}).
 \begin{lemma}
 \label{prime-lemma}
Let $n \in \mathbb{N}$. If $k\not=m$, and both $k$ and $m$ are less than or equal to $n$, then there is a prime  $p \leq 4.4\log n$ such that $k \not\equiv m$ (mod $p$).
\end{lemma}
\subsection*{Proof of Lemma \ref{lem-modp-counting}:}
\begin{proof}
    Recall that $\tilde{f}=\hat{f}_1+F_1 \in\F\angle\xi$ is the polynomial obtained after Step (1).
It is important to note that any monomial in $\tilde{f}$, including all monomials in  $A^{\tilde{m}}_{{\hat{f}_1}}$ and $B_{F_1}^{\tilde{m}}$, is formed by concatenating $\xi$-patterns (see Defintion \ref{e-pattern}).

Consider a monomial $m_j=m_{j,1}m_{j,2} \cdots m_{j,N_j} \in B_{F_1}^{\tilde{m}}$
, where $m_{j,1},m_{j,2}, \ldots, m_{j,N_j}$ are $\xi$-patterns. While the degree of 
$m_j$ matches the overall degree $D=D_1\times D_2$ of the polynomial, we cannot assert whether $N_j$ is equal to $D_2$ or not (See Case 2 in \ref{two_cases}). However, we can affirm that there is at least one $\xi$-pattern in $m_j$ for which the sum of the exponents does not equal $D_1$.

For $B_{F_1}^{\tilde{m}}$, there are two possibilities:\\

In each of the two cases,  we construct an automaton $\mathcal{M}$ to evaluate the polynomial $\tilde{f}$ using substitution matrices obtained from the automaton.
Let $M_{\xi_i}$  denote the substitution matrix corresponding to the variable 
$\xi_i$ obtained from the automaton. We then evaluate the polynomial  $\tilde{f}$ 
on these matrices, leading to the resulting matrix: $ \tilde{f}(M_{\xi_1},M_{\xi_2},\ldots,M_{\xi_s})$. It is important to note that this can be expressed as: $$ \tilde{f}(M_{\xi_1},M_{\xi_2},\ldots,M_{\xi_s})=\hat{f}_1(M_{\xi_1},M_{\xi_2},\ldots,M_{\xi_s}) + F_1(M_{\xi_1},M_{\xi_2},\ldots,M_{\xi_s}).$$

\begin{itemize}
 \item {\bf Case 1:} There exists a monomial $m_l \in B_{F_1}^{\tilde{m}}$ such that  $N_l \neq D_2$. Let $\lambda \equiv D_2 \mod p$.\\
 
Let’s fix such a monomial $m_l \in B_{F_1}^{\tilde{m}}$ represented as $m_\ell=m_{\ell,1}m_{\ell,2} \cdots m_{\ell,N_\ell}$. As previously noted, an automaton can identify the boundaries between any two sub-monomials  $m_{\ell,i}$ and $m_{\ell,i+1}$ (for $i<N_\ell$) in the monomial $m_\ell$.
Therefore, for any natural number $p$ we can compute $N_\ell \mod p$ using an automaton.

Recall that for each monomial in $A^{\tilde{m}}_{{\hat{f}_1}}$, the number of sub-monomials (i.e., the number of \emph{$\xi$-patterns}) is exactly $D_2$. We can compute the number of \emph{$\xi$-patterns} modulo $p$ as boundaries between any two sub-monomials  $m_{\ell,i}$ and $m_{\ell,i+1}$ can be identified by the automaton.

Since $N_j,D_2\leq D$, by Proposition \ref{prop-modp-automaton}, there exists an automaton $\mathcal{M}$ with a number of states bounded by $O(\log D)$.  If we consider the output state of $\mathcal{M}$ when processing  $m_l$, it will differ from the output state of all monomials in $A^{\tilde{m}}_{{\hat{f}_1}}$.
The automaton $\mathcal{M}$ does not change the monomial;  it simply maps the n.c. variable $\xi_i$ to itself while computing the number of \emph{$\xi$-patterns} modulo $p$. Notably, all monomials in $A^{\tilde{m}}_{{\hat{f}_1}}$ reach the same state in the automaton, which we denote as $\lambda$ (with $\lambda=O(\log D)$).

By Proposition \ref{prop-modp-automaton}, at least one monomial $m_l \in B_{F_1}^{\tilde{m}}$  will not reach the state $\lambda$. 
If we consider the output of the automaton $\mathcal{M}$ in Figure \ref{xi-pattern consumption} on $\tilde{f}$ as $(q_0,q_\lambda)$-th entry of the resulting matrix, 
then the n.c. polynomial at this entry is given by 
$$f_\lambda=\hat{f}_1 + \tilde{F}_1$$ 
where the coefficient of the monomial $m_l$ in the polynomial $\tilde{F}_1$ is 0. 
Importantly, since $m_l\in B_{F_1}^{\tilde{m}}$ the coefficient of $m_l$ in $F$ is non-zero.

It is crucial to note that $f_\lambda\neq 0$ because $\hat{f}_1 \neq 0$ (by Lemma \ref{gen-proj}), and the non-zero monomials of $\tilde{F}_1$ are a subset of the non-zero monomials of $F_1$,  with no common non-zero monomials between $\hat{f}_1$ and $F_1$ (see Claims \ref{obs-good-mon} and \ref{obs-bad-mon}).

Additionally, while $\tilde{f}\neq f_\lambda$, all monomials of $A^{\tilde{m}}_{{\hat{f}_1}}$ are included in $f_\lambda$ with the same coefficients as in $\tilde{f}$. However, at least one monomial $m_l\in B_{F_1}^{\tilde{m}}$ is absent from $\tilde{F}_1$, thus from $f_\lambda$. Since only monomials in $A^{\tilde{m}}_{{\hat{f}_1}}\sqcup B_{F_1}^{\tilde{m}}$ are transformed into the commutative monomial $\tilde{m}$ during product sparsification and commutative transformation (Steps 2 and 3), performing these steps on $f_\lambda$ will yield a non-zero commutative polynomial, specifically ensuring that the coefficient of  $\tilde{m}$  is non-zero in the transformed commutative polynomial.\\

 \item {\bf Case 2:} For each monomial $m_\ell \in B_{F_1}^{\tilde{m}}$, $N_\ell=D_2$. Let $\lambda \equiv D_1 \mod p$.\\

Let's fix a monomial $m_\ell\in B_{F_1}^{\tilde{m}}$ expressed as $m_\ell=m_{\ell,1}m_{\ell,2} \cdots m_{\ell,D_2}$. For any monomial $m_\ell$ and $r\in[D_2]$, we define
 $$\mathcal{K}_{m_{\ell,r}}=\sum_{k=1}^{s} \ell_k.$$

We know that there exists an \emph{$\xi$-pattern}  $m_{\ell,r}$ in the monomial $m_\ell$, such that $\mathcal{K}_{m_{\ell,r}} \neq D_1$ by Claim \ref{obs-bad-mon}.
For all $m_t \in A^{\tilde{m}}_{{\hat{f}_1}}$, it follows that $\mathcal{K}_{m_{t,r}}=D_1$ for each $r \in [D_2]$ (as $N_t=D_2$) by Claim \ref{obs-good-mon}.

As previously noted, an automaton can identify the boundaries between any two sub-monomials  $m_{\ell,i}$ and $m_{\ell,i+1}$ (for $i<N_\ell$) in the monomial $m_\ell$.
Since we do not know which sub-monomial $m_{\ell,r}$ has the property $\mathcal{K}_{m_{\ell,r}} \neq D_1$, the automaton guesses $r \in [D_2]$ and computes $\left(\mathcal{K}_{m_{\ell,r}} \mod p\right)$ for a given natural number $p$. 

Given that  $\mathcal{K}_{m_{\ell,r}},D_1\leq D$, by Proposition \ref{prop-modp-automaton}, there exists an automaton $\mathcal{M}$ that computes $\left(\mathcal{K}_{m_{\ell,r}} \mod p\right)$. The automaton we construct has the number of states bounded by $O(s^2)$, that is $O(\log^2 D)$ as $D$ is exponential in $s$ for us. The output state of $\mathcal{M}$ when processing $m_l$ will differ from the output state of all monomials in $A^{\tilde{m}}_{{\hat{f}_1}}$.

As in Case 1, the automaton $\mathcal{M}$ does not modify the monomial; it simply computes the sum of the exponents of the $r$-th sub-monomial $m_{\ell,r}$ of the monomial $m_\ell$ modulo $p$. The substitution automaton $\mathcal{M}$ given in Figure \ref{RemainderNFA} does exactly this, with the initial state $q_0$ guessing the sub-monomial number $r\in[D_2]$ and the remainder of the automaton calculating $\left(\mathcal{K}_{m_{\ell,r}} \mod p\right)$.

For any monomial  $m_t \in A^{\tilde{m}}_{{\hat{f}_1}}$,  for every guess  $r\in[D_2]$,
 we have $\lambda \equiv \mathcal{K}_{m_{t,r}} \mod p$, where $\mathcal{K}_{m_{t,r}}=D_1$. Thus, when considering the output of the automaton $\mathcal{M}$ on the monomial $m_\ell$ as the sum of two entries $$M[q_0,q_{f_1}]+M[q_0,q_{f_2}],$$ 
 the coefficient $\alpha_{m_t}$ of the monomial $m_t$ is scaled by a factor of $D_2$.

 Conversely, for the above fixed monomial $m_\ell \in B_{F_1}^{\tilde{m}}$, there exists a guess $r\in[D_2]$ such that $\mathcal{K}_{m_{\ell,r}} \neq D_1$ (by Claim \ref{obs-bad-mon}). Therefore, this computation path will not reach either of the final states $q_{f_1}$ or $q_{f_2}$. Thus, the coefficient $\alpha_{m_\ell}$ of the monomial $m_\ell$ is scaled by at most a factor of $(D_2-1)$ in $M[q_0,q_{f_1}]+M[q_0,q_{f_2}]$.

Then the resulting output n.c. polynomial is given by $$f_\lambda=D_2\cdot\hat{f}_1 + \tilde{F}_1,$$ where the coefficient of the monomial $m_l$ in the polynomial $\tilde{F}$ is at most $\alpha_{m_\ell}\cdot (D_2-1)$. 

It is important to note that $f_\lambda\neq 0$. While $\tilde{f}\neq f_\lambda$, 
all monomials of $A^{\tilde{m}}_{{\hat{f}_1}}$ are present in $f_\lambda$ with 
their coefficients scaled by exactly $D_2$. However, at least one monomial $m_l\in B_{F_1}^{\tilde{m}}$  has its coefficient in $\tilde{f}$ scaled by at most $(D_2-1)$.

Since only monomials in $A^{\tilde{m}}_{{\hat{f}_1}}\sqcup B_{F_1}^{\tilde{m}}$ are transformed into the commutative monomial $\tilde{m}$ after product sparsification and commutative transformation (Steps 2 and 3), performing these steps on $f_\lambda$, will yield a non-zero commutative polynomial. In particular, the coefficient of $\tilde{m}$ is non-zero.

 As only monomials in $A^{\tilde{m}}_{{\hat{f}_1}}\sqcup B_{F_1}^{\tilde{m}}$
are transformed into the commutative monomial $\tilde{m}$ after product sparsification and commutative transformation (Steps 2 and 3),  if we carry out product sparsification and commutative transformation on $f_\lambda$, the resulting commutative polynomial is non-zero. In particular, the coefficient of $\tilde{m}$ is non-zero in the transformed commutative polynomial.
\end{itemize}
This completes the proof of the lemma.
\end{proof}

\subsection{Substitution Automaton for Case 1}
\begin{figure}[H]
  \includegraphics[width=\linewidth]{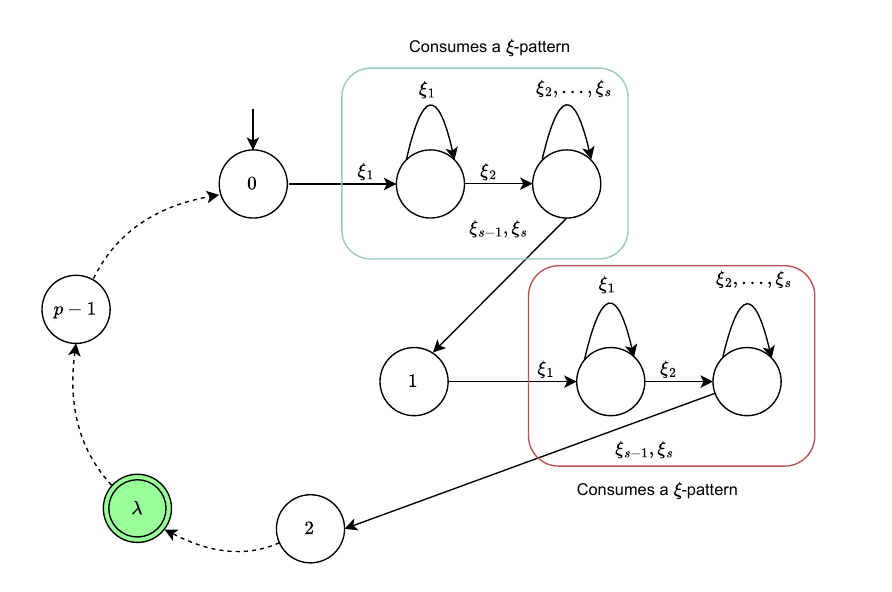}
  \caption{Substitution automaton that computes number of $\xi$-patterns modulo $p$.}
  \label{xi-pattern consumption}
\end{figure}
\subsection*{Description of the automaton in Figure \ref{xi-pattern consumption}}
Recall that the polynomial $\tilde{f}$ can be expressed as: $$\tilde{f}=\hat{f}_1+F_1.$$ 

Given that each monomial of $\tilde{f}$ is a product of $\xi$-patterns (see Definition \ref{e-pattern}),
 we can explain the automaton's workings in terms of these $\xi$-patterns. The automaton in Figure \ref{xi-pattern consumption} computes the number of $\xi$-patterns modulo $p$ in a given monomial $m$. Each $\xi$-pattern is fully processed by a section of the automaton marked as \emph{Consumes a $\xi$-pattern} before the next one is considered. Through this mechanism, the automaton keeps track of the number of $\xi$-patterns encountered, maintaining a count modulo $p$.
 
\subsection{Substitution Automaton for Case 2}

\begin{figure}[H]
  \includegraphics[width=\linewidth]{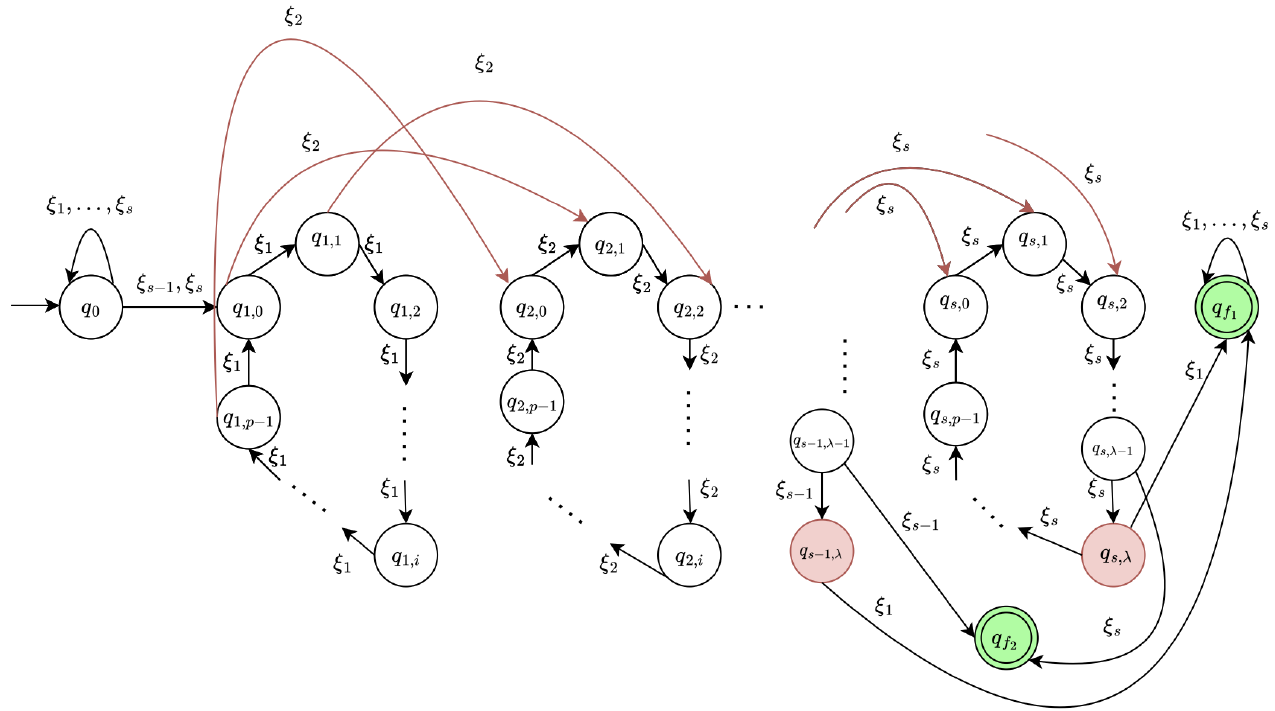}
  \caption{Modulo counting automaton that changes coefficients. Here, $\lambda\equiv D_1 \mod p$}
  \label{RemainderNFA}
\end{figure}
\subsection*{Description of the automaton in Figure \ref{RemainderNFA}}
Recall that the polynomial $\tilde{f}$ can be expressed as: $$\tilde{f}=\hat{f}_1+F_1.$$ 

Given that each monomial of $\tilde{f}$ is a product of $\xi$-patterns (see Definition \ref{e-pattern}),
 it is logical to explain the automaton's guesses in terms of these $\xi$-patterns. 

The substitution automaton in Figure \ref{RemainderNFA} guesses a $\xi$-pattern to compute the sum of the exponents of this guessed $\xi$-pattern modulo a prime number $p$. The initial state $q_0$ is the only guessing state in the automaton, where the automaton selects a $\xi$-pattern for calculating the sum of exponents mod $p$.  

Suppose the sum of the exponents for the guessed $\xi$-pattern is $D_1$. This guess will lead to either state  $q_{f_1}$ or state $q_{f_2}$ based on the following two conditions:

\begin{itemize}
    \item It reaches the final state $q_{f_2}$ if the guessed $\xi$-pattern is the last $\xi$-pattern appearing in the given monomial $m$ (the last variable could be either  $\xi_{s-1}$ or $\xi_{s}$, influencing the transition from either state $q_{s-1,\lambda-1}$ or $q_{s,\lambda}$) \emph{ or}
    \item It reaches the final state $q_{f_1}$ if the guessed $\xi$-pattern is not the last one, indicating that there exists at least one other $\xi$-pattern following it.
\end{itemize}
We use two different final states, $q_{f_1}$ and $q_{f_2}$,  to ensure that each $\xi$-pattern is fully processed before considering the next one.

Similarly, suppose the sum of the exponents for the guessed $\xi$-pattern is denoted by $\mathcal{K}$. In that case, the automaton will transition to the states based on the value $\mathcal{K}\mod p$ and the variable appearing in the last position (either $\xi_{s-1}$ or $\xi_{s}$).

The correct guess happens on the variable occurring at the last position of the $\xi$-pattern, which can be either $\xi_{s-1}$ or $\xi_{s}$.
In this case, the automaton will transition to the state $q_{1,0}$. Since the first variable of any $\xi$-pattern is $\xi_1$, and upon making a correct guess, the automaton continues to process the guessed $\xi$-pattern further.

However, if the guess is incorrect--- meaning the automaton transitions to the state $q_{1,0}$ while still processing the current $\xi$-pattern--- it indicates that additional variables ($\xi_{s-1}$ or $\xi_{s}$) still needs to be processed. Once the automaton reaches the state $q_{1,0}$, there will be no transition for the variables $\xi_{s-1}$ or $\xi_{s}$, resulting in the discarding of this particular guess.

\section{\\Proof of Lemma \ref{composing-substitution-mat} (Composing $\mathrm{K}$ Substitution Matrices)} \label{Appendix-B}
\begin{proof}
 The proof is by induction on $\mathrm{K}$. 
 \begin{itemize}
     \item \textbf{Base Case: $\mathrm{K}=1$} \\We have $f_1=f_0(A_{11},\ldots,A_{1n_1})=f(A_{11},\ldots,A_{1n})$. The lemma holds for $C_i=A_{1i}, i \in [n]$.

     \item \textbf{Base Case: $\mathrm{K}=2$} \\Let $f=\sum_{m \in X^D} \alpha_m.m.$  Each matrix $A_{ij}$ can be expressed as $$A_{ij}=\sum_{k=1}^{n_{i+1}} A^{(k)}_{ij} z_{i+1,k},$$ where $A^{(k)}_{ij} \in\F^{d_i\times d_i} $  for all $k \in [n_i+1]$.\\

     Consider a non-zero monomial  $m$ defined as $m=x_{\ell_1}x_{\ell_2}\ldots x_{\ell_D}$, where $\ell_i \in [n]$ for all $i \in [D]$. If we replace each variable $x_{\ell_i}$ by $A_{1\ell_i}$ in $m$, let $f_{1,m}$ denote the $(1,d_1)$ entry of the matrix product  $\prod_{j=1}^DA_{1\ell_j}$. We note that $f_{1,m}\in \F\angle {Z_2}$.\\

     Next, we substitute $z_{2i}$ with $A_{2i}$ in $f_{1,m}$, yielding $f_{2,m}$ as the $(1,d_2)$ entry of the matrix $f_{1,m}(A_{21},A_{22},\ldots,A_{2n_2})$.\\

     Now, consider  $a=(a_1,\ldots,a_D)\in [n_2]^D$ and define the n.c. monomial $m_{a,Z_2}=\prod_{i=1}^Dz_{2,a_i}$. We can express:

\begin{eqnarray*}
  \prod_{i=1}^D  A_{1\ell_i} & = & \prod_{i=1}^D\left(\sum_{k=1}^{n_2} A^{(k)}_{1\ell_i}z_{2k}\right)\\
  &=&\sum_{(a_1,\ldots,a_D)\in [n_2]^D} \left(\prod_{i=1}^DA^{(a_i)}_{1\ell_i}\right)\times m_{a,Z_2}
\end{eqnarray*}
This factorization holds because $m_{a,Z_2}$ is generated only by the terms 
$\prod_{i=1}^D\big(A^{(a_i)}_{1\ell_i}z_{2a_i}\big)$, where $A^{(a_i)}_{1\ell_i} \in \F^{d_1 \times d_1}$ for all $i \in [D]$.

Let $\alpha_{1,d_1}$ be the $(1,d_1)$-th entry of the matrix $\prod_{i=1}^DA^{(a_i)}_{1\ell_i} $. This $\alpha_{1,d_1}\in \F$ represents the coefficient of the monomial $m_{a,Z_2}$ in the $(1,d_1)$  entry of $\prod_{i=1}^DA_{1\ell_i}$.

Recall that $A_{1j}=\sum_{k=1}^{n_2} A^{(k)}_{1j} z_{2k}$ for all $j \in [n_1]$. For each $i \in [n]$, define: $$C_i=\sum_{k=1}^{n_2} A^{(k)}_{1i} \scalebox{1.2}{$\otimes$} A_{2k}.$$

Thus, $C_i$ is a matrix of dimension $(d_1\cdot d_2)$.\\

We aim to show that $f_2$  corresponds to the $(1,d_1\times d_2)$-th entry of the matrix obtained by substituting $x_{\ell_i}$ with $C_{\ell_i}$ for all $i \in [D]$. Substituting $x_{\ell_i}$ by $C_{\ell_i}$ for all $i \in [D]$ in $m$, we have
\[
C_{\ell_1} C_{\ell_2} \ldots C_{\ell_D} = \prod_{i=1}^D \left( \sum_{k=1}^{n_2} A^{(k)}_{1\ell_i} \scalebox{1.2}{$\otimes$} A_{2k} \right).
\]

Expanding this, we have:
\[
= \sum_{(a_1,\ldots,a_D) \in [n_2]^D} \left( \prod_{i=1}^D A^{(a_i)}_{1\ell_i} \scalebox{1.2}{$\otimes$} A_{2a_i} \right).
\]

By the mixed product property of \( \scalebox{1.2}{$\otimes$} \) and matrix multiplication (See Lemma \ref{mixed_product_property-lem}), we have:
\[
= \sum_{(a_1,\ldots,a_D) \in [n_2]^D} \left( \left( \prod_{i=1}^D A^{(a_i)}_{1\ell_i} \right) \scalebox{1.2}{$\otimes$} \left( \prod_{i=1}^D A_{2a_i} \right) \right).
\]

Let $M_{a,Z_2}=\prod_{i=1}^DA_{2a_i}$ and $\beta_{1,d_2}$ be the  $(1,d_2)$-th entry of the matrix $M_{a,Z_2}$. This $\beta_{1,d_2}\in\F$ corresponds to evaluating the monomial $m_{a,Z_1}$ using matrices $A_2=(A_{21},A_{22},\ldots,A_{2n_2})$ by replacing $z_{2j}$ with $A_{2j}$.\\

The $(1,d_1.d_2)$-th entry of the matrix $\big(\prod_{i=1}^DA^{(a_i)} _{1\ell_i}\big) \scalebox{1.2}{$\otimes$} \big( \prod_{i=1}^DA_{2a_i}\big)$ is equal to $(1,d_2)$-th entry of $M_{a,Z_2}$ scaled by  $\alpha_{1,d_1}$. Summing over all $a\in [n_2]^D$, the $(1,d_1\cdot d_2)$-th entry of this matrix $\prod_{i=1}^DC_{\ell_i}$ equals $f_{2,m}$.\\

By linearity, this extends to $f(C_1,C_2,\ldots,C_n)$, showing that the $(1,d_1\cdot d_2)$-th entry of the matrix $f(C_1,C_2,\ldots,C_n)$ is equal to the polynomial $f_2$. \\

This completes the base case, where $\mathrm{K}=2$.
\item  \textbf{Inductive Step:} \\
{\bf Induction Hypothesis: } Assume that for $\mathrm{K}-1$, we can express $f$ evaluated at matrices  $A_i=(A_{i1},\ldots,A_{in_i})$ as 

$$f_i=f_{i-1}(A_{i1},A_{i2},\ldots,A_{in_{i}})$$

for all $i \geq 1$. Let $f_{\mathrm{K}-1}$ be the resulting polynomial over the n.c. variables $Z_{\mathrm{K}}=\{z_{\mathrm{K}1},\ldots,z_{\mathrm{K}n_{\mathrm{K}}}\}$.

For the induction hypothesis, we assume that there is a matrix substitution 
$B=(B_{1},B_{2},\ldots,B_{n})$ such that each $B_i$ has dimensions $\prod_{i\in[\mathrm{K}-1]}d_i$ and the polynomial $f_{\mathrm{K}-1}$ is equal to the $(1,\prod_{i\in[\mathrm{K}-1]}d_i)$-th entry of the matrix  $f(B_{1},B_{2},\ldots,B_{n})$. Each matrix $B_i$ can be written as: $$B_i=\sum_{j=1}^{n_\mathrm{K}} B^{(j)}_{i} z_{\mathrm{K},j}. $$

By the inductive hypothesis, we find that $f_{\mathrm{K}-1}$ is equivalent to the $(1,\prod_{i\in[\mathrm{K}-1]}d_i)$-th entry of the matrix  $f(B_{1},B_{2},\ldots,B_{n})$.  Define: $$f_\mathrm{K}=f_{\mathrm{K}-1}(A_{\mathrm{K}1},\ldots,A_{\mathrm{K}n_\mathrm{K}})[1,d_{\mathrm{K}}].$$
This reduces to the base case.

Consequently, there exists a matrix substitution $C=(C_{1},C_{2},\ldots,C_{n})$, where each $C_i$ has dimensions $\prod_{i\in[\mathrm{K}]}d_i$ and is given by:
$$C_i=\sum_{j=1}^{n_\mathrm{K}} B^{(j)}_{i} \scalebox{1.2}{$\otimes$} A_{\mathrm{K},j}.$$

Thus, the polynomial $f_\mathrm{K}$ matches the  $(1,\prod_{i\in[\mathrm{K}]}d_i)$-th entry of the matrix $f(C_{1},C_{2},\ldots,C_{n})$. This concludes the proof of the lemma.
\end{itemize}
\end{proof}

\section{Missing Proofs from \cref{step1}}
\label{step1-proofs}
\subsection{Proof of Proposition \ref{m-rho-form}}
\begin{proof}
    The path $\rho$ starts at $q_0$ and ends at the state $q_{s-1}$, labeled by the monomial $m$. When this path $\rho$ returns to $q_0$ for the $k$-th time, it has converted some initial part of the monomial $m$ into $\alpha_k m'_1\cdot m'_2\cdots m'_{k}$, where $\alpha_k\in\F[Y \sqcup Z]$. The number of sub-monomials, denoted by $N$, depends on the number of times the path $\rho$ returns to the initial state $q_0$ before eventually reaching either $q_{s-1}$ or $q_0$. We can see from the automaton given in Figure \ref{fig1} that each $m'_\ell$  is of the form $m'_\ell=\xi^{\ell_1}_1.\xi^{\ell_2}_2\cdots\xi^{\ell_s}_s$, where $\ell_k>0,k\in[s-1]$ and $\ell_s\geq 0$. The commutative part is grouped into the coefficient $\alpha_k$. Note that $\ell_k>0,k\in[s-1]$. This is because to return to $q_0$, the path $\rho$ has to use transitions involving $\xi_1,\ldots,\xi_{s-1}$ variables. But the exponent $\ell_s\geq 0$, because to return to $q_0$ the path $\rho$ may use the transition from $q_{s-2}$ to $q_0$ instead of going to the state $q_{s-1}$, and $\xi_s$ variable only appears at the transition from $q_{s-1}$ to itself. This completes the proof.
\end{proof}
\subsection{Proof of Claim \ref{obs-good-mon}}
\begin{proof}
The path $\rho$ respects the boundary between $m_j$ and $m_{j+1}$ in $m$ for all $j<D_2$.
   The proof follows by noting that each transition of the substitution automaton $\mathcal{A}$ has exactly one $\xi$ variable in it and each transition of the automaton consists of reading a variable appearing at some position in $m$ and substituting it according to transition rules of $\mathcal{A}$.  
\end{proof}

\subsection{Proof of Claim \ref{obs-bad-mon}}
\begin{proof}
    We will consider the following two possibilities for $\rho$:
    \begin{itemize}
        \item {\bf Case 1:} Suppose the path $\rho$ visits the state $q_0$ in the middle of processing some sub-monomial $m_j$ of $m$ (that is when the automaton visits the state $q_0$ when it reads some variable appearing in a position other than the first position).  At this point, the length of the prefix of $m$ that has been processed is not a multiple of $D_1$. Additionally, at this stage, the sum of exponents of $\xi$ variables in at least one of the generated $m'_{\ell}$ up to this point is not equal to $D_1$.
        \item {\bf Case 2:} Suppose $\rho$ is in a state $q_r, r\neq 0$, while replacing the variable appearing at the 1st position of the sub-monomial $m_j$. That is $\rho$ starts replacing the sub-monomial $m_j$ from a state $q_r$ that is not $q_0$. If $\rho$ reaches $q_0$ again in the middle of $m_j$, then this scenario is already handled by Case 1. Otherwise, it fully replaces $m_j$ before returning to the state $q_0$. 
Note that since $\rho$ started processing the sub-monomial $m_j$ from a state $q_r\neq q_0$ and it fully replaces at least $m_j$ before returning $q_0$, we can make the following observation.
At this point, the sum of exponents of $\xi$ variables in the sub-monomial $m'_{\ell}$, which is generated by $\rho$ after fully processing $m_j$  and returning to $q_0$, is strictly greater than $D_1$. This is because when $\rho$ is in state  $q_r$, the path $\rho$ has recently transitioned from $q_0$ while processing some earlier sub-monomial $m_{j'}$ where $j'<j$.
Between the time $\rho$ returning to $q_0$ after starting from $q_0$ for $m_{j'}$ and the completion of $m_{j}$,
 at least all sub-monomials $m_{j'},m_{j'+1},\ldots,m_j$ of the monomial $m$ are processed to produce a single $m'_\ell$ in $m_\rho$. Thus, the sum of exponents of $\xi$ variables in the sub-monomial $m'_{\ell}$ is strictly greater than $D_1$ and in particular we can say that it is $c\times D_1$ where $c>1$.  
    \end{itemize}
This completes the proof.
\end{proof}

\subsection{Proof of Claim \ref{f-hat-form}}
\label{proof-f-hat-form}
 \begin{proof}
     Recall that $f=\sum_{i \in [s]}\prod_{j \in [D_2]} Q_{i,j}$. For ease of notation, we assume that the position of linear forms within each $Q_{i,j}$ polynomial starts from 1. For a monomial $m=m_1\cdot m_2\cdots m_{D_2}$ with coefficient $\alpha_m$, the polynomial $\hat{f}_{\alpha_m\cdot m}$ is the sum of all monomials obtained from computation paths $\rho$ labeled by $m$ from Case 1:

      \begin{eqnarray*}
     \hat{f}_{\alpha_m\cdot m}&=&\alpha_m\cdot\hat{f}_{m}\\
     &=& \alpha_m\times \sum\limits_{\rho:q_0 \overset{m}\leadsto q_{s-1}} m_{\rho}.
     \end{eqnarray*}
     
     For each $j<D_2$,  the boundary between $m_j$ and $m_{j+1}$ in $m$ is respected in all such computation paths $\rho$. In particular, for each path $\rho$ and sub-monomial $m_j$, we can associate a set $I_{\rho,j}\subseteq [D_1]$ such that $r \in I_{\rho,j}$, if the path $\rho$ makes transition from $q_{k-1}$ to $q_{k}$ (for some $k \in [s-1]$) or from $q_{s-2}$ to $q_0$ while reading  a variable at position $r$ in $m_j$. The transition is unique for given $\rho$ and $r \in I_{\rho,j} $.  Since computational path $\rho$ respects all boundaries of the sub-monomials, it makes exactly $s-1$ transitions and returns to $q_0$ when starting to process $m_{j+1}$. Thus, $|I_{\rho,j}|=s-1$.
     
      Let $I_{\rho,j}=\{\ell_1,\ell_2,\cdots,\ell_{s-1}\}$, where $\xi_{I_{\rho,j}}=\xi^{\ell_1}_1.\xi^{\ell_2-\ell_1}_2\cdots\xi^{D-\ell_{s-1}}_{s}$ ($\ell_1<\ell_2<\cdots<\ell_{s-1}$). From the automaton in Figure \ref{fig1}, it follows that $I_{\rho,j}$ can be any subset of size $s-1$.
 For a sub-monomial $m_j$ of $m$, let $m'_{j,\rho}$ be the transformed sub-monomial that includes both commutative variables $Y\sqcup Z$ and n.c. variables $\xi$. 
Define $c_{j,\rho}$ as the commutative part of $m'_{j,\rho}$. Then $m'_{j,\rho}=c_{j,\rho} \cdot \xi_{I_{\rho,j}}$. 

Given $\rho$, we can define the sequence of sets $I_{\rho,j}\subseteq [D_1]$ for $j\in[D_2]$ of size $s-1$ and conversely for a sequence of sets $I_{j}\subseteq [D_1]$ for $j\in[D_2]$ of size $s-1$, we can define the path $\rho$.

 We can then express  $\hat{f}_m$  as: 
 
       \begin{eqnarray*}
        \hat{f}_m&=&\sum\limits_{\rho:q_0 \overset{m}\leadsto q_{s-1}} m_{\rho}\\
        &=&\sum\limits_{\rho:q_0 \overset{m}\leadsto q_{s-1} } \left(\prod_j m'_{j,\rho}\right) \\
        &=&\sum\limits_{\rho:q_0 \overset{m}\leadsto q_{s-1} } \left( \prod_j c_{j,\rho} \cdot\xi_{I_{\rho,j}}\right) \\
        &=&\prod_{j=1}^{D_2-1}\left(\sum\limits_{\rho:q_0 \overset{m_j}\leadsto q_0 } c_{j,\rho}\cdot\xi_{I_{\rho,j}} \right) \times \left(\sum\limits_{\rho:q_0 \overset{m_{D_2}}\leadsto q_{s-1} } c_{j,\rho}\cdot\xi_{I_{\rho,j}} \right)\\
        &=&\prod_{j=1}^{D_2}\left(\sum_{\substack{I \subseteq [D_1], |I|=s-1} }c_{j,I}\xi_{I}\right)
       \end{eqnarray*}
Here, $c_{j,\rho}$ and $c_{j,I}$  are commutative monomials over $Y\sqcup Z$  defined by the path $\rho$ and the set $I$, respectively.
Recall $m=m_1\cdot m_2\cdots m_{D_2}$. The second last equality holds because the computational path $\rho:q_0 \overset{m}\leadsto q_{s-1}$ can be decomposed into $\rho:q_0 \overset{m_1}\leadsto q_0\overset{m_2}\leadsto q_0\overset{m_3}\leadsto\cdots \overset{m_{D_2-1}}\leadsto q_0\overset{m_{D_2}}\leadsto q_{s-1} $ and each sub-path is independent of other sub-path. Since we sum over all possible computational paths from $q_0$ to $q_{s-1}$ that is labeled by the monomial $m$, this equality holds.

The last two equalities hold because, $I_{\rho,j}$ can be any subset of size  $s-1$ and is independent of $I_{\rho,k}$ for $k\neq j$ (where $k\leq D_2$). 

Recall that $Mon(f)$ is the set of all monomials computed by the depth-5 $+$-regular circuit for $f$. Since coefficients of sub-monomials do not change due to the automaton's operations, we can ignore them in the notation for ease.
Therefore, we have:
 \begin{eqnarray*}
\hat{f}&=&\sum_{\substack{m \in Mon(f)}}\hat{f}_m \\
&=&\sum_{\substack{m \in Mon(f)}}\prod_{j=1}^{D_2}\left(\sum_{\substack{I \subseteq [D_1], |I|=s-1} }c_{j,I}\xi_{I}\right)\\
&=&\sum_{i\in [s]}\prod_{j\in[D_2]} \left (\sum_{\substack{I \subseteq [D_1], |I|=s-1} }Q_{i,j,I}\times \xi_I \right )
   \end{eqnarray*}  

Let $Q_{i,j,I}$ be the set of all monomials in  
$\left(\sum_{\substack{I \subseteq [D_1], |I|=s-1} }c_{j,I}\xi_{I}\right)$
that has the same non-commutative part $\xi_I$.
The last equality holds because if $m\in Mon(f)$, it can be written as $m=m_1\cdot m_2\cdots m_{D_2}$ and for some  $i \in[s]$ such that each sub-monomial $m_j \in X^{D_1}$, has a non-zero coefficient in $Q_{ij}$ and it can be processed by the substitution automaton independent of other sub-monomials.
Since we sum over all possible computational paths for each monomial in $m \in Mon(f)$, the equality holds.
Also suppose $m$ is computed by $\prod_{j\in[D_2]} Q_{i,j}$ for some $i\in[s]$ then by definition $m\in Mon(f)$. This completes the proof.
 \end{proof}

\subsection{Proof of Claim \ref{f-hat-spurious}}
\begin{proof}
Let $m\in X^D$ be a monomial computed by the depth-5 circuit that computes $f$. For some $i\in[s]$, $m$ has a non-zero coefficient in $\prod_jQ_{ij}$. The output of the substitution automaton $\mathcal{A}$ on $m$ can be expressed as $\hat{f}_m+F_m$. 
Consequently, the output of the automaton $\mathcal{A}$ on the polynomial $f$ is given by:  
  $$f'= \sum\limits_{m \in Mon(f)}(\hat{f}_m +F_m)$$ 
  which is the polynomial $f'= \hat{f} +F$.
 \end{proof}

\section{Additional Proofs and Automaton} 

\subsection{Proof of Claim \ref{ord-pow-sum-c-nc}}
\label{small_proofs}
\begin{proof}
Since $g$ is an ordered power-sum polynomial, we can treat the n.c. variables $\xi=\{\xi_1,\xi_2,\ldots,\xi_s\}$ as commutative variables without introducing any new cancellations. This is because the exponents of the variables $\xi=\{\xi_1,\xi_2,\ldots,\xi_s\}$ are different for any two distinct monomials of $g$. Thus, $g$ as a n.c. polynomial is non-zero if and only if $g$ as a commutative polynomial is non-zero. 
\end{proof}

\subsection{Substitution Automaton for Small Degree Case (see \cref{small-degree-case})}
\begin{figure}[H]
\begin{center}
\begin{tikzpicture}
\node(pseudo) at (-1,0){}; 
\node(0) at (0,0)[shape=circle,draw,double,fill=green!40, minimum size=1cm]
     {$q_0$}; 
     \node(1) at (2,0)[shape=circle,draw,minimum size=1cm] {$q_1$};
     \node(2) at (4,0)[shape=circle,draw, minimum size=1cm] {$q_2$}; 
     \node(3) at (6,0)[shape=circle,draw, minimum size=1cm] {$q_{c-1}$}; 
     (pseudo) edge
     
     \path [->] (0) 
     edge node [above] {$x_{1j}$} (1) (1) 
     edge node [above] {$x_{2j}$} (2) (2)
     edge[dotted] node [above] {$\cdots$} (3)

     (3)      edge [bend left=30]  node [below]  {$x_{cj}$}     (0)

     (pseudo) edge (0);

\end{tikzpicture}
\caption{The transition diagram for the variable $x_j : 1\leq j \leq n$}\label{automaton-small-degree}
\end{center}
\end{figure}
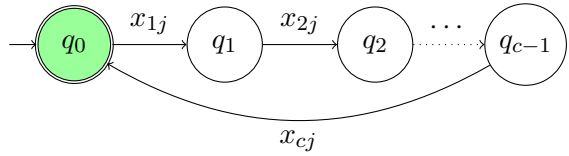 

\begin{displaymath}
\mathbf{M_{x_j}} =
\left( \begin{array}{cccccc}
0 & x_{1j} & 0 & \ldots & 0\\
0 & 0 & x_{2j} &  \ldots & 0\\
\vdots & \vdots & \vdots & \ddots &\vdots \\
0 & 0 & 0 & \ldots & x_{c-1,j}\\
x_{c,j} & 0 & 0 & \ldots & 0 
\end{array} \right)
\end{displaymath}

\bibliographystyle{alpha}

\bibliography{refs}

 \end{document}